%% file: main.tex
\documentclass[11pt]{article}

\usepackage[a4paper,margin=1in]{geometry}
\usepackage{graphicx}
\usepackage{style}
\usepackage{subfiles}
\usepackage{subcaption}
\usepackage{xspace}
\usepackage{enumitem}
\usepackage{ifthen}
\usepackage{tikz}
\usepackage{float}
\usepackage[section]{placeins}
\usepackage{hyperref}

\usetikzlibrary{calc}
\usetikzlibrary{fit,backgrounds}
\usetikzlibrary{arrows.meta,positioning,fit,calc,patterns}

\input{tikz-hypergraph}

\input{macros}

\newtheorem{claim}{Claim}

\title{Counting HyperGraphlets via Color Coding:\\
A Quadratic Barrier and How to Break It}

\author{
Marco Bressan \\
University of Milan \\
\texttt{marco.bressan@unimi.it}
\and
Stefano Clemente \\
University of Milan \\
\texttt{stefano.clemente@unimi.it}
\and
Giacomo Fumagalli \\
University of Milan \\
\texttt{giacomo.fumagalli@unimi.it}
}

\date{} 

\begin{document}

\maketitle

\begin{abstract}
We study the problem of counting $k$-\emph{hyper}graphlets, an interesting but surprisingly ignored primitive, with the aim of understanding if efficient algorithms exist.
To this end we consider \emph{color coding}, a well-known technique for approximately counting $k$-graphlets in graphs.
Our first result is that, on hypergraphs, color coding encounters a \emph{quadratic barrier}: under the Orthogonal Vector Conjecture, no implementation of it can run in time sub-quadratic in the size of the input.
We then introduce a simple property, $(\alpha,\beta)$-niceness, that hypergraphs from real-world datasets appear to satisfy for small values of $\alpha$ and $\beta$.
Intuitively, an $(\alpha,\beta)$-nice hypergraph can be split into two sub-hypergraphs having respectively rank at most $\alpha$ and degree at most $\beta$.
By applying different techniques to each sub-hypergraph and carefully combining the outputs, we show how to run color coding in time $2^{O(k)} \cdot \big(2^\beta |V| + \alpha^k |E| + \alpha^2 \beta \size{H}\big)$, where $H=(V,E)$ is the input hypergraph.
Afterwards, we can sample colorful $k$-hypergraphlets uniformly in expected $k^{O(k)} \cdot (\beta^2+\ln |V|)$ time per sample.
Experiments on real-world hypergraphs show that our algorithm neatly outperforms the naive quadratic algorithm, sometimes by more than an order of magnitude.
\end{abstract}

\noindent\textbf{Code availability:} \\
\url{https://github.com/gfumagalli9/hyper-motivo/tree/hyper-motivo}

\bigskip

\section{Introduction} 
\subfile{1_introduction}

\section{Preliminaries}\label{sec:prelims}
\subfile{2_basic_definition_and_properties}

\section{Color Coding and its Quadratic Barrier}\label{sec:cc}
\subfile{3_algorithm}

\section{The $(\alpha,\beta)$-Niceness of a Hypergraph}
\subfile{4_alphabeta_algorithm}

\section{Experiments}\label{sec:experiments}
\subfile{5_Experiments}

\section{Conclusions and future work}
We have shown that counting $k$-hypergraphlets is subtly harder than counting $k$-graphlets, but equally efficient algorithms exist under mild structural assumptions on the input hypergraphs.
Our work is among the very first to investigate this problem formally, and leaves open many questions for future research. Is there an assumption \emph{weaker} than $(\alpha,\beta)$-niceness that is still sufficient for color coding to run in time linear in $H$?
Under the definition of $k$-hypergraphlet given by so-called section hypergraphs (see \Cref{sub:about_induced}), which properties of $H$ allow for efficient counting algorithms? Is $(\alpha,\beta)$-niceness helpful for other problems, too?

\section*{Acknowledgments}
The authors thank the Laboratory for Web Algorithmics\footnote{\url{https://law.di.unimi.it/}} for kindly providing access to computational facilities.
\clearpage

\bibliography{sample}

 \clearpage

\section*{APPENDIX}
\subfile{7_appendix}



\end{document}

%% file: tikz-hypergraph.tex
\def\rotateclockwise#1{
  \newdimen\xrw
  \pgfextractx{\xrw}{#1}
  \newdimen\yrw
  \pgfextracty{\yrw}{#1}
  \pgfpoint{\yrw}{-\xrw}
}

\def\rotatecounterclockwise#1{
  \newdimen\xrcw
  \pgfextractx{\xrcw}{#1}
  \newdimen\yrcw
  \pgfextracty{\yrcw}{#1}
  \pgfpoint{-\yrcw}{\xrcw}
}

\def\outsidespacerpgfclockwise#1#2#3{
  \pgfpointscale{#3}{
    \rotateclockwise{
      \pgfpointnormalised{
        \pgfpointdiff{#1}{#2}}}}
}

\def\outsidespacerpgfcounterclockwise#1#2#3{
  \pgfpointscale{#3}{
    \rotatecounterclockwise{
      \pgfpointnormalised{
        \pgfpointdiff{#1}{#2}}}}
}

\def\outsidepgfclockwise#1#2#3{
  \pgfpointadd{#2}{\outsidespacerpgfclockwise{#1}{#2}{#3}}
}

\def\outsidepgfcounterclockwise#1#2#3{
  \pgfpointadd{#2}{\outsidespacerpgfcounterclockwise{#1}{#2}{#3}}
}

\def\outside#1#2#3{
  ($ (#2) ! #3 ! -90 : (#1) $)
}

\def\cornerpgf#1#2#3#4{
  \pgfextra{
    \pgfmathanglebetweenpoints{#2}{\outsidepgfcounterclockwise{#1}{#2}{#4}}
    \let\anglea\pgfmathresult
    \let\startangle\pgfmathresult

    \pgfmathanglebetweenpoints{#2}{\outsidepgfclockwise{#3}{#2}{#4}}
    \pgfmathparse{\pgfmathresult - \anglea}
    \pgfmathroundto{\pgfmathresult}
    \let\arcangle\pgfmathresult
    \ifthenelse{180=\arcangle \or 180<\arcangle}{
      \pgfmathparse{-360 + \arcangle}}{
      \pgfmathparse{\arcangle}}
    \let\deltaangle\pgfmathresult

    \newdimen\x
    \pgfextractx{\x}{\outsidepgfcounterclockwise{#1}{#2}{#4}}
    \newdimen\y
    \pgfextracty{\y}{\outsidepgfcounterclockwise{#1}{#2}{#4}}
  }
  -- (\x,\y) arc [start angle=\startangle, delta angle=\deltaangle, radius=#4]
}

\def\corner#1#2#3#4{
  \cornerpgf{\pgfpointanchor{#1}{center}}{\pgfpointanchor{#2}{center}}{\pgfpointanchor{#3}{center}}{#4}
}

\def\hedgem#1#2#3#4{
  
  \outside{#1}{#2}{#4}
  \pgfextra{
    \def\hgnodea{#1}
    \def\hgnodeb{#2}
  }
  foreach \c in {#3} {
    \corner{\hgnodea}{\hgnodeb}{\c}{#4}
    \pgfextra{
      \global\let\hgnodea\hgnodeb
      \global\let\hgnodeb\c
    }
  }
  \corner{\hgnodea}{\hgnodeb}{#1}{#4}
  \corner{\hgnodeb}{#1}{#2}{#4}
  -- cycle
}

%% file: macros.tex
\newcommand{\E}{\mathbb{E}}
\newcommand{\ignore}[1]{}

\newcommand{\Ind}[1]{\mathbbm{1}\{#1\}}

\newcommand{\Gaif}{\mathsf{Gaif}\xspace}
\newcommand{\motivo}{\textsc{Motivo}\xspace}

\newcommand{\ccStyle}[1]{\ensuremath{\mathrm{#1}}\xspace}

\newcommand{\ccFPT}{\ccStyle{FPT}}

\newcommand{\ccSharpW}[1]{\#\ccStyle{W}[#1]}

\newcommand{\poly}{\operatorname{poly}}

\newcommand{\rank}{\operatorname{rank}}
\newcommand{\size}[1]{\ensuremath{\lVert{#1}\rVert}}

\newcommand{\N}{\mathbb{N}}
\newcommand{\R}{\mathbb{R}}

\newcommand{\problem}[1]{\textsc{#1}}
\newcommand{\algo}[1]{\textsc{#1}}
\newcommand{\hypermotivo}{\algo{HyperMotivo}\xspace}
\newcommand{\algoNWIE}{\algo{NW-IE}\xspace}
\newcommand{\assumption}[1]{\text{\emph{#1}}}

\newcommand{\bv}{\mathbf{v}}

\newcommand{\probHC}{\problem{HC}\xspace}
\newcommand{\probNW}{\problem{NW}\xspace}
\newcommand{\probNC}{\problem{NC}\xspace}
\newcommand{\probOV}{\problem{OV}\xspace}

\newcommand{\probKC}{\problem{k-Clique}\xspace}
\newcommand{\probKSH}{\problem{k-SH}\xspace}
\newcommand{\probCRT}{\problem{CRT}\xspace}

\newcommand{\nwNaive}{\algo{NW-Naive}\xspace}

\newcommand{\weaksub}[2]{{#1}[{#2}]}

\newcommand{\subcount}{\ensuremath{\#\!\operatorname{ind}}}

\newcommand{\Gaifman}{\mathsf{Gaif}}

\newcommand{\scV}{\mathcal{V}}

\newcommand{\Giacomo}[1]{#1}
\newcommand{\New}[1]{#1}

\newcommand{\SampleNeigh}{\algo{SampleNeigh}\xspace}
\newcommand{\Sample}{\algo{Sample}\xspace}
\newcommand{\C}{\mathcal{C}}

\newcommand{\algoNWNaive}{\algo{NW-Naive}}

\newcommand{\CC}{C\nolinebreak\hspace{-.05em}\raisebox{.4ex}{\tiny\bf +}\nolinebreak\hspace{-.10em}\raisebox{.4ex}{\tiny\bf +}}
\def\CC{{C\nolinebreak[4]\hspace{-.05em}\raisebox{.4ex}{\tiny\bf ++}}}

\newcommand{\err}{\operatorname{err}}

%% file: 1_introduction.tex
\ignore{
For a graph $G$ and an integer $k \ge 2$, a $k$-graphlet of $G$ is any induced connected subgraphs of $G$ on $k$ vertices.
The graphlet counting problem asks, given $G$ and $k$, to count the number of $k$-graphlets of $G$ by their isomorphism type (i.e., the number of induced $k$-stars, the number of induced $k$-paths, and so on).
The last two decades have produced a stream of results for graphlet counting, including algorithms based on random walks~\cite{Jha&2015,Han&2016,Chen&2016,Paramonov2019} and tight analyses of their mixing times~\cite{Agostini19mixing,Gionis&2020mixing,bressan2021efficient}, algorithms based on color coding~\cite{Bressan21TKDD}, and algorithms for streaming and distributed settings~\cite{Schilleretal15,Lai&2015,DeStefani&2017,Eden17triangles,Zhang20distributed,Bourreau2024}.
Many of those algorithms allow for optimized implementations which can deliver impressive practical performance.
In practice, $k$-graphlets can be counted in graphs with billions of edges even for $k=6,7,8$, in a matter of hours, and with very good accuracy guarantees---see for example~\cite{Bressan21TKDD}.
}

\New{
Hypergraphs, the higher-order generalization of graphs, are emerging as expressive and powerful models to represent complex interactions.
Recent work suggests that many real-world phenomena that have been traditionally analysed using graphs should rather be analysed using hypergraphs, in particular by looking at the frequency of small induced substructures (called \emph{hypergraphlets} or \emph{hypermotifs}).
For example, in collaboration networks, the classic phenomenon of triadic closure actually extends to its higher-order cousin, \emph{simplicial closure}, the tendency of a group to spawn sub-groups as well~\cite{patania2017shape,benson_simplicial_2018}.
Similarly,~\cite{Lotito2022} shows that the frequency vector of hypergraph motifs is correlated with the domain of the hypergraph, something that was previously known of graph motifs.
Other examples where hypergraph motif statistics carry fundamental information come from biology \cite{gaudeletdognin}, chemistry \cite{Andersen2017ChemicalTM}, social networks \cite{juul_pattern}, and machine learning \cite{antelmi}.
Clearly, in order to make all of this actionable, we need good \emph{algorithms} to compute hypergraph motif statistics.
Unfortunately, this is still an unsolved challenge.
According to \cite{gaudeletdognin}, already a decade ago we lacked ``sophisticated algorithms for mining'' hypergraph motifs, and the picture seems to not have changed since then.
The goal of our work is to make a first step towards the development of such algorithms.
}
\ignore{
For these reasons, researchers are investigating the problem of counting sub-hypergraphs, \Giacomo{with applications in biology \cite{gaudeletdognin}, chemistry \cite{Andersen2017ChemicalTM}, social networks \cite{juul_pattern}, machine learning \cite{antelmi}.}
For example, the problem of counting how many subgroups of $10$ people that are part of a larger community in a social network can be cast as a sub-hypergraph counting problem.
Unfortunately, the landscape for sub-hypergraph counting is, to say the least, meager, especially if compared with subgraph counting.
}
%
%

To discuss the matter, let us define the problem more formally.
A hypergraph is a pair $H=(V,E)$ where $V$ is a finite set and $E \subseteq 2^V \setminus \{\emptyset\}$.
The sub-hypergraph of $H=(V,E)$ induced by $U \subseteq V$, denoted by $H[U]$, is the hypergraph with vertex set $U$ and edge set $\{e \cap U : e \in E\} \setminus \{\emptyset\}$.\footnote{There are other notions of induced sub-hypergraphs, but the corresponding problem is computationally even harder, see below.}
A $k$-hypergraphlet of $H$ is an induced connected sub-hypergraph $H[U]$ of $H$ with $|U|=k$.
The \emph{Hypergraphlet Counting} problem (HC for short) asks, given a hypergraph $H$ and an integer $k \ge 2$, to compute\footnote{We actually allow for approximate counts (exact counts are computationally hard in a strong sense), but we avoid being too precise for now.} the number $t_i$ of $k$-hypergraphlets of $H$ that are isomorphic to $H_i$, where $H_i$ for $i=1,2,\ldots$ ranges over all isomorphism-distinct $k$-vertex hypergraphs.

\New{
In the special case of \emph{graphs}, HC boils down to Graphlet Counting (GC), the problem of counting connected $k$-vertex subgraphs.
For GC, the state-of-the-art algorithms are built on the celebrated color-coding technique of Alon, Yuster and Zwick~\cite{AlonYZ95}.
This technique yields sampling and approximate graphlet counting algorithms that for every fixed $k$ run in time linear in the input.
More precisely, there is a preprocessing phase that takes time $O(\size{H})$ where $\size{H}=|V|+\sum_{e \in E}|e|$, followed by random sampling in $O(\lg |V|)$-time per sample.
Besides being elegant and clean, this approach yields formal statistical guarantees through concentration bounds.
An optimized implementation, \motivo, can easily manage graphs with billions of edges~\cite{Bressan21TKDD}.
Quite surprisingly, the picture for HC is completely different: we have basically no efficient algorithm, except for the brute-force ones that take time $|V|^{\Theta(k)}$~\cite{lotito_exact_2024}.
Our goal is to explore this gap, and answer the following question:
\begin{center}
\textit{Can we devise algorithms for HC as good as those for GC?\newline If yes, how? If not, why?}
\end{center}
}
\ignore{
, including algorithms based on random walks~\cite{Jha&2015,Han&2016,Chen&2016,Paramonov2019} and tight analyses of their mixing times~\cite{Agostini19mixing,Gionis&2020mixing,bressan2021efficient}, algorithms based on color coding~\cite{Bressan21TKDD}, and algorithms for streaming and distributed settings~\cite{Schilleretal15,Lai&2015,DeStefani&2017,Eden17triangles,Zhang20distributed,Bourreau2024}.
Many of those algorithms allow for optimized implementations which can deliver impressive practical performance.
In practice, $k$-graphlets can be counted in graphs with billions of edges even for $k=6,7,8$, in a matter of hours, and with very good accuracy guarantees---see for example~\cite{Bressan21TKDD}.
}
%
Given the success of color coding for counting graphlets, it seems compelling to try adapting it to hypergraphlets.
The \emph{Gaifman} graph or \emph{primal} graph $G=\Gaifman(H)$ of $H$ is the graph where $u,v$ are adjacent if they share a hyperedge of $H$.
It is easy to see that $H[U]$ is a $k$-hypergraphlet if and only if $G[U]$ is a $k$-graphlet.
This observation can be exploited as follows.
First, compute $G=\Gaifman(H)$; we call this step \emph{projection}.\footnote{In the literature, this is also called \emph{Clique Expansion (CE).}}
Then, use an $O(\size{G})$-time color-coding algorithm on $G$ to sample random $k$-graphlets $G[U]$; each graphlet can then be converted into its corresponding $k$-hypergraphlet $H[U]$ in $H$.
The same sample weighting techniques of~\motivo then yield unbiased hypergraphlet count estimates with concentration guarantees.
This looks great, but there is a big catch: the size of $G$ can be \emph{quadratic} in $\size{H}$.
Indeed, every edge $e \in E(H)$ implies at least ${|e| \choose 2}$ edges of $G$.
Thus, while for graphs this approach yields an $O(\size{H})$-time algorithm, for hypergraphs it only yields a much worse $O(\size{H}^2)$-time algorithm.
Needless to say, this is prohibitive for large hypergraphs.
Is this unavoidable, or can we do better?


\subsection*{Our contributions}

\noindent\textbf{1) A quadratic barrier.}
We show that, when the input is a hypergraph, the $\Omega(\size{H}^2)$ barrier above is unavoidable for the color-coding approach.
More precisely we prove that, under the well-known Orthogonal Vectors Conjecture,\footnote{To be precise we use the Moderate-Dimension Orthogonal Vector Conjecture, which says that there is no $O(n^{2-\epsilon} \poly(d))$-time algorithm for the following problem: given $n$ vectors from $\{0,1\}^d$, decides if there are two of those vectors that are orthogonal.} the output of color-coding's dynamic program cannot be computed in time $O(n^{2-\epsilon})$ for any $\epsilon>0$.
Thus, without further assumptions, color coding cannot yield linear-time algorithms for HC, as it does for GC.
We also show that, under another classic definition of induced sub-hypergraph, even deciding if $H$ contains \emph{at least one} $k$-hypergraphlet is $\ccSharpW{1}$-hard and unlikely to be solvable time $|V|^{o(\sqrt[3]{k})}$.
Thus, our definition is arguably the only one that possibly allows for efficient algorithms.

\smallskip

\noindent\textbf{2) A new algorithm.} We introduce a new hypergraph property,  \emph{$(\alpha,\beta)$-niceness}. A hypergraph $H=(V,E)$ is $(\alpha,\beta)$-nice if $E$ admits a partition $E_{\le \alpha}, E_{>\alpha}$ such that $(V,E_{\le\alpha})$ has rank at most $\alpha$ and $(V,E_{>\alpha})$ has degree at most $\beta$.\footnote{The degree is the maximum  number of edges incident with a vertex, and the rank is the maximum size of an edge. See~\Cref{sec:prelims}.}
All the hypergraphs used in our experiments appear to be $(\alpha,\beta)$-nice for small values of $\alpha,\beta$.
We then give a two-phase algorithm, \hypermotivo, that preprocesses $H$ in time
    \[
    2^{O(k)} \cdot \left(2^\beta\, |V| + \alpha^k\, |E| + \alpha^2 \beta\, \size{H}\right)
    \]
where, again, $\size{H}=|V|+\sum_{e \in E}|e|$.
Afterwards, \hypermotivo can draw samples in expected time $k^{O(k)} \cdot (\beta^2 + \log |V|)$ per sample.
All the statistical guarantees provided by \motivo's color-coding approach are retained (e.g., the concentration of estimates).
Each phase of \hypermotivo is based on a careful combination of several ingredients, each of which exploits the definition of $(\alpha,\beta)$-niceness in a specific way in order to achieve the bounds above.
\smallskip

\noindent\textbf{3) Experiments.} We provide an efficient, multi-threaded C++ implementation of \hypermotivo, and we test it on six hypergraphs obtained from real-world data (plus a synthethic one).
In particular, we compare \hypermotivo to the naive approach above that uses \motivo on $\Gaifman(H)$.
We observe an improvement of several times in terms of time and memory usage, achieving a $30 \times$ on some hypergraphs.
Thanks to it, we are able to provide estimates of the frequency distribution of the most frequent $k$-hypergraphlets, for $k=3,4,5$, on all our hypergraphs.

\begin{figure}
    \centering
\begin{tikzpicture}[framed,
    scale=.7,
    vertex/.style={circle,draw,thick,fill=white,minimum size=8mm},
    small edge/.style={thick},
    big edge/.style={ellipse,thick,fill opacity=0.15}
]
\node (a) at (-1.5,0) {$a$};
\node (b) at (-1.1,-.8) {$b$};
\node (c) at (-.8,.7) {$c$};
\node (d) at (-.4,-.1) {$d$};
\node (e) at (0,-.9) {$e$};
\node (f) at (.3,1) {$f$};
\node (g) at (.6,.1) {$g$};
\node (h) at (.9,-.8) {$h$};
\draw[small edge] (a) -- (b);
\draw[small edge] (b) -- (e);
\draw[thick, inner sep=5pt, tension=.8] plot [smooth cycle] coordinates {($(a)+(-.2,0)$) ($(b)+(-.1,-.2)$) ($(e)+(0,-.2)$) ($(g)+(.15,.15)$) ($(c)+(0,.2)$)};
\draw[thick, inner sep=5pt, tension=.9] plot [smooth cycle] coordinates {($(d)+(-.2,-.2)$) ($(g)+(.2,0)$) ($(f)+(0,.25)$)};
\end{tikzpicture}
\hspace*{3em}
\begin{tikzpicture}[framed,
    scale=.7,
    vertex/.style={circle,draw,thick,fill=white,minimum size=8mm},
    small edge/.style={thick},
    big edge/.style={ellipse,thick,fill opacity=0.15}
]
\node (a) at (-1.5,0) {$a$};
\node (b) at (-1.1,-.8) {$b$};
\node (c) at (-.8,.7) {$c$};
\node (d) at (-.4,-.1) {$d$};
\node (e) at (0,-.9) {$e$};
\node (f) at (.3,1) {$f$};
\node (g) at (.6,.1) {$g$};
\node (h) at (.9,-.8) {$h$};
\draw[small edge] (a) -- (b);
\draw[small edge] (b) -- (e);
\draw[thick, inner sep=5pt, tension=.8, draw opacity=0] plot [smooth cycle] coordinates {($(a)+(-.2,0)$) ($(b)+(-.1,-.2)$) ($(e)+(0,-.2)$) ($(g)+(.15,.15)$) ($(c)+(0,.2)$)};
\draw[thick, inner sep=5pt, tension=.9] plot [smooth cycle] coordinates {($(d)+(-.2,-.2)$) ($(g)+(.2,0)$) ($(f)+(0,.25)$)};
\end{tikzpicture}
\hspace*{0em}
\begin{tikzpicture}[framed,
    scale=.7,
    vertex/.style={circle,draw,thick,fill=white,minimum size=8mm},
    small edge/.style={thick},
    big edge/.style={ellipse,thick,fill opacity=0.15}
]
\node (a) at (-1.5,0) {$a$};
\node (b) at (-1.1,-.8) {$b$};
\node (c) at (-.8,.7) {$c$};
\node (d) at (-.4,-.1) {$d$};
\node (e) at (0,-.9) {$e$};
\node (f) at (.3,1) {$f$};
\node (g) at (.6,.1) {$g$};
\node (h) at (.9,-.8) {$h$};
\draw[thick, inner sep=5pt, tension=.8] plot [smooth cycle] coordinates {($(a)+(-.2,0)$) ($(b)+(-.1,-.2)$) ($(e)+(0,-.2)$) ($(g)+(.15,.15)$) ($(c)+(0,.2)$)};
\draw[thick, inner sep=5pt, tension=.9,draw opacity=0] plot [smooth cycle] coordinates {($(d)+(-.2,-.2)$) ($(g)+(.2,0)$) ($(f)+(0,.25)$)};
\end{tikzpicture}
    \caption{A toy hypergraph $H$ (left) and its $\alpha$-split into $H_{\le \alpha}$ (right, first) and $H_{>\alpha}$ (right, second) for $\alpha=4$.
    Our algorithm employs different techniques on the two sub-hypergraphs and combines the results carefully, both in the preprocessing and sampling phase.}
    \label{fig:toy}
\end{figure}
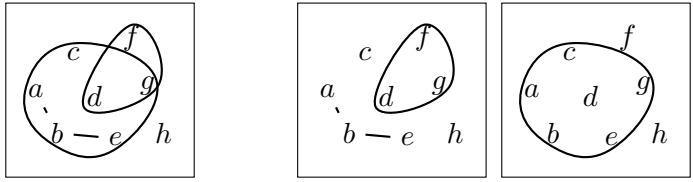

%% file: 2_basic_definition_and_properties.tex

A hypergraph is a pair $H=(V, E)$ where $V$ is a finite set and $E \subseteq 2^V \setminus \{\emptyset\}$. The elements of $V$ are called vertices and the elements of $E$ are called edges.
%
%
%
The \emph{rank} of $H$ is $r(H) = \max_{e \in E} |e|$.
The \emph{type} of a vertex $v \in V$ is the set of edges incident with it, $E(v) = \{ e \in E : v \in e\}$, and its \emph{degree} is $d(v)=|E(v)|$.
We let $\Delta(H)=\max_{v \in V}d(v)$.
Two vertices $u,v \in V$ are \emph{adjacent} if $E(u) \cap E(v) \ne \emptyset$, that is, if they share an edge; we write $u \sim v$.
The \emph{neighborhood} of a vertex $v$ is $N(v) = \{u \in V : u \sim v\}$.
A \emph{path} in $H$ is a sequence $x_1 e_1 x_2 e_2 \dots x_s e_s x_{s+1}$ such that $x_1, \dots, x_{s+1} \in V$ are distinct, $e_1,\ldots,e_s \in E$, and $x_i,x_{i+1} \in e_i$ for all $i=1,\ldots,s$.
The \emph{length} of the path is $s$, and we say that the path is a \emph{$u$-$v$ path}, or \emph{a path from $u$ to $v$}, where $u=x_1$ and $v=x_{s+1}$; we also say that $u,v$ are \emph{connected}.
We say $H$ is connected if all $u,v \in V$ are connected.
The \emph{Gaifman graph} (or \emph{primal graph}) of a hypergraph $H$, denoted by $\Gaif(H)$, is the simple undirected graph with $V(G)=V(H)$ and $E(G) = \{ \{u, v\} \subseteq V \mid \exists e \in E, \{u, v\} \subseteq e \}$.
In words, $u,v \in V(H)$ are adjacent in $\Gaif(H)$ if and only if they share some edge in $H$.

Let $H=(V,E)$ and $H'=(V',E')$ be hypergraphs.
We say $H'$ is a \emph{subhypergraph} of $H$ if $V' \subseteq V$ and $e' \in E$ for every $e' \in E'$.
We shall now present the notion of induced sub-hypergraph considered in this work.

\begin{definition}[Induced sub-hypergraph]\label{def:subhyp} Let $H = (V, E)$ be a hypergraph and $U\subseteq V$.
The sub-hypergraph induced by $U$ in $H$, denoted $\weaksub{H}{U}$, is the hypergraph with vertex set $V(\weaksub{H}{U}) = U$ and edge set $E(\weaksub{H}{U})=\{ e \cap U : e \in E, e \cap U \ne \emptyset \}$.
\end{definition}
%
A \emph{k-hypergraphlet} of $H$ is any induced connected sub-hypergraph of $H$ on $k$ vertices.
%
Let:
\begin{align}
    \scV_k(H) = \{U \subseteq V \,:\, \weaksub{H}{U} \text{ is connected}\}.
\end{align}
The set of $k$-hypergraphlets of $H$ is:
\begin{align}
    \mathcal{G}_k(H) = \{\weaksub{H}{U} \,:\, U \in \scV_k(H) \}.
\end{align}
For every connected hypergraph $H_1,H_2,\ldots$ on $k$-vertices let:
\begin{align}
    \subcount(H,H_i) \triangleq \left|\{U \in \scV_k(H) \,:\, \weaksub{H}{U} \simeq H_i\}\right|.
\end{align}
where $\simeq$ denotes isomorphism.
The main problem we consider in this work is:
%
\begin{definition}\label{def:HC}
    The \emph{Hypergraphlet Counting} problem (HC) asks, given a hypergraph $H$ and an integer $k$, to compute $\subcount(H,H_i)$ for every connected hypergraph $H_1, H_2, \ldots$ on $k$ vertices.
\end{definition}
It is well-known that the exact version of HC is computationally hard.
For instance, it subsumes the problem of counting $k$-cliques, which under standard complexity-theoretical assumptions is not solvable in polynomial time and, in fact, not even in time $|V|^{o(k)}$~\cite{Chenetal05}.
Thus, we consider a relaxed, approximate version of \Cref{def:HC}, where one aims at estimates of the counts $\subcount(H,H_i)$ accurate within an additive $\pm \epsilon \cdot \left(\sum_i\subcount(H,H_i)\right)$.
This means that hypergraphlets with a high relative frequency are well approximated.
This is the kind of guarantees given by state-of-the art algorithms for graphlet counting, such as \motivo~\cite{Bressan21TKDD}.

The description of our algorithm borrows notation and concepts from~\cite{Bressan21TKDD}.
A $k$-\emph{treelet} is a tree on $k$ vertices (we often talk about treelet of a graph $G$ to mean a subgraph of $G$ that is a tree on $k$ vertices).
We let $[k]=\{1,\ldots,k\}$.
A $k$-\emph{coloring} of a (hyper)graph with vertex set $V$ is a mapping $c : V \to [k]$.
A subset $U \subseteq V$ is \emph{colored} with/by $S \subseteq [k]$ if $\{c(u):u \in U\}=S$, and is \emph{colorful} if it additionally satisfies $|U|=|S|$.

We assume the RAM model with words of size $\Theta(\log n + k)$.
To ease our analysis, we also assume that indexing arrays/tables whose keys are tuples of $h\le k$ vertices of $H$, or treelets on $h \le k$ vertices, or subsets of $[k]$, takes constant time (in particular the counters $\C(\cdot)$, see below).
Note that this does not mean we treat $k$ as a constant; in particular, our bounds state explicitly any exponential dependence on $k$.

\subsection{On the definition of hypergraphlet}\label{sub:about_induced}
One may argue that the right definition of induced sub-hypergraph is the one where each edge is either taken integrally or not taken at all, and not truncated as per \Cref{def:subhyp}.
Such sub-hypergraphs are often referred to as \textit{section hypergraphs}:
\begin{definition}\label{def:SectionHypergraph}
    Let $H=(V,E)$ be a hypergraph and $U \subseteq V$. The \emph{section hypergraph} induced by $U$ in $H$, denoted by $H\langle U \rangle$, is the hypergraph with $V(H \langle U \rangle) = U$ and $E(H\langle U\rangle) = \{e \in E : e \subseteq U\}$.
\end{definition}
Counting section sub-hypergraphs is a natural problem, and has been studied by~\cite{lotito_exact_2024}.
Unfortunately, we show that counting section sub-hypergraphs is hard in a very strong sense.
Let \probKSH denote the following problem: given a hypergraph $H = (V, E)$ and an integer $k \ge 2$, decide if $H$ contains at least one $k$-vertex section sub-hypergraph that is connected---that is, $U \subseteq V$ with $|U|=k$ such that $H\langle U \rangle$ is connected.
For graphs, this is equivalent to deciding whether there is a connected component on at least $k$ vertices, and is thus in linear time.
Somewhat surprisingly, we show:
\ignore{
\begin{definition}\label{def:kSectionHypergraphProblem}
The \emph{k-Section Hypergraph} problem (\probKSH) asks, given a hypergraph $H = (V, E)$ and an integer $k \ge 2$, to decide if there exists $U \subseteq V$ such that $|U|=k$ and $H[U]$ is connected.
\end{definition}}
\begin{theorem}\label{thm:KSH-Hard}
    Under the Exponential Time Hypothesis, \probKSH cannot be solved in time $n^{o(\sqrt[3]{k})}$.
\end{theorem}
\noindent Thus, if one defines $k$-hypergraphlets via section sub-hypergraphs, then there is no hope for efficient counting algorithms, even approximate.
Our definition of hypergraphlet is therefore, arguably, the most natural one that is algorithmically useful.
The proof of \Cref{thm:KSH-Hard} can be found in the appendix.

\section{Related Work}
As said, we build on the color-coding technique~\cite{AlonYZ95} and the related algorithm of~\cite{Bressan&2019VLDB}.
However, our techniques are substantially novel, although simple.
Only a handful of other works seem related to our problem.
\cite{Lee&2020-hypergraph-motifs} studies the problem of counting, for each possible Venn diagram on $3$ sets, the number of hyperedges triplets in $H$ that yield that Venn diagram (where every region of the diagram tells whether the corresponding intersection is empty or not). Thus, they do not count induced sub-hypergraphs, and not on a given number $k$ of vertices. Moreover, their algorithm incurs an $\Omega(|E|^2)$ lower bound, as it precomputes a graph where two hyperedges of $H$ are adjacent when they intersect.
\cite{lotito_exact_2024} gives algorithms for counting induced \emph{section} sub-hypergraphs on $k$ vertices, but only for $k=3$ (which they do exactly) or $k=4$ (approximately, via sampling).
The basic idea is to first enumerate all hyperedges of size $3$, computing the sub-hypergraph they induce. One then deletes all hyperedges of size at least $3$, remaining with a graph, and counts exactly all its $3$-graphlets through the ESU exhaustive enumeration algorithm~\cite{wernicke_efficient_2006}. Clearly, this requires time $\Omega(|V|^3)$ in the worst case.
A few other results are known about the complexity of counting sub-hypergraphs exactly, but again under the notion of section sub-hypergraph, and with the goal of having a polynomial dependence on $H$, rather than a linear one~\cite{bressan2025complexitycountingsmallsubhypergraphs,Marx13}.


%% file: 3_algorithm.tex
\label{sec:quad_lim}
This section describes the color-coding approach and its inherent quadratic-time barrier.
Roughly speaking, we show that the color-coding technique requires time $\Omega(|V| + \sum_{e \in E}|e|^2)$ on hypergraphs, unless a popular conjecture from computational complexity fails.
The description of the color-coding technique is taken from~\cite{Bressan21TKDD}, and we often use ``\motivo'' and ``color-coding'' interchangeably.

Let $G=(V,E)$ be a graph, and suppose we want to sample connected $k$-vertex subgraphs of $G$ uniformly at random.
The basic idea of color coding is to make the problem easier by coloring the vertices randomly.
The technique has a \emph{build-up phase} and a \emph{sampling phase}, as follows.

\paragraph{The build-up phase}
First, every vertex $v \in V$ is independently assigned a random uniform color $c_v \in \{1,\ldots,k\}$.
Next, for every $h = 1,\ldots,k$, every rooted treelet $T$ on $h$ vertices, and every subset $S \subseteq \{1,\ldots,k\}$ with $|S|=h$, for every $v \in V$ we compute the number $\C(T, S, v)$ of $h$-treelets of $G$ that are isomorphic to $T$, colored by $S$, and rooted in $v$. 
This computation is performed using a dynamic programming approach.
First, $T$ is virtually split into two subtrees $T_1$, $T_2$ such that $T_1$ has the same root of $T$ and $T_2$ is rooted at a child of that root.
Then, the counter of $T$ is computed according to the following equation:
\begin{equation}\label{eq:treelet_counter}
    \C(T, S, v) = \frac{1}{d}\sum_{\substack{S_1,S_2 \subset S\\S_1 \cap S_2 = \emptyset}} \C(T_1,S_1,v)\sum_{u \sim v}\C(T_2,S_2,u).
\end{equation}
where $d$ is a normalizing factor that depends only on $T$.
It is not hard to prove that the whole build-up phase takes time $2^{O(k)} \cdot O(\size{G})$.

\paragraph{The sampling phase}\label{motivo_sample}
Using the counters $\C$ defined above, one can sample uniformly at random a $k$-treelet of $G$ that is \emph{colorful}---that is, such that every color $c \in \{1,\ldots,k\}$ appears in some vertex of the treelet.
This is done through multi-stage sampling, as follows.
First, pick a tree $T$ on $k$ vertex, as well as a vertex $v \in V$, with probability proportional to $\C(T, S, v)$, the total number of colorful $k$-treelets of $G$ rooted in $v$ that are isomorphic to $T$.
Second, split $T$ into $T_1$ and $T_2$ as above, and choose a partition $S_1,S_2$ of $S=[k]$ with probability proportional to $\C(T_1, S_1, v) \cdot \sum_{u \sim v} \C(T_2, S_2, u)$.
Third, choose a neighbor $u$ of $v$ with probability proportional to $\C(T_2, S_2, u)$.
Finally, repeat the procedure recursively to sample a random copy of $T_1$ rooted at $v$ with colors $S_1$, and a random copy of $T_2$ rooted at $u$ with colors $S_2$.
The union of the two copies yields the random colorful copy of $T$ rooted at $v$.
Once a colorful copy of $T$ is sampled, one takes its vertex set $U$ and stores the corresponding $k$-graphlet $G[U]$.
Finally, as the probability of sampling $G[U]$ is proportional to the number of its spanning trees, one accepts $G[U]$ with probability \emph{inversely} proportional to that number, and repeats the process otherwise.
This makes the distribution of accepted graphlets uniform over all colorful $k$-graphlets of $G$.
One can implement the sampling phase so that it runs in expected $2^{O(k)} \cdot O(\log n)$ time per sample.

\ignore{
    In the sampling phase, these counters are used to sample graphlets uniformly at random. A single sampling step works as follows: a random vertex is first selected, and a colorful treelet $T$ is chosen with probability proportional to $\C(T, [k], v)$. To select the next vertex, $T$ is recursively split in $T_1$ and $T_2$, and the subsets of color $S_1$ and $S_2$ are chosen with probability proportional to $\C(T_1, S_1, v) \cdot \sum_{u \sim v} \C(T_2, S_2, u)$, where the neighbor $u$ is chosen proportionally to $\C(T_2, S_2, u)$. This procedure is repeated recursively on $\C(T_1, S_1, v)$  and $\C(T_2, S_2, u)$. 
    Finally, the colorful graphlet $H$ induced by the vertices of $T$ is retrieved from the graph and its counter (namely, the counter of its isomorphism type) is increased.
}

\medskip
The technique above can be extended to hypergraphs.
The key observation is that, if $G=\Gaifman(H)$, then $H[U]$ is connected \emph{if and only if} $G[U]$ is connected.
Thus, the $k$-hypergraphlets of $H$ are precisely the $k$-graphlets of $G$ (in terms of the vertex subsets inducing them).
Hence, we may count the $k$-hypergraphlets of $H$ by first computing $G$ and then running  color-coding on it (with the only catch of checking the $k$-hypergraphlet $H[U]$ rather than the $k$-graphlet $G[U]$).
    
\begin{algorithm}[h]
\LinesNumbered
\caption{The Naive Baseline}
\label{algo:Naive}
    {\fontsizealgo
    \Input{hypergraph $H=(V, E)$, integer $k \ge 2$}
    compute $G = \Gaifman(H)$\;
    run \motivo on $G,k$\;
    }
\end{algorithm}

\ignore{
\begin{figure}[h]
            
            
    \input{figures/4-edge-hypergraph.tikz}
    \vspace{1ex}
    \begin{minipage}{\columnwidth}
        \centering
    
        \begin{subfigure}[t]{0.49\columnwidth}
                \centering
                \begin{tikzpicture}
                \node[fill, shape=circle,label=above:\(v_1\)] (v1) {};
                \node[fill, shape=circle,right of=v1,label=above:\(v_2\)] (v2) {};
                \node[fill, shape=circle,right of=v2,label=above:\(v_3\)] (v3) {};
                \node[fill, shape=circle,below of=v1,label=below:\(v_4\)] (v4) {};
                \node[fill, shape=circle,right of=v4,label=below:\(v_5\)] (v5) {};
                \node[fill, shape=circle,right of=v5,label=below:\(v_6\)] (v6) {};
    
                \draw[line width=1.5pt] (v4) -- (v5) -- (v1);
                \draw[dotted, line width=1.2pt] (v1) -- (v2) -- (v3) -- (v6) -- (v5) -- (v2);
                \draw[dotted, line width=1.2pt] (v2) -- (v6);
                \draw[dotted, line width=1.2pt] (v5) -- (v3);
                \end{tikzpicture}\textbf{}
        \end{subfigure}\hfill
        \begin{subfigure}[t]{0.49\columnwidth}
                \centering
                \begin{tikzpicture}
                \node[fill, shape=circle,label=above:\(v_1\)] (v1) {};
                \node[fill, shape=circle,right of=v1,label=above:\(v_2\)] (v2) {};
                \node[fill, shape=circle,right of=v2,label=above:\(v_3\)] (v3) {};
                \node[fill, shape=circle,below of=v1,label=below:\(v_4\)] (v4) {};
                \node[fill, shape=circle,right of=v4,label=below:\(v_5\)] (v5) {};
                \node[fill, shape=circle,right of=v5,label=below:\(v_6\)] (v6) {};
    
                \draw[line width=1.5pt] (v4) -- (v5) -- (v2);
                \draw[dotted, line width=1.2pt] (v1) -- (v2) -- (v6) -- (v3) -- (v2) -- (v6) -- (v3) -- (v5) -- (v6);
                \draw[dotted, line width=1.2pt] (v1) -- (v5);
                \end{tikzpicture}
        \end{subfigure}
    
        \vspace{0.5ex}
    
        \begin{subfigure}[t]{0.49\columnwidth}
                \centering
                \begin{tikzpicture}
                \node[fill, shape=circle,label=above:\(v_1\)] (v1) {};
                \node[fill, shape=circle,right of=v1,label=above:\(v_2\)] (v2) {};
                \node[fill, shape=circle,right of=v2,label=above:\(v_3\)] (v3) {};
                \node[fill, shape=circle,below of=v1,label=below:\(v_4\)] (v4) {};
                \node[fill, shape=circle,right of=v4,label=below:\(v_5\)] (v5) {};
                \node[fill, shape=circle,right of=v5,label=below:\(v_6\)] (v6) {};
    
                \draw[line width=1.5pt] (v4) -- (v5) -- (v3);
                \draw[dotted, line width=1.2pt] (v1) -- (v2) -- (v5) -- (v6) -- (v3) -- (v2) -- (v6);
                \draw[dotted, line width=1.2pt] (v1) -- (v5);
                \end{tikzpicture}
        \end{subfigure}\hfill
        \begin{subfigure}[t]{0.49\columnwidth}
                \centering
                \begin{tikzpicture}
                \node[fill, shape=circle,label=above:\(v_1\)] (v1) {};
                \node[fill, shape=circle,right of=v1,label=above:\(v_2\)] (v2) {};
                \node[fill, shape=circle,right of=v2,label=above:\(v_3\)] (v3) {};
                \node[fill, shape=circle,below of=v1,label=below:\(v_4\)] (v4) {};
                \node[fill, shape=circle,right of=v4,label=below:\(v_5\)] (v5) {};
                \node[fill, shape=circle,right of=v5,label=below:\(v_6\)] (v6) {};
    
                \draw[line width=1.5pt] (v4) -- (v5) -- (v6);
                \draw[dotted, line width=1.2pt] (v1) -- (v2) -- (v3) -- (v6) -- (v2) -- (v5) -- (v3);
                \draw[dotted, line width=1.2pt] (v1) -- (v5);
                \end{tikzpicture}
        \end{subfigure}
    \end{minipage}

  \caption{Example of hypergraph and all of its treelets rooted in $v_4$}
  \label{fig:hypertreelets}
\end{figure}

\begin{figure}
    \centering
    \caption{Caption}

    \label{fig:testhyper2}
\end{figure}
}

As noted in the introduction, \Cref{algo:Naive} in the worst case runs in time $\Omega(\size{H}^2)$, because in the worst case $\size{\Gaifman(H)}=\Omega(\size{H}^2)$.
One could try to bypass this obstacle by applying \Cref{eq:treelet_counter} to $H$, rather than to $\Gaifman(H)$.
However, applying \Cref{eq:treelet_counter} by listing all neighborhoods in $H$ still takes time $\Omega(|V|+\sum_{e\in E}|e|^2)$, and it is not clear how this could be avoided.
The next section shows that this is not an accident.

\subsection{The quadratic barrier}
We show that, unless the widely accepted Orthogonal Vector Conjecture fails, computing the counters given by \Cref{eq:treelet_counter} actually \emph{requires} time essentially $\Omega(|V|+\sum_{e\in E}|e|^2)$.
The Orthogonal Vector problem (OV) asks, given a set $A$ of $n$ Boolean vectors of dimension $d$, whether there exist $u,v \in A$ orthogonal, i.e., such that $u[i]\cdot v[i] = 0$ for all $i \in \{1,...,d\}$.
The Strong Exponential Time Hypothesis (SETH) implies the following lower bound for OV, named \assumption{Moderate-dimension OVC}; see~\cite{gao2016orthogonal}.
\begin{conjecture}[\emph{Moderate-dimension OVC}]\label{conj:OVC}
There is no algorithm for OV running in time ${O}(n^{2-\varepsilon}poly(d))$ for any $\varepsilon > 0$.
\end{conjecture}
\noindent Our result is:
\begin{theorem}[Quadratic barrier of color coding on hypergraphs]\label{thm:quadratic}
Assume the moderate-dimension OVC (\Cref{conj:OVC}).
Then, for every $k \ge 2$ and $\epsilon > 0$, no algorithm can compute the counters of \Cref{eq:treelet_counter} on a hypergraph $H=(V,E)$ in time $O(|V|+\sum_{e \in E}|e|^{2-\epsilon})$.
\end{theorem}
As computing the counters of \Cref{eq:treelet_counter} is the very heart of the color-coding technique, \Cref{thm:quadratic} says that color coding alone is unlikely to give\\a sub-quadratic algorithm for Hypergraphlet Counting.
Note also that \Cref{thm:quadratic} still agrees with the fact that, \emph{on graphs}, color-coding yields linear-time algorithms---indeed, for a graph we have $\sum_{e \in E}|e|^2 = 4|E|$.

The rest of this section gives the proof of \Cref{thm:quadratic}; the uninterested reader may skip it.

\subsubsection{Proof of \Cref{thm:quadratic}}\label{sec:NWIntro}
The proof goes through an intermediate problem, the \emph{Neighbor Count} problem (\probNC), which, given in input a hypergraph $H=(V,E)$, asks to compute the cardinality $|N(v)| = |\{u \in V : u \sim v\}|$ of each vertex $v \in V$.
We show that OVC reduces efficiently to \probNC, then we show that \probNC reduces efficiently to the problem of computing the counters of \Cref{eq:treelet_counter}.

\begin{lemma}\label{lemma:no_neigh}
\probOV can be solved in time $O(dn)+T(n,d)$, where $T(n,d)$ is the complexity of \probNC on hypergraphs with $n$ vertices and $d$ edges.
\end{lemma}
\begin{proof}
Let $\bv_1,\ldots,\bv_n \in \R^d$ be the input of \probOV.
Consider the hypergraph $H=(V,E)$ with $V=\{\bv_1,\ldots,\bv_n\}$ and $E=\{e_1,\ldots,e_d\}$, where $e_\ell=\{\bv_i : \bv_i[\ell]=1\}$ for $\ell=1,\ldots,d$; note that $H$ can be constructed in time $O(nd)$.
Now, observe that $\bv_i \cdot \bv_j = 0$ if and only if $E_{\bv_i} \cap E_{\bv_j} = \emptyset$, that is, if $\bv_i \nsim \bv_j$.
Thus, the answer to \probOV is YES if and only if $H$ contains two vertices that are not adjacent, that is, if and only if $|N(v)|<n-1$ for some $v \in V$.
This can be checked in time $O(|V|)=O(dn)$ from the solution to \probNC.
We conclude that \probOV can be solved in time $O(dn)+T(n,d)$, as claimed.
\end{proof}

\begin{lemma}
For every fixed $k \ge 2$, \probNC can be solved in time $O(\size{H})+R(k|V|,|E|)$, where $R(kn,d)$ is the complexity of computing the counters of \Cref{eq:treelet_counter} on a hypergraph with $kn$ vertices and $d$ edges.
\end{lemma}
\begin{proof}
Let $H=(V,E)$ be the input of \probNC.
We construct the hypergraph $H'=(V',E')$ where $V' = V \times [k]$ and $E'$ contains, for every $e \in E$, the hyperedge $e' =\{(u,c) : u \in e, \ c \in [k] \}$ to $E(H')$.
We also construct the coloring $c: V' \to [k]$ given by the indices $c$ of the vertices $(v,c) \in V'$.
As $k$ is fixed, both $H$ and the coloring $c$ can be constructed in time $O(\size{H})$.
Note also that $|V'|=k|V|$ and $|E'|=|E|$.

Suppose now we have computed the counters of $H'$ in time $R(k|V|,|E|)$.
Let $T$ be the $k$-star (i.e., the star on $k-1$ leaves, $K_{1,k-1}$) and let $S=[k]$.
For every $v' = (v,c) \in V'$, by definition of the counters, and by a simple counting argument, we have:
\[
\C(K_{1,k-1}, [k], v') = (|N(v)|)^{k-1}\,,
\]
where $N(v)$ is the neighborhood of $v$ in $H$.
Therefore $|N(v)| = \left(\C(K_{1,k-1}, [k], v')\right)^{\frac{1}{k-1}}$ for every $v \in V$ and every $v' = (v,c) \in V'$.
Hence, given the counters of $H'$, one can compute the solution to \probNC on $H$ in time $O(|V|)=O(\size{H})$.
This completes the proof.
\end{proof}

\noindent \textbf{Wrap-up.}
We conclude the proof of \Cref{thm:quadratic}.
Suppose that, for some $k \ge 2$ and $\epsilon > 0$, we can compute the counters of \Cref{eq:treelet_counter} in time $O(|V| + \sum_{e \in E} |e|^{2-\varepsilon})$.
This implies:
\[
R(kn,d) = O\left(n + \sum_{e \in E} |e|^{2-\epsilon}\right) = O\left(d n^{2-\epsilon}\right)
\]
Moreover, by combining the two lemmas above, \probOV has complexity:
\[
O(dn) + T(n,d) = O(dn) + \left(O(nd)+ R(kn,d)\right) = O(dn) + R(kn,d)
\]
We conclude that \probOV has complexity $O(d n^{2-\epsilon})$, contradicting \Cref{conj:OVC}.

%% file: figures/4-edge-hypergraph.tikz
    \begin{tikzpicture}

        \node[vertex] (v1) at (0,0) {};
        \node[vertex] (v2) at (1,0) {};
        \node[vertex] (v3) at (1,1) {};
        \node[vertex] (v4) at (3,0) {};
        \node[vertex] (v5) at (1,-1) {};
        \node[vertex] (v6) at (3,-1) {};
        \node[vertex] (v7) at (-1,-1) {};
        \node[vertex] (v8) at (-0.3,-1.5) {};

        \def\r{0.3}
        \hyperedge[color=violet,tension=0.9]{v3,v1,v2};
        \hyperedge[color=red,tension=0.3]{v6,v4,v2,v5};
        \hyperedge[color=blue,tension=0.3]{v6,v4,v3};
        \hyperedge[color=green,tension=0.3]{v5,v1,v7,v8};
    \end{tikzpicture}

%% file: 4_alphabeta_algorithm.tex
\label{sec4:abnice}
\Cref{sec:cc} shows that, without further assumptions, there is no hope to obtain $O(\size{H}^2)$-time algorithms for HC via color coding.
Thus, we shall look for additional properties of $H$ that can make color coding easier.
One such property is having small rank, say $\rank(H) \le 10$.
This obviously implies $\sum_{e \in E} |e|^2 = O(|E|)$, making the naive algorithm (\Cref{algo:Naive}) run in time linear in $\size{H}$.
However, this assumption is quite strong, and real-world hypergraphs seem far from it---see \Cref{tab:freq}.
Another such property could be having small degree, say $\Delta(H) \le 10$.
It is not immediately clear how this could accelerate color coding; but, again, \Cref{tab:freq} suggests that in practice this property often does not hold.

We bypass these limitations by introducing \emph{$(\alpha,\beta)$-niceness}, an ``interpolation'' between small rank and small degree.
On the one hand, hypergraphs from real-world datasets are $(\alpha,\beta)$-nice for small values of $\alpha,\beta$; on the other hand, we show how to exploit $(\alpha,\beta)$-niceness to bring color coding into near-linear time.

\ignore{

In this section, we present our algorithm: \hypermotivo.

In \Cref{sec:quad_lim}, we showed the existence of a tight lower bound that applies to any approach based on color coding.
To circumvent this limitation from a practical standpoint, we propose a \emph{parameterized} version of \probHC, with respect to which \hypermotivo will be shown to be FPT-linear.

In \Cref{subsec:motiv}, we introduce the notions of $\alpha$-split and $(\alpha, \beta)$-niceness and we motivate this parameterization, analyzing the structure of some real-world hypergraphs and providing empirical evidence of its validity.

Starting from \Cref{fpt:nw}, we gradually introduce the building blocks of \hypermotivo, mainly \emph{another} algorithm for the approximated \probHC, parameterized by the maximum degree of the hypergraph.

Finally, in \Cref{sec:hypermotivo}, we describe \hypermotivo in full and show how it combines all the previously developed ideas, effectively extending the naive method based on \motivo and preserving its worst-case time complexity while - as confirmed by our experiments - achieving significant performance improvements on many instances.

\subsection{Definition and motivation of$(\alpha, \beta)$-niceness}\label{subsec:motiv}
}

\begin{definition}\label{def:alfa_beta_nice}
    Let $H=(V,E)$ be a hypergraph and let $\alpha \ge 0$. The \emph{$\alpha$-split} of $H$ is the pair of hypergraphs ($H_{\leq \alpha},H_{> \alpha})$ where $H_{\leq \alpha} = (V, E_{\leq \alpha})$ and $H_{> \alpha} = (V, E_{> \alpha})$ satisfy:
    \begin{equation}
        \begin{aligned}
        & E_{\leq \alpha} = \{e \in H: |e| \leq \alpha \}\\      
        & E_{> \alpha} = \{e \in H: |e| > \alpha \}\\
        \end{aligned}
    \end{equation}
    A hypergraph $H=(V,E)$ is $(\alpha,\beta)$-nice if $\Delta(H_{> \alpha}) \le \beta$.
\end{definition}
\noindent In other words, $H$ is $(\alpha,\beta)$-nice if it can be edge-partitioned into a hypergraph $H_{\le \alpha}$ with rank at most $\alpha$ and a hypergraph $H_{> \alpha}$ of degree at most $\beta$.
We call $H_{\le \alpha}$ the \emph{lower part} and $H_{>\alpha}$ the \emph{upper part} of the split. Furthermore, we denote by $N_{\leq \alpha}(v)$ and $N_{> \alpha}(v)$ the set of neighbors of $v$ in their respective parts. The same convention is used for the type $E(v)$ as well.

Note that our algorithm \hypermotivo has a build-up time of $f(\alpha,\beta,k) \cdot O(\size{H})$, see \Cref{thm:hypermotivo}, and is thus efficient if $\alpha,\beta,k$ are small.
\Cref{fig:alfa_beta} shows that this is indeed the case in practice.
For every $\alpha=0,1,\ldots$, the figure shows the smallest value of $\beta$ for which a hypergraph is $(\alpha,\beta)$-nice, for all hypergraphs used in our experiments.
Note that, in each curve, the rightmost point corresponds to $\alpha=\rank(H)$ and $\beta=0$, while the topmost point to $\beta=\Delta(H)$ and $\alpha=0$.
These are trivial values, because obviously any $H$ is both $(\rank(H),0)$-nice and $(0,\Delta(H))$-nice.
The interesting point of the curves is that real-world hypergraphs are often $(\alpha,\beta)$-nice for $\alpha,\beta$ \emph{orders of magnitude smaller} than $\rank(H),\Delta(H)$.
\New{This seems to be consistent with the general idea that, in social networks, individuals typically participate in a large number of small communities, but only in a few (if any) large ones~\cite{Leskovec2008CommunitySI}.}
Thus, $(\alpha,\beta)$-niceness is a good candidate for a parameterization of an algorithm's running time.
The dot across each curve denotes the pair $(\alpha,\beta)$ chosen by \hypermotivo, which is typically with $\alpha \simeq 10^3$ and $\beta \simeq 10$.

The next section shows how to leverage $(\alpha,\beta)$-niceness to break through the quadratic barrier of color coding on hypergraphs.

\begin{figure}[t]
    \centering
    \def\myheight{10em}
    \begin{subfigure}[t]{0.49\linewidth}
        \centering
        \includegraphics[height=\myheight,trim={0 0.15cm 0 0},clip]{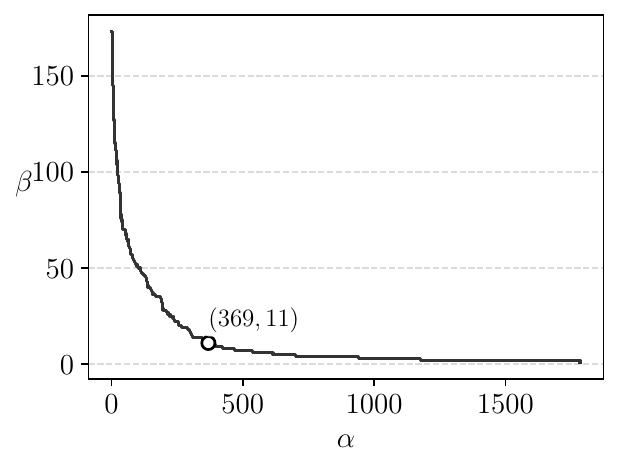}
        \caption{MA}
        \label{fig:alfa_beta_MA}
    \end{subfigure}\hfill
    \begin{subfigure}[t]{0.49\linewidth}
        \centering
        \includegraphics[height=\myheight,trim={0 0.15cm 0 0},clip]{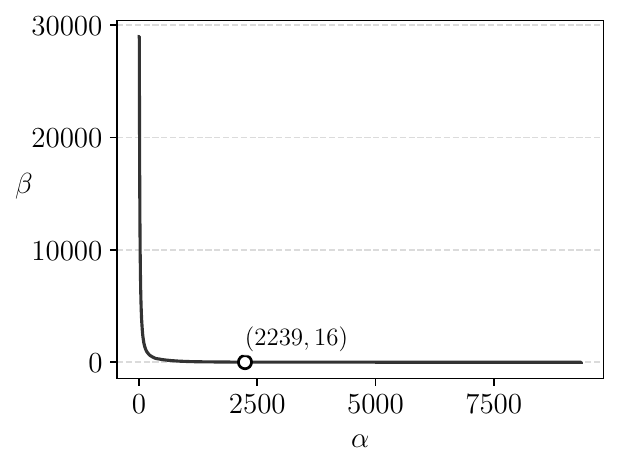}
        \caption{AR}
        \label{fig:alfa_beta_AR}
    \end{subfigure}

    \smallskip

    \begin{subfigure}[t]{0.49\linewidth}
        \centering
        \includegraphics[height=\myheight,trim={0 0.15cm 0 0},clip]{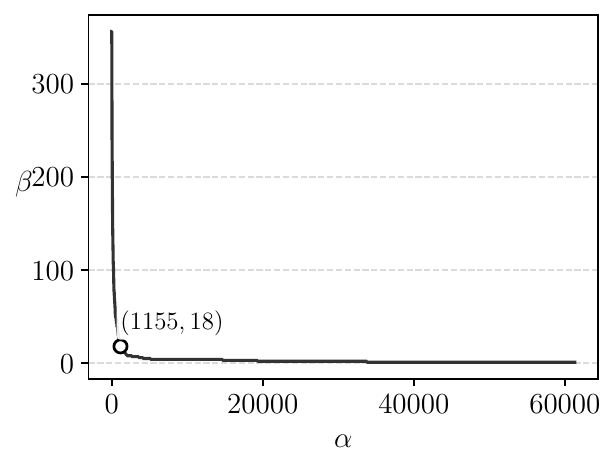}
        \caption{SA}
        \label{fig:alfa_beta_SA}
    \end{subfigure}\hfill
    \begin{subfigure}[t]{0.49\linewidth}
        \centering
        \includegraphics[height=\myheight,trim={0 0.15cm 0 0},clip]{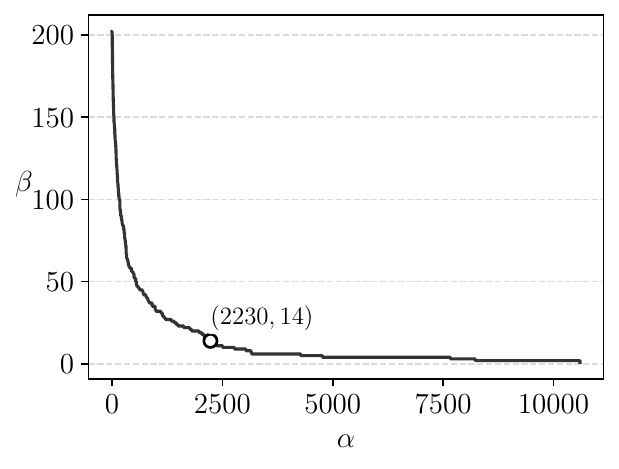}
        \caption{GP}
        \label{fig:alfa_beta_CP}
    \end{subfigure}

    \smallskip

    \begin{subfigure}[t]{0.49\linewidth}
        \centering
        \includegraphics[height=\myheight,trim={0 0.15cm 0 0},clip]{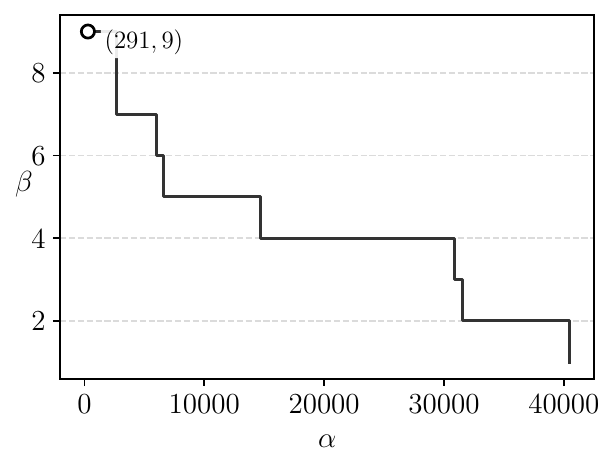}
        \caption{AX}
        \label{fig:alfa_beta_AX}
    \end{subfigure}\hfill
    \begin{subfigure}[t]{0.49\linewidth}
        \centering
        \includegraphics[height=\myheight,trim={0 0.15cm 0 0},clip]{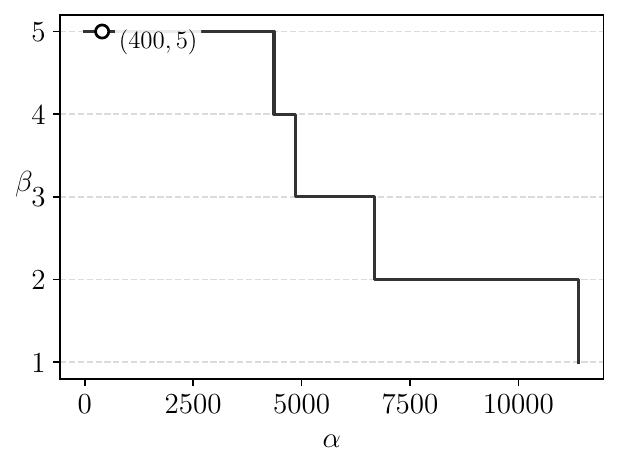}
        \caption{SE}
        \label{fig:alfa_beta_SE}
    \end{subfigure}

    \smallskip

    \begin{subfigure}[b]{\linewidth}
        \centering
        \includegraphics[height=\myheight,trim={0 0.15cm 0 0},clip]{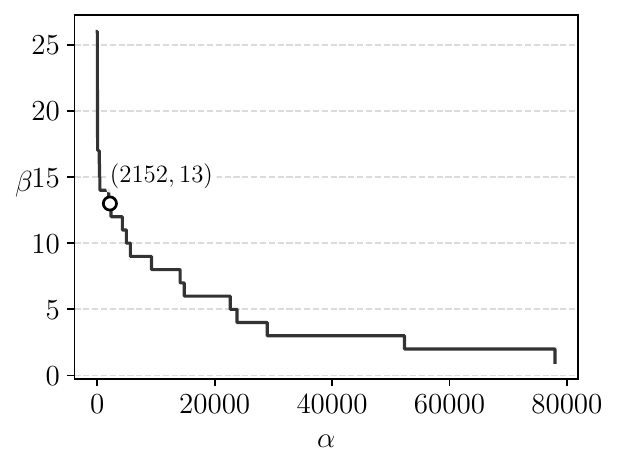}
        \caption{RH}
        \label{fig:alfa_beta_RH}
    \end{subfigure}

    \caption{The $(\alpha,\beta)$-curve of our hypergraphs. For each $\alpha=0,\ldots,\rank(H)$, the curve shows the smallest of $\beta$ such that the hypergraph is $(\alpha,\beta)$-nice. The dot marks the point $(\alpha,\beta)$ chosen by \hypermotivo. See \Cref{sec:experiments_datasets} for more details.}
    \label{fig:alfa_beta}
\end{figure}

\section{\hypermotivo}

\newcommand{\mysize}{\normalsize}

\begin{figure}
    \centering
    \begin{tikzpicture}[
      font=\small,
      >=Latex,
      node distance=4mm and 10mm,
      block/.style={
        draw=black,
        rounded corners,
        thick,
        fill=white,
        align=center,
        text width=5.6cm,
        minimum height=7.5mm,
        inner sep=3pt
      },
      motivo/.style={block, line width=0.9pt},
      hyper/.style={
        block,
        fill=gray!15,
        fill opacity=1,
        text opacity=1,
        execute at begin node=\bfseries
      },
      overlap/.style={
        block,
        line width=1.0pt,
        dashed,
        pattern=dots,
        pattern color=black
      },
      phase/.style={draw=black, dashed, rounded corners, thick, inner sep=6pt},
      lab/.style={font=\bfseries}
    ]
    
    \node[motivo] (color)
    {Choose random coloring $c:V\to[k]$};
    
    \node[hyper, below=of color] (split)
    {Choose $\alpha$ and compute $(H_{\le \alpha},H_{>\alpha})$\\
    \mysize \Cref{sub:splitting}
    };
    
    \node[motivo, below=20pt of split] (build_loop)
    {Split colorful treelet $(T,S)$ into $(T_1,S_1), (T_2, S_2)$};
    
    \node[hyper, below left=4mm and -22mm of build_loop] (build_low)
    {$\eta_{\le \alpha}(v)=\sum_{u\in N_{\le \alpha}(v)}\omega(u)$\\
    via neighbor enumeration on $H_{\le \alpha}$};
    
    \node[hyper, below right=4mm and -22mm of build_loop] (build_high)
    {$\eta_{> \alpha}(v)=\sum_{u\in N_{> \alpha}(v)}\omega(u)$\\
    via inclusion-exclusion on $H_{>\alpha}$};
    
    \coordinate (midB) at ($(build_low.south)!0.5!(build_high.south)$);
    
    \node[hyper] (overlapB) at ($(midB)+(0,-.65)$)
    {Overlap correction};
    
    \node[motivo, text width=8.6cm, below=18pt of overlapB] (dp)
    {Update counters\\
    $C(T,S,v)=C(T,S,v)+C(T_1,S_1,v)\cdot\eta(v)$};
    
    \node[motivo, below=18mm of dp] (untilacc)
    {Until $u$ is accepted};
    
    \node[motivo, below left=4mm and -22mm of untilacc] (samplow)
    {Sample $u$ from $N_{\le \alpha}(v)$ on $G_{\le \alpha}$};
    
    \node[hyper, below right=4mm and -22mm of untilacc] (samphigh)
    {Sample $u$ via hyperedges in $H_{> \alpha}$};
    
    \coordinate (midS) at ($(samplow.south)!0.5!(samphigh.south)$);
    
    \node[hyper] (overlapS) at ($(midS)+(0,-.65)$)
    {Accept/Reject $u$};
    
    \draw[->, thick, dashed]
      (overlapS.north) to
      node[midway, fill=white, inner sep=1pt] {\mysize reject}
      (untilacc.south);
    
    \node[hyper, text width=8.6cm, below=18pt of overlapS] (induced)
    {Compute $H[U]$ from sampled set $U$\\
    \mysize \Cref{lemma:subhypergraphlet_build}};
    
    \node[motivo, below=of induced] (final)
    {Update $H[U]$ counter};
    
    \coordinate (inTop) at ($(color.north)+(0,8mm)$);
    \node (input) at (inTop) {$H=(V,E)$};
    \draw[->, thick] ($(inTop)+(0,-0.2)$) -- (color.north);
    
    \draw[->, thick] (color) -- (split);
    \draw[->, thick] (split) -- (build_loop);
    
    \draw[->, thick] (build_loop.west) -| (build_low.north);
    \draw[->, thick] (build_loop.east) -| (build_high.north);
    
    \draw[->, thick] (build_low) |- node[midway, left] {\scriptsize $\eta_{\leq \alpha}$} (overlapB.west);
    \draw[->, thick] (build_high) |- node[midway, right] {\scriptsize $\eta_{> \alpha}$} (overlapB.east);
    
    \draw[->, thick] (overlapB) -- node[midway, left] {\scriptsize $\eta$} (dp);
    
    \draw[->, thick] (dp) -- (untilacc);
    
    \draw[->, thick] (untilacc.west) -| (samplow.north);
    \draw[->, thick] (untilacc.east) -| (samphigh.north);
    
    \draw[->, thick] (samplow) |- (overlapS.west);
    \draw[->, thick] (samphigh) |- (overlapS.east);
    
    \draw[->, thick] (overlapS) -- node[midway, right=5pt] {\scriptsize accept} (induced);
    \draw[->, thick] (induced) -- (final);
    
    \node[
      phase,
      fit=(build_loop)(build_low)(build_high)(overlapB)(dp),
      label={[anchor=west,xshift=-24mm,yshift=2mm,fill=white,inner sep=1pt]above:
      \footnotesize BUILD-UP (\Cref{subsec:hypermotivo-buildup}, \Cref{alg:alpha-build})}
    ] (buildbox) {};
    
    \coordinate (sampPadTop) at ($(untilacc.north)+(0,7mm)$);
    \node[
      phase,
      fit=(sampPadTop)(untilacc)(samplow)(samphigh)(overlapS)(induced)(final),
      label={[anchor=west,xshift=-24mm,yshift=2mm,fill=white,inner sep=1pt]above:
      \footnotesize SAMPLING (\Cref{subsec:hypermotivo-sample}, \Cref{algo:sample})}
    ] (sampbox) {};
    
    \node[
      draw=black,
      rounded corners,
      densely dotted,
      line width=0.8pt,
      inner sep=2pt,
      fit=(untilacc)(samplow)(samphigh)(overlapS),
      label={[lab,anchor=west,xshift=-30mm,yshift=2mm,fill=white,inner sep=1pt]above:
      \footnotesize Sample rooted $k$-colorful treelet}
    ] (subsamplebox) {};
    
    \end{tikzpicture}
    \caption{Structure of \hypermotivo. White blocks are in common with \motivo; gray ones are introduced in this work.}
    \label{fig:hypermotivo_diagram}
\end{figure}
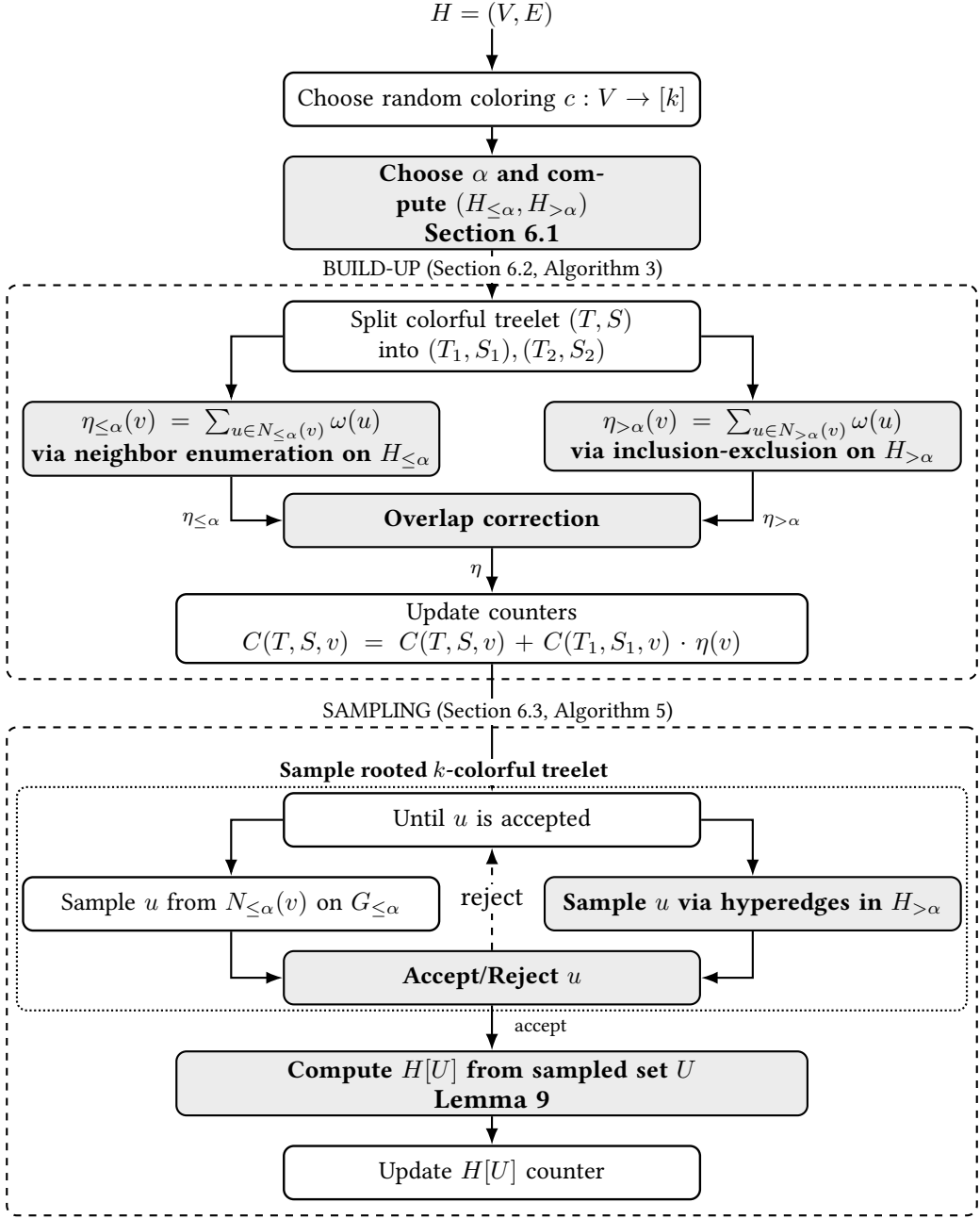

We describe \hypermotivo, our algorithm for sampling and approximately counting $k$-hypergraphlets through color coding.
The main idea behind \hypermotivo is to exploit the $(\alpha,\beta)$-niceness of a hypergraph (\Cref{sec4:abnice}) to bypass the quadratic barrier of the naive color-coding algorithm (\Cref{sec:quad_lim}).
Our main result is:
\begin{theorem}\label{thm:hypermotivo}
There exists a two-phase algorithm, \hypermotivo, with the following guarantees.
Given in input an integer $k \ge 2$, a value $\alpha \ge 0$, and an $(\alpha,\beta)$-nice hypergraph $H=(V,E)$, 
\begin{itemize}[leftmargin=1.5em]
    \item the \emph{build-up phase} takes time 
    \[2^{O(k)}\cdot O\left(2^\beta |V| + \alpha^k |E| + \alpha^2 \beta \size{H}\right)\,;\]
    \item the \emph{sampling phase} takes, on each invocation, expected time
    \[k^{O(k)} \cdot (\beta^2 + \ln{|V|})\]
    and yields a colorful $k$-hypergraphlet $H[U]$  uniformly at random.
\end{itemize}
\end{theorem}
\Giacomo{A high-level overview of the algorithm is shown in
\Cref{fig:hypermotivo_diagram}.}
Note that, instead of sampling $H[U]$ uniformly at random, one can sample it with probability proportional to the number of its spanning trees $\sigma(H[U])$, while also computing $\sigma(H[U])$ itself.
The $k^{O(k)}$ factor in the sampling time decreases to $2^{O(k)}$, and one can estimate the hypergraphlets counts in the same way as \motivo, by adding $\frac{1}{\sigma(H[U])}$ to the count of hypergraphlets isomorphic to $H[U]$.
The same statistical guarantees of \motivo apply, including concentration bounds that depend on the number of $k$-hypergraphlets of $H$; see~\cite{Bressan&2018b}.
Note that, in our implementation, we avoid the $\alpha^k |E|$ part of the build-up phase.
In theory, this makes the sampling phase slower; in practice, it is still almost as fast as \motivo's.
See \Cref{sec:experiments}.
Finally, let us observe that \Cref{thm:hypermotivo} is obtained by a very careful combination of several parts---from computing counters, to sampling neighbors, to computing $H[U]$---each of which exploits $(\alpha,\beta)$-niceness in a very specific way.
We find this quite interesting.

The remainder of the section is organized as follows. \Cref{sub:splitting} describes how \hypermotivo computes an $\alpha$-split of the input hypergraph. \Cref{subsec:hypermotivo-buildup} describes the build-up phase of \hypermotivo, and \Cref{subsec:hypermotivo-sample} describes its sampling phase.

\subsection{Splitting the hypergraph}\label{sub:splitting}
This section discusses the first step of \hypermotivo: given the hypergraph $H$, compute an $\alpha$-split $(H_{\leq\alpha}, H_{>\alpha})$.
\New{In fact, we give a $O(\size{H})$-time algorithm that finds a ``good'' value of $\alpha$, computes the corresponding split $(H_{\leq\alpha}, H_{>\alpha})$, and moreover provides an asymptotic estimate of the number of operations performed by our algorithm when given in input $(H_{\leq\alpha}, H_{>\alpha})$.
In~\cref{sec:experiments} we also show that this choice of $\alpha$ is rather robust to variations in $\alpha$ itself (and, therefore, our method to choose $\alpha$ does not ``overfit'').}

It is straightforward that, given $\alpha$, one can compute $(H_{\leq\alpha}, H_{>\alpha})$ in time $O(\size{H})$ by simply placing each edge $e$ of $H$ in $H_{\leq\alpha}$ if $|e|\le \alpha$ and in $H_{>\alpha}$ otherwise.
We shall thus focus on the choice of $\alpha$.
Roughly speaking, we seek to choose $\alpha$ so as to minimize the running time of \hypermotivo's build-up phase.
That running time is dominated by the computation of the counters $C(\cdot,\cdot,\cdot)$, which is carried out independently, and by two distinct algorithms, on $H_{\leq\alpha}$ and on $H_{>\alpha}$.
As explained later in \Cref{subsec:hypermotivo-buildup}, the times $T_{\le \alpha}$ and $T_{> \alpha}$ spent by the two algorithms can be bound as follows (we use $x \lesssim y$ to denote $x = O(y)$):
\begin{align}
    T_{\le \alpha} &\lesssim \sum_{e \in E_{\le \alpha}} |e|^2 \\
    T_{> \alpha} &\lesssim \sum_{v \in V} 2^{d_{> \alpha}(v)}
\end{align}
If $H$ is $(\alpha,\beta)$-nice, it follows immediately that:
\begin{align}\label{eq:T_bound}
    T_{\le \alpha} + T_{> \alpha} \lesssim \alpha^2 |E| + 2^{\beta} |V|
\end{align}
Therefore, one sensible way to choose $\alpha$ is to (1) compute all pairs $(\alpha,\beta)$ such that $H$ is $(\alpha,\beta)$-nice (i.e., the curve shown in \Cref{sec4:abnice}), and (2) choose the pair that minimizes the right-hand side of \Cref{eq:T_bound}.
To this end we proceed as follows.
First, we sort the edges of $H$ in nondecreasing order of size, $e_1,\ldots,e_{|E|}$.
We also let $\alpha_0 = 0$ and $\beta_0 = \Delta$, where $\Delta$ is the maximum degree of $H$.
Clearly, $H$ is $(\alpha_0,\beta_0)$-nice.
Then, for $i=1,\ldots,|E|$, we let $\alpha_i=|e_i|$, and we let $\beta_i$ be the maximum degree of $H_{> \alpha_i}$.
Note that, to compute $\beta_i$, it is sufficient to keep track of the degrees in $H_{> \alpha_i}$, and take their maximum through a max-heap.
To update those degrees from $H_{> \alpha_{i-1}}$ to $H_{> \alpha_i}$, one simply decreases by $1$ the degree of every vertex in $e_i$.
This yields an update time of $O(|e_i| \log |V|)$ for computing $\alpha_i,\beta_i$ from $\alpha_{i-1},\beta_{i-1}$.
Once we have the complete curve $(\alpha_i,\beta_i)_{i=0,\ldots,|E|}$, it is straightforward to choose $i$ so that $(\alpha_i,\beta_i)$ minimizes the right-hand side of \Cref{eq:T_bound}.
Overall, we have proved:
\begin{lemma}\label{lem:compute_split_1}
Given $H$, in time $O(\size{H} \log |V|)$ one can compute a pair $(\alpha,\beta)$, and the corresponding $\alpha$-split $(H_{\leq\alpha}, H_{>\alpha})$, such that $H$ is $(\alpha,\beta)$-nice and the right-hand side of \Cref{eq:T_bound} is minimized.
\end{lemma}
The simple algorithm behind \Cref{lem:compute_split_1} can actually be refined in order to compute a pair $(\alpha,\beta)$ that minimizes the more accurate estimate of the running time:
\begin{align}\label{eq:T_bound_2}
    T_{\le \alpha} + T_{> \alpha} \lesssim \sum_{e \in E_{\le \alpha}} |e|^2+\sum_{v \in V} 2^{d_{> \alpha}(v)}
\end{align}
Indeed, we can prove:
\begin{lemma}\label{lem:compute_split_2}
Given $H$, in time $O(\size{H})$ one can compute a pair $(\alpha,\beta)$, and the corresponding $\alpha$-split $(H_{\leq\alpha}, H_{>\alpha})$, such that $H$ is $(\alpha,\beta)$-nice and the right-hand side of \Cref{eq:T_bound_2} is minimized.
\end{lemma}
\noindent The algorithm of \Cref{lem:compute_split_2} is the one actually used by \hypermotivo.
The rest of this section describes this algorithm, proving \Cref{lem:compute_split_2}. At a high level, the procedure evaluates all candidate thresholds $s_i$ by sweeping them from larger to smaller values and maintaining on the fly the quantities appearing in the right-hand side of \Cref{eq:T_bound_2}. After a linear-time preprocessing phase, each new threshold is handled by a constant amount of work, so that the best $(\alpha,\beta)$-split can be found in overall $O(\size{H})$ time.

We begin by constructing, for every vertex $v \in V$, the list $I_v$ of hyperedges incident with $v$, sorted in non-increasing order of hyperedge size.
Thus $I_v[0]$ is the largest hyperedge incident with $v$, $I_v[1]$ is the second largest, and so on. For each vertex–hyperedge incidence $(v,e)$, we let $s = |e|$ be the size of $e$ and let $p$ be the position of $e$ in $I_v$ (with $p=0$ for $I_v[0]$). We then create a triple $(v,s,p)$ for every incidence and collect all such triples into a single list $S$. Next, we sort $S$ by non-increasing $s$ and, for equal $s$, by non-increasing $p$, and process them in this order. Let $s_1 < s_2 < \dots < s_{m}$ be the distinct hyperedge sizes appearing in $H$. For each $i \in [m]$, we denote by $S_i$ the (maximal) contiguous sublist of $S$ consisting of all triples $(v,s,p)$ with $s = s_i$. Since $S$ is sorted by non-increasing $s$, these blocks appear in $S$ in the order $S_m, \dots, S_1$. We now show how, by sweeping through all the blocks $S_i$ in this order, we can compute the right-hand side of \Cref{eq:T_bound_2} for every candidate threshold $s_i$ and its corresponding $s_i$-split $(H_{\leq s_i}, H_{> s_i})$.

We initialize the procedure at the largest candidate value, $s_m$. In this configuration, all hyperedges belong to the lower part, so
$T_{\le s_m} = \sum_{e \in E}|e|^2$, while $T_{> s_m} = 0$. Moreover, we precompute, for each $i \in [m]$, the number $c_i = |\{e \in E : |e| = s_i\}|$. Now, by iterating on $S_m, \dots, S_1$ in order, it is easy to see that after processing block $S_m$ we are at threshold $s_m$, and for every $i = m-1, \dots, 1$ we can compute $T_{\le s_i}$ from $T_{\le s_{i-1}}$ as $T_{\le s_i} = T_{\leq s_{i-1}} - c_i \cdot s_i^2$. For the upper part, fix $i \in [m]$ and $v \in V$. Among all incidences $(v,e)$ with $|e| = s_i$, let $p_i(v)$ be the maximum position of such an edge in $I_v$. The key observation is that, since $I_v$ is sorted by non-increasing edge size, the entries $I_v[0], \dots, I_v[p_i(v)]$ are exactly the hyperedges incident to $v$ that belong to the upper part at threshold $s_i$, and therefore $d_{> s_i}(v) = p_i(v) + 1$. Within block $S_i$, the triples $(v,s_i,p)$ corresponding to a fixed vertex $v$ appear in non-increasing order of $p$, and hence the first such triple we encounter has $p = p_i(v)$. Thus, to update $T_{> s_i}$ we just need to maintain for each vertex $v$ its current degree. When we process the first triple $(v,s_i,p)$ for vertex $v$ in $S_i$, we set its degree to $p+1$ and update its contribution in the sum accordingly.

At the end of this process, we have computed the right-hand side of \Cref{eq:T_bound_2} for every candidate threshold $s_i$. It is trivial to see that we can also keep track of the $s_i$ that minimizes the equation. At the same time, for each $s_i$ we maintain the maximum degree $\beta(s_i) = \max_{v \in V} d_{> s_i}(v) = \max_{v \in V} p_i(v) + 1$. Let $\alpha^\star$ be the threshold chosen by this procedure, and let $\beta^\star = \beta(\alpha^\star)$. By construction, the corresponding split $(H_{\le \alpha^\star}, H_{> \alpha^\star})$ satisfies $\Delta(H_{> \alpha^\star}) \le \beta^\star$, i.e., $H$ is $(\alpha^\star,\beta^\star)$-nice, and the quantity in \Cref{eq:T_bound_2} is minimized at $(\alpha^\star,\beta^\star)$. Since sorting $S$ takes $O(\size{H})$ time via counting sort, and the sweep as well, this proves \Cref{lem:compute_split_2}.

\ignore{
In our implementation, we actually minimize a weighted variant of this cost function to better balance the contributions of the lower and upper parts; this choice is described more in detail in \Cref{sec:finding_alpha}.
}

\subsection{The build-up phase}\label{subsec:hypermotivo-buildup}
The role of the build-up phase is, given $H$ and its coloring, to compute the counters $\C(T,S,v)$ of \Cref{eq:treelet_counter}.
By \Cref{thm:quadratic}, there is little hope of doing so in sub-quadratic time \emph{as a function solely of} $|V|$ and $|e|$ for $e \in E$.
The goal of this section is to show how, if $H$ is $(\alpha,\beta)$-nice, then we can perform the build-up phase in time near-linear in $\size{H}$, times a function of $\alpha,\beta$.
Let $(H_{\le \alpha},H_{> \alpha})$ be the $\alpha$-split of $H$; see \Cref{sub:splitting}.
The idea is to compute counters for $H_{\le \alpha}$ and $H_{> \alpha}$ separately, and postprocess them in order to get the counters of $H$.
To this end, observe how, by \Cref{eq:treelet_counter}, $\C(T,S,v)$ can be written as a sum $S_1,S_2,\ldots$ of terms in the form:
\begin{align}
    \C(T_1,S_1,v) \cdot \C(T_2,S_2,N(v))
\end{align}
where $N(v)$ is the set of vertices adjacent to $v$ in $H$, and
\begin{align}
    \C(T_2,S_2,N(v)) = \sum_{u \in N(v)} \C(T_2,S_2,u) \,.
\end{align}
It is immediate to see that, if one can compute $\C(T_2,S_2,N(v))$ in time $t$, then one can compute all counters $\C(T,S,v)$ in time $2^{O(k)} \cdot t$.

By abstracting details away, one reduces to the following problem.
A \emph{weighted hypergraph} is a hypergraph $H=(V,E)$ with a vertex weighting $w : V \to \N_0$.
Extend $w$ to subsets of $V$ in the natural way.
The \emph{neighbor weight} of a vertex $v$ is $\eta(v) = w(N(v))$.
The \emph{Neighbor Weight Problem} (\probNW) asks, given $(H,w)$, to compute $(\eta(v))_{v \in V}$.
Clearly, computing $\C(T_2,S_2,N(v))$ for all $v \in V$ reduces, in time $O(\size{H})$, to \probNW on $(H,w)$ with weights $w(v)=\C(T_2,S_2,v)$.
To prove the first item of \Cref{thm:hypermotivo}, we shall thus prove how one can solve \probNW in the time claimed by that item.

\ignore{
\hypermotivo  of $C(T,[k],v)$ colorful $k$-treelets that are isomorphic to $T$ and rooted in $v$. Also, the underlying mechanism is based upon the same dynamic programming approach and \Cref{eq:treelet_counter}.
However, \hypermotivo processes separately the hypergraphs $H_{\le \alpha}$ and $H_{> \alpha}$ of the $\alpha$-split, recovering the correct counters for $H$ at the end.
In the lower part, where the size of the hyperedges is bounded by $\alpha$, \hypermotivo simply executes the naive algorithm (see \Cref{algo:Naive}). 
In the upper part, where the maximum degree is bounded by $\beta$, it instead executes a procedure based on the inclusion–exclusion principle.
Before introducing the algorithm for the build-up phase, we briefly discuss and motivate this technique. 
}

\subsubsection{Solving \probNW on $(\alpha,\beta)$-nice hypergraphs}
Let $(H,w)$ be a \allowbreak \ weighted hypergraph, let $(H_{\le\alpha},H_{>\alpha})$ be the $\alpha$-split of $H$, and suppose $H$ is $(\alpha,\beta)$-nice.
We shall adopt the following strategy.
First, we solve \probNW on the weighted hypergraph $(H_{\le\alpha},w)$ and obtain the neighbor weights $\eta_{\le\alpha}(v)$.
Second, we solve \probNW on the weighted hypergraph $(H_{>\alpha},w)$ and obtain the neighbor weights $\eta_{>\alpha}(v)$.
Third, for each vertex we sum the counters and obtain 
\begin{align}
    \tilde\eta(v) = \eta_{\le\alpha}(v)+\eta_{>\alpha}(v) = \sum_{u \in N_{\le\alpha}(v)}w(u) + \sum_{u \in N_{>\alpha}(v)}w(u)
\end{align}
Clearly we may have $\eta(v) < \tilde\eta(v)$, since a vertex $u \in N(v)$ can appear in both $N_{\le\alpha}(v)$ and $N_{>\alpha}(v)$.
The third step thus entails subtracting $w(u)$ from $\tilde\eta(v)$ for all those $u$ to retrieve $\eta(v)$.
The rest of the subsection describes the three steps above.

\subsubsection{Solving \probNW on $H_{\leq \alpha}$}
For $H_{\le \alpha}$, we simply compute $G_{\le \alpha} = \Gaifman(H_{\le\alpha})$, and then compute $\eta_{\le\alpha}(v)$ by explicitly enumerating the neighborhood $N_{\le\alpha}(v)$ of $v$ in $G_{\le \alpha}$.
As $\rank(H_{\le\alpha}) \le \alpha$, computing $G_{\le \alpha}$ takes time $O(|V|+\alpha^2 |E_{\le \alpha}|) \le O(|V|+\alpha^2 |E|)$.
Enumerating the neighborhoods obviously requires linear time.
This gives a total time of $O(|V|+\alpha^2 |E|)$.

\subsubsection{Solving \probNW on $H_{>\alpha}$}
To ease the notation we describe the algorithm on a generic weighted hypergraph $(H,w)$.
The first observation is that, obviously, every $v \in V$ satisfies:
\begin{align}
    \eta(v) = w\Big(\bigcup E(v)\Big) - w(v)\,.
\end{align}
Hence, solving \probNW reduces to computing $w(\bigcup E(v))$ for each $v$. 
As $|E(v)| \le \beta$, we resort to compute $w(\bigcup E(v))$ via \emph{inclusion-exclusion}, as follows.
Let $E(v)=\{e_1, \dots, e_\ell\}$. Then, the inclusion-exclusion principle says that:
\begin{equation}\label{eq:weighted_incl_excl}
    w\Big(\bigcup E(v)\Big) = w\Big(\bigcup_{i \in [\ell]} e_i\Big) = \sum_{\emptyset \neq I \subseteq [\ell]}(-1)^{|I| + 1} \,w\Big(\bigcap_{i \in I} e_i\Big)\,.
\end{equation}
Thus, it suffices to (1) compute $w\left(\bigcap_{i \in I} e_i\right)$ for all $I \subseteq [\ell]$ for all $E(v)$, and (2) apply \Cref{eq:weighted_incl_excl}.
This is what \Cref{algo:NW-IE} does.

\begin{algorithm}[h]
\LinesNumbered
\caption{\algoNWIE}
\label{algo:NW-IE}
{\fontsizealgo
    \Input{A weighted hypergraph $(H=(V, E),w)$}
    \Output{The vector $\eta$ of neighbor weights}
    \BlankLine
    $t \gets$ empty dictionary with default value $0$\;
    \ForEach{$v \in V$}{
        \ForEach{$X \subseteq E(v)$}{
            $t[X] \gets t[X] + w(v)$\; \nllabel{algo:nw:t_increm}
        }
    }
    $\eta \gets$ empty dictionary with default value $0$\; \nllabel{algo:nw:eta_init}
    \ForEach{$v \in V$}{
        \BlankLine
        \ForEach{$\emptyset \ne X \subseteq E(v)$}{
            $\eta(v) \gets \eta(v) + (-1)^{|X|+1} t[X]$\; \nllabel{algo:nw:eta_increm}
        }
        $\eta(v) \gets \eta(v) - w(v)$\;
    }
    \Return $\eta$\;
}
\end{algorithm}


We prove:
\begin{lemma}\label{lemma:nwiecomplexity}\algoNWIE solves \probNW in time $O(2^{\Delta(H)}\cdot |V|)$.
\end{lemma}
\begin{proof}
 Fix $v \in V$ and $X \subseteq E(v)$.
 Note that \Cref{algo:nw:t_increm} increments $t[X]$ by $w(v)$ if and only if $v \in \cap_{e \in C} e$.
 Therefore, at \Cref{algo:nw:eta_init} we have:
     \begin{equation}\label{eq:phase1}
         t[X] = \sum_{v \in \bigcap X} w(v) =w\Big(\bigcap X\Big) \,.
     \end{equation}
 By construction, then, the second loop ensures that for any $v \in V$:
     \begin{align}
         \eta(v) &= \sum_{\emptyset \ne X \subseteq E(v)} (-1)^{|X|+1} t[X] - w(v)
        \\ &= \sum_{\emptyset \ne X \subseteq E(v)} (-1)^{|X|+1} w\left(\bigcap X\right) - w(v) && \text{by \cref{eq:phase1}}
         \\ &= w\left(\bigcup E(v)\right) - w(v) && \text{by \cref{eq:weighted_incl_excl}}
     \end{align}
 For the time complexity, for every $v \in V$ the algorithm performs at most $2^{|E(v)|} \le 2^{\Delta(H)}$ operations, each of which takes time $O(1)$.
\end{proof}
To conclude, since $\Delta(H_{>\alpha}) \leq \beta$, applying \algoNWIE to $(H_{>\alpha},w)$ yields the counters $\eta_{>\alpha}(v)$ in time $O(2^\beta\cdot |V|)$.

\subsubsection{Correcting $\tilde\eta(v)$}
The third and final step in the build-up phase is retrieving the correct neighbor weights $\eta$ from $\eta_{\le\alpha}$ and $\eta_{>\alpha}$.
To this end, for each $v \in V$ we iterate over $u \in N_{\le \alpha}(v)$ and check if $u \in N_{> \alpha}(v)$ too.
If this is the case, then we subtract $w(u)$ from $\tilde\eta(v)$.
Iterating over $v$ and $u$ takes time $O(\alpha^2 |E|)$ by using $\Gaifman(H_{\le\alpha})$.
For each such iteration, checking $u \in N_{>\alpha}(v)$ takes time $O(\beta)$ by intersecting the list of edges (i.e., the types) of $u$ and $v$ in $H_{> \alpha}$.
This requires those types to be sorted, which one can do beforehand in time:
\begin{align}
    O\left(\sum_{v \in V} d_{>\alpha}(v) \ln d_{>\alpha}(v)\right) &= O\left(\size{H_{>\alpha}} \ln \beta\right) = O\left(\size{H}\beta\right)\,.
\end{align}
Hence, computing $\eta$ from $\eta_{\le\alpha}$ and $\eta_{>\alpha}$ takes total time $O(\alpha^2 \beta \size{H})$.

\subsubsection{Wrapping up}
\Cref{alg:alpha-build} gives the pseudocode for computing the counters $\C(\cdot)$.
Note that this is not \emph{all} the build-up phase: in order to have an efficient sampling phase, we need some additional preprocessing, detailed in \Cref{subsec:hypermotivo-sample}.
In particular, the $O(\alpha^k |E|)$ term appearing in the bounds of \Cref{thm:hypermotivo} is due to \Cref{lemma:subhypergraphlet_build}.
The other two terms are instead due to \Cref{alg:alpha-build}:
\begin{lemma} \label{lemma: complexity_build_hypermotivo}
     If $H$ is $(\alpha, \beta)-$nice, \Cref{alg:alpha-build} computes the counters of \Cref{eq:treelet_counter} in time  $2^{O(k)}\cdot \left(2^\beta |V| + \alpha^2 \beta \size{H}\right)$.
\end{lemma}
 \begin{proof}
 The correctness follows from the discussion above, which shows that at \cref{algo:build:C_update} we have $\eta(v)=\C(T_2,S_2,v)$, from \cref{eq:treelet_counter}, and by construction of the algorithm. Let us bound the running time.

 By \Cref{lem:compute_split_2}, \cref{algo:build:split} takes time $O(\size{H})$.
 By the discussion above, \cref{algo:build:Gaif_low} and \cref{algo:build:sort_Ev} take time $O(\alpha^2 \beta \size{H})$.
 The same discussion showed that each execution of the loop of \cref{algo:build:for_v} takes time take time $O(\alpha^2 \beta \size{H})$ as well, while \cref{algo:build:w}, \cref{algo:build:etalow}, and \cref{algo:build:etaup} take time respectively $O(|V|)$, $O(\alpha^2 |E|)$, and $O(2^\beta |V|)$.
 Thus, every iteration of the loop of \cref{algo:build:main_for} takes time
 \begin{align}
     O\left(2^\beta |V| + \alpha^2 |E| + \alpha^2 \beta \size{H}\right) = O\left(2^\beta |V| + \alpha^2 \beta \size{H}\right) \,.
 \end{align}
 which already dominates the running time of the rest.
 Noting that the said loop makes $2^{O(k)}$ iterations completes the proof.
 \end{proof}

\begin{algorithm}[h]
\caption{\hypermotivo build-up}
\label{alg:alpha-build}
\LinesNumbered
{\fontsizealgo
    \Input{A hypergraph $H=(V, E)$ and $k\geq 2$}
    \Output{The counters $\C(T,S,v)$ of \Cref{eq:treelet_counter}}
    \BlankLine
    $\C \gets $ empty dictionary with default value $0$\;
    \ForEach{$v \in V$}{
        $c(v) \gets $ random color from $[k]$\;\
        $\C(T, \{c(v)\},v) \gets 1$ where $T$ is the trivial tree on one vertex\;
    }
    compute the $\alpha$-split $(H_{\leq\alpha}, H_{>\alpha})$ of $H$ as per \Cref{lem:compute_split_2}\nllabel{algo:build:split}\;
    compute $G_{\le\alpha} = \Gaifman(H_{\le\alpha})$\nllabel{algo:build:Gaif_low}\;
    sort the types of each $v \in V$ in $H_{>\alpha}$\nllabel{algo:build:sort_Ev}\;
    \For{every $h = 2,\ldots,k$, every $h$-treelet $T$, and every $S \in \binom{[k]}{h}$}{
        let $T_1, T_2$ be the canonical decomposition of $T$\;
        \ForEach{partition $S_1,S_2$ of $S$ with $|S_1|=|T_1|$\nllabel{algo:build:main_for}}{
            let $w : V \to \N$ be given by $w(v)=\C(T_2,S_2,v)$\nllabel{algo:build:w}\;
            $\eta_{\le\alpha} \leftarrow \algoNWNaive(H_{\le\alpha},w)$\nllabel{algo:build:etalow}\;
            $\eta_{>\alpha} \leftarrow \algoNWIE(H_{>\alpha},w)$\nllabel{algo:build:etaup}\;
            \ForEach{$v \in V$\nllabel{algo:build:for_v}}{
                $\eta(v) \leftarrow \eta_{\le\alpha}(v)+\eta_{>\alpha}(v)$\label{algo:build:etasum}\;
                \ForEach{$u \in N_{\le\alpha}(v)$\nllabel{algo:build:for_u}}{
                    \If{$u \in N_{> \alpha}(v)$\nllabel{algo:build:u_v_check}}{
                        $\eta(v) \leftarrow \eta(v) - w(u)$\nllabel{algo:build:etacorrect}\;
                   }
                }
                $\C(T, S, v) \gets \C(T, S, v) + d^{-1}\C(T_1, S_1, v) \cdot \eta(v)$\nllabel{algo:build:C_update}\;
            }
        }
    }
    \Return $\C$ 
}
\end{algorithm}

\subsection{The sampling phase}\label{subsec:hypermotivo-sample}
The goal of the sampling phase is to sample colorful, uniform $k$-hypergraphlets of $H$.
To this end, we follow the same high-level strategy of \motivo.
To begin with, we need an efficient \emph{neighbor-sampling} routine.
Given a vertex $v$, a set of colors $S_2 \subseteq [k]$, and a treelet $T_2$ with $|T_2|=|S_2|$, we need to draw $u \in N(v)$ with probability proportional to $\C(T_2, S_2, u)$ (see \Cref{motivo_sample}), while avoiding the explicit visit of $N(v)$.
Once we have this routine, we can then implement an efficient procedure for sampling $k$-colorful rooted treelets uniformly at random from $H$.
At this point, we can use the sampled $k$-colorful rooted treelets to sample $k$-hypergraphlets, by following standard rejection sampling techniques.
While at a high level this may look easy, the nontrivial challenge is to make all these routines run in time nearly independent of $\size{H}$.
Somewhat surprisingly, we will show how to do so by repeatedly exploiting the $(\alpha,\beta)$-niceness of $H$.


For the remainder of the section, we assume that the projection $\Gaif(H_{\le \alpha})$ has already been computed, and, for each $v \in H$, $N_{\le \alpha}(v)$ and $E_{> \alpha}(v)$ have been sorted. Such computations can be carried out during the build-up phase without increasing its running time.

\subsubsection{Sampling a neighbor}
The subroutine for sampling a neighbor is presented in \Cref{sample_neigh}. 
For the purpose of readability we simplify the notation as follows.
Recall that, given a $k$-treelet $T$ with $k \ge 2$, there is a canonical decomposition of $T$ into smaller trees $T_1$ and $T_2$; see \Cref{sec:prelims}.
For any subset $S_1$ of colors, and any vertex $v$, we then write $\C(S_1, v)$ for the counter $\C(T_1,S_1,v)$, and $\C(S_2, v)$ for the counter $\C(T_2,S_2,v)$.
Similarly, for an edge $e \in E$, we write $\C(S_1,e)$ for $\sum_{u \in e}\C(T_1,S_1,e)$, and so on.
Finally, let $W_{\le \alpha}(v) = \sum_{u \in N_{\le \alpha}(v)}\C(S_2,u)$ and $W_{> \alpha}(v) = \sum_{u \in N_{> \alpha}(v)}\C(S_2,u)$.

\begin{algorithm}[h]
\caption{\SampleNeigh}
\label{sample_neigh}
\LinesNumbered
        {\fontsizealgo
         \Input{$(T,S,v)$}
         \Output{$(T_2,S_2,u)$}
         let $T_1,T_2$ be the canonical decomposition of $T$\;
         choose $S_1,S_2$ w.p.\ $\propto \C(S_1,v) \cdot \sum_{u \in N(v)}\C(S_2,u)$\; \nllabel{sample:line:labels}
         let $p(N_{\le \alpha}(v)) := \frac{W_{\le \alpha}(v)}{W_{\le \alpha}(v) + W_{> \alpha}(v)}$\;
        \While{$u$ is not accepted\nllabel{sample:line:outer_while}}{
        \With{probability $p(N_{\le \alpha}(v))$}{
        choose $u \in N_{\le \alpha}(v)$ with probability $\frac{\C(S_2,u)}{W_{\le \alpha}(v)}$ \nllabel{sample:line:lower_draw}\;
        }
        \Else{\nllabel{sample:line:else}
         \While{$u$ is not accepted\nllabel{sample:line:inner_while}}{
        choose $e \in E(v)$ with probability $\frac{\C(S_2, e)}{\sum_{e' \in E(v)}
        \C(S_2,e')}$\nllabel{sample:line:upper_draw_edge}\;
         choose $u \in e$  with probability $\frac{\C(S_2, u)}{\C(S_2,e)}$\nllabel{sample:line:upper_draw_node}\;
         accept $u$ with probability $\frac{1}{|E(v) \cap E(u)|}$\nllabel{sample:line:upper_accept} 
         }\nllabel{sample:line:outer_while_beforelast}}
        accept $u$ with probability $\frac{1}{\Ind{u \in N_{\le \alpha}(v)}
     + \Ind{u \in N_{> \alpha}(v)}}$\nllabel{sample:line:outer_accept}\;
        }
        \Return $(T_2, S_2, u)$\;
        }
\end{algorithm}

\begin{lemma}\label{lemma:number_of_samples}
Given as input $(T,S,v)$, \SampleNeigh returns $(T_2,S_2,u)$ such that $u \in N(v)$ with probability proportional to $C(T_2,S_2,u)$.
\end{lemma}

 \begin{proof}
 Consider any iteration of the loop of \cref{sample:line:outer_while}.
 Let $p(u)$ be the probability, ignoring \cref{sample:line:outer_accept}, that the iteration draws the particular vertex $u$.
 We shall prove that:
 \begin{align}\label{eq:proportional}
    p(u) \propto \C(S_2,u) \cdot  \left(\Ind{u \in N_{\le \alpha}(v)}
      + \Ind{u \in N_{> \alpha}(v)}\right)
 \end{align}
 By \cref{sample:line:outer_accept}, this implies that every single iteration draws \emph{and} accepts $u$ with probability proportional to $C(T_2,S_2,u)$, proving the claim.

 Let $p_{\le \alpha}(u)$ be the probability that an execution of \cref{sample:line:lower_draw} draws $u$, and $p_{> \alpha}(u)$ the probability that an execution of the block of \cref{sample:line:else} draws $u$.
 By the law of total probability:
 \begin{align}\label{eq:total_law}
 p(u)
 &=p(N_{\leq \alpha}(v))\cdot p_{\leq \alpha}(u) + p(N_{> \alpha}(v))\cdot p_{> \alpha}(u)
 \\
 & = \frac{W_{\le\alpha}(v)\cdot p_{\leq \alpha}(u)}{W_{\le\alpha}(v)+W_{>\alpha}(v)} + \frac{W_{>\alpha}(v)\cdot p_{> \alpha}(u)}{W_{\le\alpha}(v)+W_{>\alpha}(v)} \label{eq:prop2}
 \end{align}
Let us then analyse $p_{\leq \alpha}(u)$ and $p_{> \alpha}(u)$ separately.

 For $p_{\leq \alpha}(u)$, \cref{sample:line:lower_draw} implies:
 \begin{align}
      p_{\leq \alpha}(u) = \frac{\C(S_2,u)}{W_{\le\alpha}(v)} \cdot \Ind{u \in N_{\le \alpha}(v)}
 \end{align}

 For $p_{> \alpha}(u)$, let $p_{inn}(u)$ be the probability that, ignoring \cref{sample:line:upper_accept}, one iteration of \cref{sample:line:inner_while} draws a particular vertex $u$; let $p(e)$ be the probability that one execution of \cref{sample:line:upper_draw_edge} returns the specific hyperedge $e$; and let $p(u|e)$ be the probability that, conditional on this event, one execution of \cref{sample:line:upper_draw_node} returns $u$.
 By the law of total probability:
 \begin{align}
 p_{inn}(u)
 &= \Ind{u \in N_{> \alpha}(v)} \cdot
 \sum_{e \in E(v)}
   p(u \mid e)\,p(e)
 \\
 &=\Ind{u \in N_{> \alpha}(v)}\\
 &\cdot
 \sum_{e \in E(v)}
   \frac{\C(S_2,u)}{\C_2(S_2,e)}
   \frac{\C(S_2,e)}
        {\displaystyle
         \sum_{e' \in E(v)} \C(S_2,e')}\, \cdot\Ind{u \in e}\,
 \\[0.25em]
 &=  \Ind{u \in N_{> \alpha}(v)} \cdot \sum_{e \in E(v)} \frac{\C(S_2,u)}{ \sum_{e' \in E(v)} \C(S_2,e')}\cdot\Ind{u \in e} \\
 &=  \Ind{u \in N_{> \alpha}(v)} \cdot \frac{\C(S_2,u) |E(v) \cap E(u)|}{ \sum_{e' \in E(v)} \C(S_2,e')}
 \end{align}
 By \cref{sample:line:upper_accept}, this implies that every single iteration draws and accepts $u$ with probability proportional to $\C(S_2, u)\cdot \Ind{u \in N_{> \alpha}(v)}$, that is:
 \begin{align}
    p_{> \alpha}(u) \propto \C(S_2,u) \cdot \Ind{u \in N_{> \alpha}(v)}
      \end{align}
 Since $\sum_{u \in N(v)} \C(S_2, u) \cdot \Ind{u \in N_{> \alpha}(v)} = W_{> \alpha}(v)$, we deduce that:
 \begin{align}
 p_{> \alpha}(u) = \frac{\C(S_2,u)}{W_{> \alpha}(v)}\cdot \Ind{u \in N_{> \alpha}(v)}    
 \end{align}
 Plugging $p_{\le \alpha}(u)$ and $p_{> \alpha}(u)$ in \Cref{eq:prop2} proves the claim. 
 \end{proof}

\subsubsection{Sampling colorful treelets}
We now use \SampleNeigh to sample a colorful rooted treelet of $H$.
To reduce the complexity of the task, we assume that the build-up phase computes a set of random generators that produce samples from certain distributions.
Using the Alias method~\cite{Vose91-alias}, these generators can be constructed in time linear in the support of the distribution, and yield independent random samples in time $O(1)$.
In particular, we will assume the following generators:
\begin{itemize}[leftmargin=10pt]
    \item A generator that returns a pair $(T,v)$ with probability proportional to $\C(T,[k], v)$.
     The size of the support is $2^{O(k)}\cdot|V|$.
     \item For every $S_2 \subset [k]$, every treelet $T_2$ on $|S_2|$ vertices, and every $v \in V$,  a generator that returns a vertex $u \in N_{\le \alpha}(v)$ with probability proportional to $C(T_2,S_2,u)$.
     This can be done by iterating over $C(T_2,S_2,u)$ for all $u \in N_{\le \alpha}(v)$, which can be done in time $O(\alpha^2 |E|)$.
     The total time is therefore $2^{O(k)} \cdot \alpha^2 |E|$.
     \item For every $S_2 \subset [k]$, every treelet $T_2$ on $|S_2|$ vertices, and every $e \in E_{>\alpha}$, we compute $C(T_2,S_2, e)$ in time $O(|e|)$. At this point we compute a generator that returns $u \in e$ with probability proportional to $C(T_2,S_2,u$), which takes again time $O(|e|)$. The total time is therefore $2^{O(k)}\cdot|E_{>\alpha}|$.    
    \item For every $S_2 \subset [k]$, every treelet $T_2$ on $|S_2|$ vertices, and every vertex type $E(v)$, a generator returning $e \in E(v)$ with probability proportional to $C(T_2,S_2,e)$, in time $O(\beta)$, for a total time of $2^{O(k)}\cdot \beta|V|$.
    \item For every $v \in V$, a generator that returns $(S_1, S_2)$ with probability proportional to $\C(T_1,S_1,v) \cdot \sum_{u \in N(v)}\C(T_2,S_2,u)$. Recall that the neighbor sums $\sum_{u \in N(v)}\C(T_2,S_2,u)$ are computed by the build-up phase. This takes total time $2^{O(k)} \cdot |V|$.
\end{itemize}
We also pre-compute the canonical decomposition of every treelet $T$ on at most $k$ vertices (this can be done in time $2^{O(k)}\cdot O(k)$).
One can check that the total time spent by the calculations above is bounded by the running time of the build-up phase.
%
%
%

\begin{lemma}\label{lemma:sample_complexity} With the random number generators precomputed as above, each invocation of {\SampleNeigh} takes expected time $O(\beta^2)$.
\end{lemma}

 \begin{proof}
  Let $I$ be the random variable that counts the total number of iterations executed by the while loop at \cref{sample:line:inner_while}. It holds $I \sim Geom(|E(v) \cap E(u)|)$. Let $T$ be the random variable that counts the total number of iterations within the body of a single execution of the while loop at \cref{sample:line:outer_while}. Then:
 \[
 T = \begin{cases}
 1 &  \text{with probability } p(N_{\le \alpha})\\
 I & \text{with probability } 1-p(N_{\le \alpha})
 \end{cases}
 \]
 It follows that 
 \begin{align}
 \E[T] = p(N_{\le \alpha}) + (1- p(N_{\le \alpha}))\E[I] \leq 1+\beta
 \end{align}
 Thus, it holds $\E[T] = O(\beta)$. Let $K$ be the random variable that counts total number of times the loop at \cref{sample:line:outer_while} is executed, and $N$ the random variable that counts the total number of executions (including all the loops).
 Consider now another version of \Cref{sample_neigh} where \cref{sample:line:outer_accept} accepts with fixed probability $\frac{1}{2}$, and let $K',N'$ be the equivalent of $K,N$ for this version.
 By a simple coupling argument, $K \le K'$ and $N \le N'$, and clearly $K'$ is independent of the variables $T_i$.
 By Wald's equality, then,
 \begin{align}
     N \le N' = \sum_{i=1}^{K'} T_i = \E[K']\E[T] \le 2 \E[T] = O(\beta) \,.
 \end{align}

 Now let us bound the time taken by the execution of every line. Outside the loop at \cref{sample:line:outer_while}, the execution of each line takes at most $O(1)$. Drawing $u$ at \cref{sample:line:lower_draw} takes $O(1)$. The same holds for \cref{sample:line:upper_draw_node}. Sampling a hyperedge at \cref{sample:line:upper_draw_edge} takes time $O(1)$, whereas computing $|E(v) \cap E(u)|$ at \cref{sample:line:upper_accept} takes time at most $O(\beta)$. Thus, if the expected number of iterations is $O(\beta)$, it follows that each invocation of {\SampleNeigh} takes expected time $O(\beta^2)$.
 \end{proof}

The procedure for sampling a treelet is presented in \Cref{algo:sample}.

\begin{algorithm}[h]
\caption{\Sample}
\label{algo:sample}
\LinesNumbered
{\fontsizealgo
\Input{$(T,S,v)$ or \texttt{NULL}}
\Output{random uniform colorful treelet of $H$ isomorphic to $T$ colored with $S$ rooted in $v$}
\If{input is \texttt{NULL}}{
$S=[k]$\;
draw $(T, v)$ with probability proportional to $\C(T,S,v)$\;
}
\If{$|T|=1$}{
    \Return $(\{v\},\emptyset)$\;
}
$(T_2,S_2,u) \leftarrow $\SampleNeigh$(T,S,v)$\;
let $T_1 = T \setminus T_2$, $S_1=S \setminus S_2$\;
\Return \Sample$(T_1,S_1,v)$ + $uv$ + \Sample$(T_2,S_2,u)$\;
}
\end{algorithm}

\begin{corollary} \label{corollary:sample} \Cref{algo:sample} takes, on each invocation, expected time $O(k\beta^2)$.
Moreover, the vertex set of the output treelet is $U \in \scV_k(H)$ with probability proportional to the number of spanning trees $\sigma(H[U])$ of $\Gaifman(H[U])$.
\end{corollary}

 \begin{proof}
     The running time bound follows from \Cref{lemma:sample_complexity} and the fact that \Sample makes $k-1$ calls to \SampleNeigh, together with the bounds on the time for drawing $(T,v)$.
     The claim on the distribution follows from \Cref{lemma:number_of_samples} and the same analysis of~\cite{Bressan&2019VLDB}.
 \end{proof}

\subsubsection{Sampling hypergraphlets}
As a final step, we show how to obtain an unbiased estimate of the frequency of the $k$-hypergraphlets. 
To this end, consider an invocation of \Sample and let $U$ the set of vertices of the returned treelet.
We then need to compute:
\begin{enumerate}
    \item the $k$-hypergraphlet $H[U]$,
    \item the number of spanning trees $\sigma(H[U])$ of $\Gaifman(H[U])$.
\end{enumerate}
Thereafter, we just increase the sample counter of $H[U]$ by $\frac{1}{\sigma(H[U])}$.
By virtue of \Cref{corollary:sample}, this yields an unbiased estimator of the number of colorful $k$-hypergraphlets of $H$ isomorphic to $H[U]$.
Note that, at this point, one can sample $k$-hypergraphlets uniformly by accepting $H[U]$ with probability $\frac{1}{\sigma(H[U])}$. Since $\frac{1}{\sigma(H[U])} \ge k^{-O(k)}$, see~\cite{Bressan&2019VLDB}, this makes the sampling time's dependence on $k$ grow to $k^{O(k)}$, as per \Cref{thm:hypermotivo}.

For item 1, one could list the vertex types of the vertices of $U$ to find the hyperedges of $H[U]$.
This, however, would have a running time that scaled with the maximum degree $\Delta(H)$.
Interestingly, we can do better by taking advantage of the $(\alpha, \beta)$-niceness of $H$.
\begin{lemma}\label{lemma:subhypergraphlet_build}
By spending an additional $O(\alpha^k \cdot |E_{\le \alpha}|)$ preprocessing time, one can compute the $k$-hypergraphlet $H[U]$ induced by any given $k$-vertex $U \subseteq V$ in time $2^{O(k)} \cdot \beta$.
\end{lemma}
 \begin{proof}
 Given $U$, we check for each $X \subseteq U$ with $|X|\ge 2$ whether there exists $e \in E$ such that $e \cap U = X$, in which case we add $X$ to the hyperedges of $H[U]$.
 We now discuss how $e$ can be found in a generic set of edges $E$; we then show how to apply the idea to $E_{\le \alpha}$ and $E_{> \alpha}$.
 Suppose we have precomputed, for every $X \subseteq U$, the number $N[X] = |\{e \in  E : X \subseteq e\}|$.
 Moreover let $N^*(X,U) = |\{e \in E : e \cap U = X\}|$; note that $N^*(X,U) > 0$ means $e \cap U = X$ for some $e \in E$.
 Now, by inclusion-exclusion,
 \begin{align}
     N^*(X,U) = \sum_{Y \subseteq U \setminus X} (-1)^{|Y|} N[X \cup Y] \,.
 \end{align}
 Therefore, given the counters $N[X]$, we can compute $N^*(X,U)$ in time $2^{O(|U|)} \le 2^{k}$.

 Let us now discuss how to implement this for $E_{\le \alpha}$ and $E_{> \alpha}$.
 For $E_{\le \alpha}$, in the build-up phase for every $e \in E_{\le \alpha}$ we enumerate all subsets $X \subseteq e$ of at most $k$ vertices, increasing $N[X]$ accordingly.
 To this end we store $N[X]$ in an associative array with logarithmic update and access time.
 This computes the counters $N[X]$ explicitly in time $O(\alpha^k \cdot |E_{\le \alpha}|)$.
 For $E_{> \alpha}$, observe that $N[X] = \left|\bigcup_{v \in X} E_{>\alpha}(v)\right|$, and the right-hand side can be computed in time $O(k \beta)$, yielding a total time of $O(2^k \cdot k \beta)$ for computing all counters once given $U$.
 By a final bound, given $U$, checking whether $e$ exists in $E = E_{\le \alpha} \cup E_{> \alpha}$ takes a total time of $2^{O(k)} \cdot \beta$, as claimed.
 \end{proof}

For item 2, we first compute $G[U]=\Gaifman(H[U])$ and then $\sigma(G[U])$.
It is well known that $\sigma(G[U])$ can be computed via Kirchhoff's matrix tree theorem in time $O(k^{\omega})$, where $\omega < 2.38$ is the matrix multiplication exponent.
Let us then focus on computing $G[U]$.
The obvious approach is to consider every pair $\{u,v\} \in U$ and check whether they are adjacent in $H[U]$, that is, in $H$.
This could be done by intersecting their types, that is, by checking whether $E(u) \cap E(v) = \emptyset$.
This however would take time $\Omega(E(u)+E(v))$, which could be linear in $V$.
Once again, we leverage the $(\alpha,\beta)$-niceness of $H$, together with the precomputations done in the build-up phase, to obtain a faster algorithm:


\ignore{
\begin{itemize}
    \item Sample an occurrence $T$ of a treelet on $k$ nodes from $H$
    \item Consider the projection $\Gaifman(H[U])$ of the sub-hypergraph $H[U]$ induced by the nodes of $T$
    \item Reject it with probability $1 - \frac{1}{\sigma(H[U])}$
\end{itemize}

The value $\sigma(H[U])$ can be computed for
any given $H[U]$ in time $O(k^\omega)$, where $\omega$ is the matrix multiplication exponent.

The main obstacle in computing this quantity is the construction of the adjacency matrix of $H[U]$, since in general this depends on the maximum degree $\Delta(H)$, which could be impractically large.
However, we observe that one can exploit the properties of $(\alpha, \beta)$-nice hypergraphs to reduce the complexity of this step.
}

\begin{lemma}\label{lemma:adjacency}
Given a subset $U \subseteq V$ of $k$ vertices, one can compute the adjacency matrix $A_U$ of $\Gaifman(H[U])$ in time $O(k^2(\beta +\ln{|V|}))$.
\end{lemma}
 \begin{proof}
     Clearly, $A_U[u][v]=1$ if and only if $u \in N_{\le \alpha}(v)$ or $E_{>\alpha}(u) \cap E_{>\alpha}(v) \ne \emptyset$.
     The condition $u \in N_{\le \alpha}(v)$ can be verified in time $O(\log |V|)$ via binary search as $N_{\le \alpha}(v)$ is sorted and has size at most $|V|$.
     The condition $E_{>\alpha}(u) \cap E_{>\alpha}(v) \ne \emptyset$ can be verified in time $O(\beta)$ since both sets are sorted and have size at most $\beta$.
     A union bound over all $k^2$ pairs yields the claim.
 \end{proof}

\ignore{
\begin{lemma}
    The \emph{sampling phase} takes, on each invocation, expected time $O(k \beta^2 + k^2 (\beta+\ln{|V|})) = O(k^2 (\beta^2 + \ln|V|))$, and returns a colorful subset $U \subseteq V$ in $\scV_k(H)$ with probability proportional to the number of spanning trees $\sigma(H[U])$ of $H[U]$, as well as $\sigma(H[U])$.
\end{lemma}
}

%% file: 5_Experiments.tex
We conducted experiments to evaluate the practical performance of \hypermotivo in terms of sensitivity to the parameters, running time, memory usage, and estimation accuracy.
We present the results for a limited choice of parameters; for further ones see our repository.
We implemented \hypermotivo in \CC\ as an extension of \motivo\footnote{\url{https://bitbucket.org/steven_/motivo/}} and made the source publicly available\footnote{\url{https://github.com/gfumagalli9/hyper-motivo/tree/hyper-motivo}}.
Like \motivo, we implemented \hypermotivo with native support for multi-threading.


\subsection{Datasets}\label{sec:experiments_datasets}
We use six hypergraphs obtained from publicly available data, and one synthetic hypergraph generated by a simple random model. See \Cref{tab:freq}.

\begin{table}[t]
\centering
\caption{Hypergraphs used in our experiments.}
\label{tab:freq}

\small

\begin{tabular}{|l|r|r|r|r|r|r|r|}
\hline
$H$ & $|V(H)|$ & $|E(H)|$ & $\rank(H)$ & $\Delta(H)$ & avg.\ deg & $\lVert H \rVert$ & $\lVert \Gaifman(H) \rVert$ \\
\hline
\texttt{SA} & 15,211,989 & 1,103,243 & 61,315 & 356 & 1.7 & 26,109,177 & 29,372,932,379 \\
\texttt{RH} & 100,000    & 55,000    & 77,965 & 26  & 9.6 & 960,168    & 12,083,305,610 \\
\texttt{AR} & 2,268,231  & 4,285,363 & 9,350  & 28,973 & 32.0 & 73,141,425 & 7,351,426,495 \\
\texttt{AX} & 259,961    & 164       & 40,571 & 9   & 1.9 & 487,469    & 5,744,552,301 \\
\texttt{GP} & 479,285    & 173,085   & 10,598 & 202 & 9.5 & 4,573,183  & 1,971,829,081 \\
\texttt{SE} & 36,775     & 702       & 11,391 & 5   & 3.1 & 113,074    & 295,309,540 \\
\texttt{MA} & 73,851     & 5,446     & 1,784  & 173 & 1.8 & 131,714    & 35,382,628 \\
\hline
\end{tabular}

\end{table}

We describe here below how each hypergraph was obtained.
For reproducibility, all scripts used to construct the hypergraphs are available at \href{https://github.com/gfumagalli9/hyper-motivo/tree/hyper-motivo}{our repository.}
Note that some hypergraphs are derived from inhomogeneous data through significant data wrangling (e.g., downloading several XML files from an interface, and then parsing and combining them).

\paragraph{\textsc{stackoverflow-answers} (\texttt{SA}) \& \textsc{mathoverflow-answers} (\texttt{MA})}
The SA and MA dataset, introduced in \cite{Veldt-2020-local}, encode answering activity respectively on StackOverflow and MathOverflow. Vertices correspond to questions, and each hyperedge contains all questions answered by a given user. 

\paragraph{\textsc{amazon-reviews} (\texttt{AR})}
The AR hypergraph has been introduced in \cite{ni2019justifying}. Vertices represent individual reviewers, while each hyperedge corresponds to a product and contains the set of its reviewers.

\paragraph{\textsc{cpc-group} (\texttt{CP})}
A patent–classification hypergraph built from USPTO PatentsView bulk data. Vertices are granted U.S. utility patents; hyperedges are Cooperative Patent Classification (CPC) group codes, each containing all patents assigned to that group.

\paragraph{\textsc{datascience-se-tags} (\texttt{SE})} Vertices represent questions, and hyperedges correspond to tags. By design, each question in stackexchange can have at most 5 tags; in other words, this means that this hypergraph is $(\alpha, 5)$-nice for every possible $\alpha$. The SE hypergraph is built from the public data dump of Datascience Stack-Exchange.

\paragraph{\textsc{arxiv-categories} (\texttt{AX})}
Vertices are arXiv articles, and hyperedges are arXiv subject categories. Each hyperedge contains all articles annotated with that category during a time span of a year. Both primary and secondary categories are included. The dataset has been generated by collecting data from the arXiv base endpoint.

\paragraph{\textsc{Random Hypergraph} (\texttt{RH})}
An artificial hypergraph generated by a simple power-law model. First, we let $V=\{1,\ldots,100.000\}$.
We then generated $55.000$ independent random hyperedges, each one as follows: we drew an integer $s \in \{2,\ldots,|V|\}$ with probability proportional to $s^{-2}$, and we then picked a subset $e \in {V \choose s}$ uniformly at random.

\subsection{Setup}
All experiments were conducted on a dedicated Linux machine equipped with two Intel(R) Xeon(R) CPU X5660  @ 2.80GHz, 128 GB of RAM and 2 TB of storage. \Giacomo{Unless otherwise specified, all experiments reported here are performed for $k=5$ using 8 threads; for other values see the supplementary material.}

For each dataset, and for each $3 \le k \le 8$, we ran both \hypermotivo and \Cref{algo:Naive}, asking each algorithm to take $K=100.000$ samples.
Whenever a time wall budget of 12 hours was reached, the process was killed.
The value $\alpha^\star$ of $\alpha$ chosen by \hypermotivo, see \Cref{fig:alfa_beta}, comes from the linear-time algorithm presented in \Cref{sub:splitting}, with the following modification: instead of minimizing the bound of \Cref{eq:T_bound_2}, we minimize a weighted version of that bound, in the form $\gamma \cdot (\sum_{e \in E_{\le \alpha}} |e|^2) +(1-\gamma)\cdot (\sum_{v \in V} 2^{d_{> \alpha}(v)})$. 
\Giacomo{We set $\gamma = 0.01$; this provided good performance across all datasets.}

In the sampling phase, the hypergraphlet $H[U]$ is not constructed using the algorithm of  \Cref{lemma:subhypergraphlet_build}. We instead use a simpler routine that, for each $v \in U$, intersects the incident hyperedges $e \in E_v$ with $U$. This approach performs well in our experiments and is considerably easier to implement. 

\Giacomo{Finally, we remark that the Amazon Reviews dataset represents a worst-case scenario for our algorithm. In this dataset, high-degree vertices are incentivized to participate in very large hyperedges, which affects the exponential term in the complexity bound. This behavior reflects the natural dynamics of review platforms, where popular products with many reviews attract an increasing number of buyers. Despite this, the performance of our algorithm is comparable to that of the baseline.}

\subsection{Computational Performance}\label{sec:experiments_results}

\subsubsection{Sensitivity to the choice of $\alpha$}
\New{
We evaluated the sensitivity of \hypermotivo's preprocessing time to perturbations of the value of $\alpha^\star$ by simply running it with values ranging from $0.5 \alpha^\star$ to $1.5 \alpha^\star$.
\Cref{fig:alpha_sensitivity} shows that, across all datasets, the algorithm is rather stable and the chosen value $\alpha^\star$ provides good performance.
(This is consistent with the fact that $\alpha^\star$ is not finely-tuned, but chosen in a way only coarsely dependent from $H$, see above).
\begin{figure}
    \centering
    \includegraphics[width=0.8\linewidth]{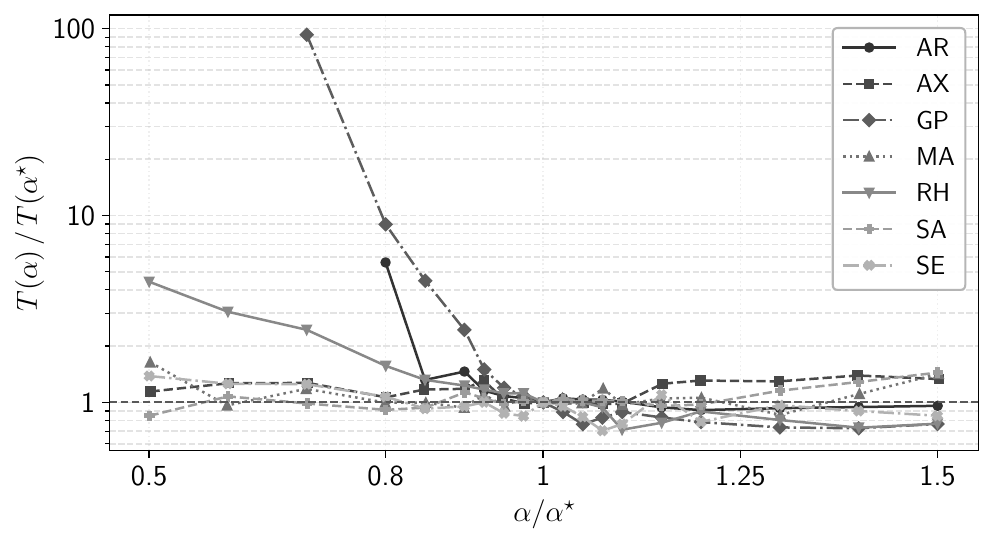}
    \caption{\Giacomo{Sensitivity of \hypermotivo to the choice of $\alpha^\star$. 
    For $\alpha$ varying around $\alpha^\star$, the plot shows the ratio $T(\alpha)/T(\alpha^\star)$ between the build-up times using $\alpha$ and using $\alpha^\star$. 
    Our choice of $\alpha^\star$ yields good performance over all inputs.
    For \texttt{AR} and \texttt{GP} we omit points with running times exceeding 24 hours.}}
    \label{fig:alpha_sensitivity}
\end{figure}
}

\subsubsection{The quadratic barrier and \hypermotivo.}
\Giacomo{
\Cref{fig:build-up-scalabililty} shows the build-up time of \Cref{algo:Naive} and of \hypermotivo as a function of the input size; note how \Cref{algo:Naive} exhibits a behavior coherent with $\Omega(\size{H}^2)$, and \hypermotivo a behavior consistent with $O(\size{H})$.
The inputs are generated as follows. 
Let $n,m,\alpha,\beta,\rho$ be fixed.
We generate an $n$-vertex, $m$-hyperedge hypergraph by sampling $\rho m$ \emph{small} hyperedges and $(1-\rho)m$ \emph{large} hyperedges independently. 
Each small hyperedge is generated by first drawing its size $s$ uniformly from $\{2,\dots,\alpha-1\}$ and then selecting $s$ uniform random vertices.
Each large hyperedge is generated by selecting $L$ uniform random vertices for a fixed $L > \alpha$; any vertex whose degree in $H_{>\alpha}$ exceeds $\beta$ is rejected and resampled.
This yields an $(\alpha,\beta)$-nice hypergraph. 
(This process is not meant to model real data, but rather to control the $(\alpha,\beta)$-niceness of $H$).}

\begin{figure}
    \centering
    \includegraphics[width=0.8\linewidth]{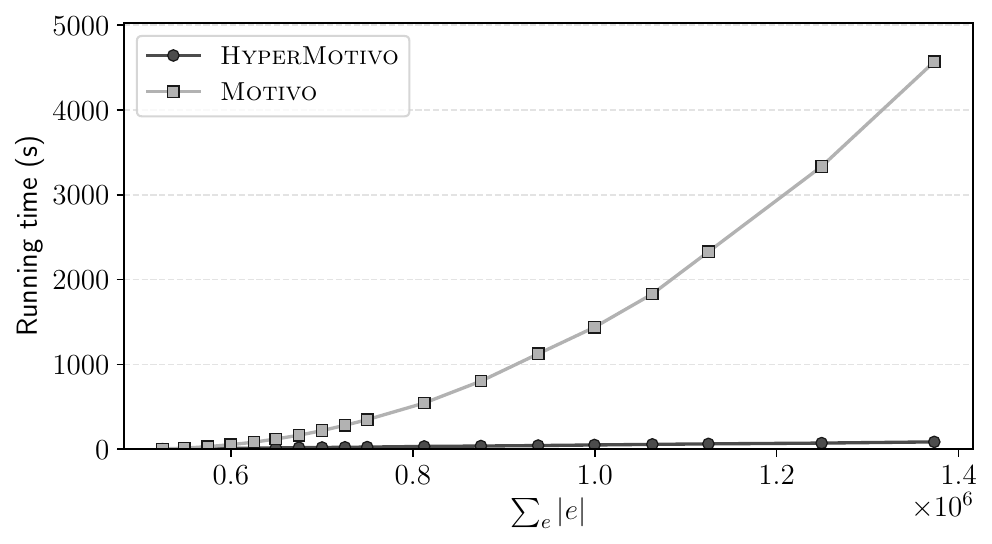}
    \caption{\Giacomo{The linear-time behavior of \hypermotivo and the quadratic-time behavior of the naive \Cref{algo:Naive}, on hypergraphs produced by a simple random model.}}
    \label{fig:build-up-scalabililty}
\end{figure}

\subsubsection{Time and memory improvement.}
\Cref{fig:K5_time_memory_comparison} shows the ratio between the running time and memory usage of the build-up phase of \Cref{algo:Naive} and \hypermotivo (the higher the bar, the larger the improvement), for $k=5$.
For \hypermotivo, this includes the computation of the $\alpha$-split (which is typically negligible) and the random number generators. 
For \Cref{algo:Naive}, this includes the computation of $\Gaifman(H)$, which often dominates, and the whole build-up phase of \motivo.
As expected, \hypermotivo is substantially faster in almost all cases, even in hypergraphs of high degree.
The only exception is \texttt{AR}, where \hypermotivo is only marginally outperformed.
\ignore{
We compared the build-up time (by far the dominant component of the overall running time) of the naive algorithm with that of \hypermotivo.
\Cref{fig:time} reports the build-up times for
k=5 using 8 cores, including the time required for the projection of the hypergraph. For easier comparison, the running time of the baseline algorithm was normalized so that the relative speed-ups for the various hypergraphs could be compared directly.
}
\ignore{
\begin{figure}[h]
  \centering
  \includegraphics[width=\linewidth]{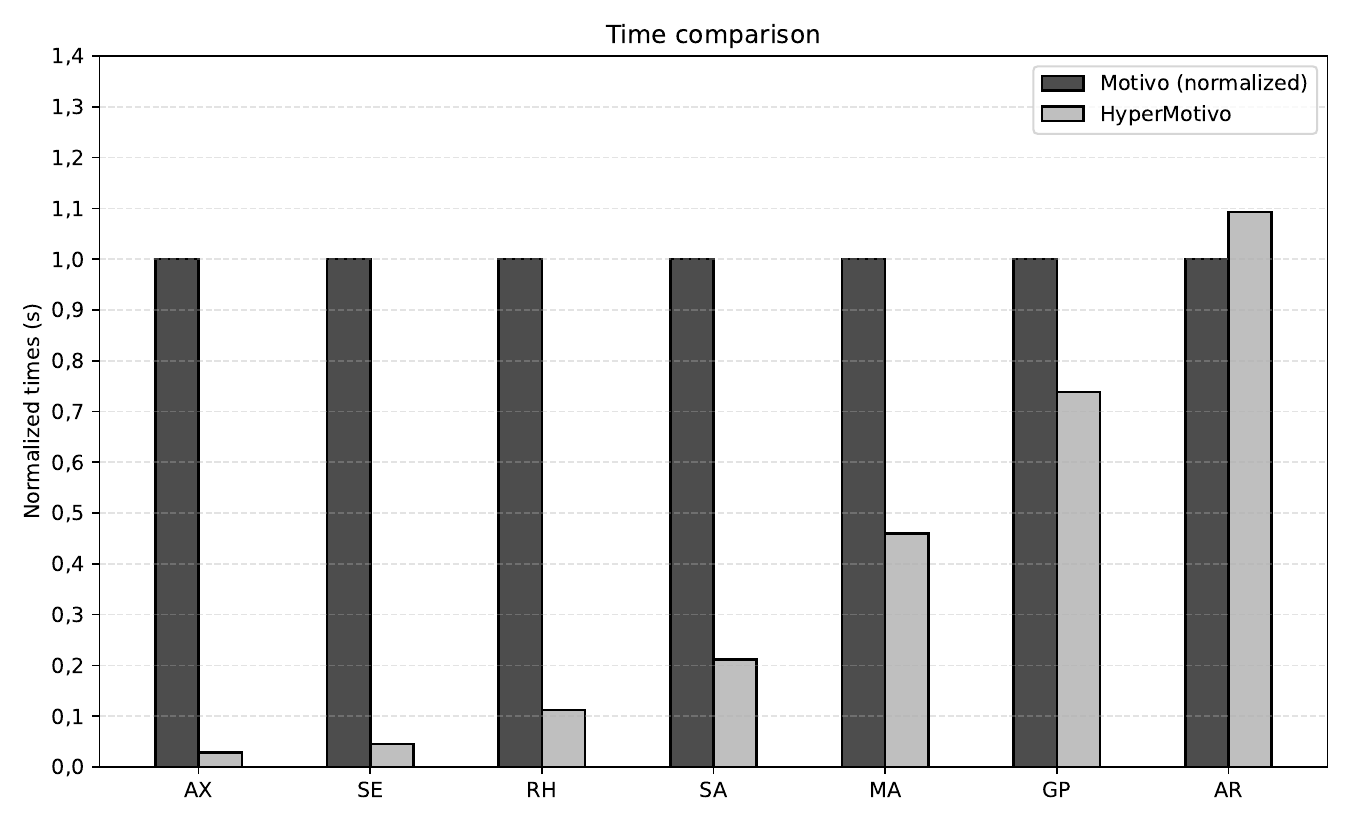}
  \caption{The build-up time of the baseline algorithm (dark bars) and \hypermotivo (light bars). Except for AR, where \hypermotivo pays a small overhead, \hypermotivo is always faster, often by orders of magnitude.}
  \label{fig:time}
\end{figure}

}
\begin{figure}
  \centering
  \includegraphics[width=0.8\linewidth]{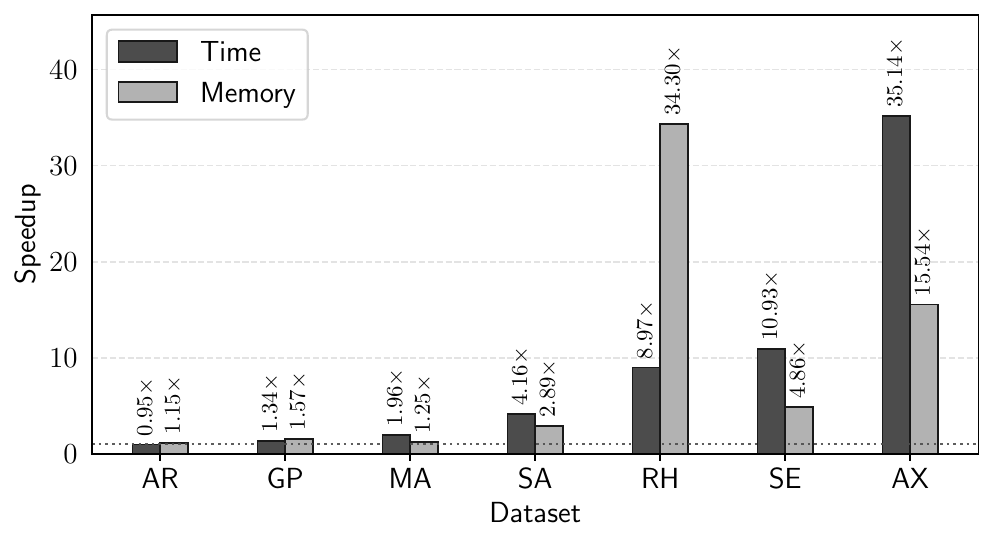}
  \caption{Reduction in wall-clock time and peak memory usage of \hypermotivo over \Cref{algo:Naive}, on the build-up phase.
  On three inputs, \hypermotivo saves at least an order of magnitude in time or memory.
  }
  \label{fig:K5_time_memory_comparison}
\end{figure}

\ignore{
\subsubsection{Memory usage}
We compared \hypermotivo to the baseline algorithm in terms of main-memory consumption. In all cases, the sampling phase is the most demanding part of the computation, as it requires loading all tables generated during the build-up phase.
\Cref{fig:memory} reports the peak main-memory usage during the entire execution of the algorithm for k=5 and 8 cores.
In all observed datasets, \hypermotivo uses significantly less memory than the baseline algorithm.
}
\ignore{
\begin{figure}
  \centering
  \includegraphics[width=0.9\linewidth]{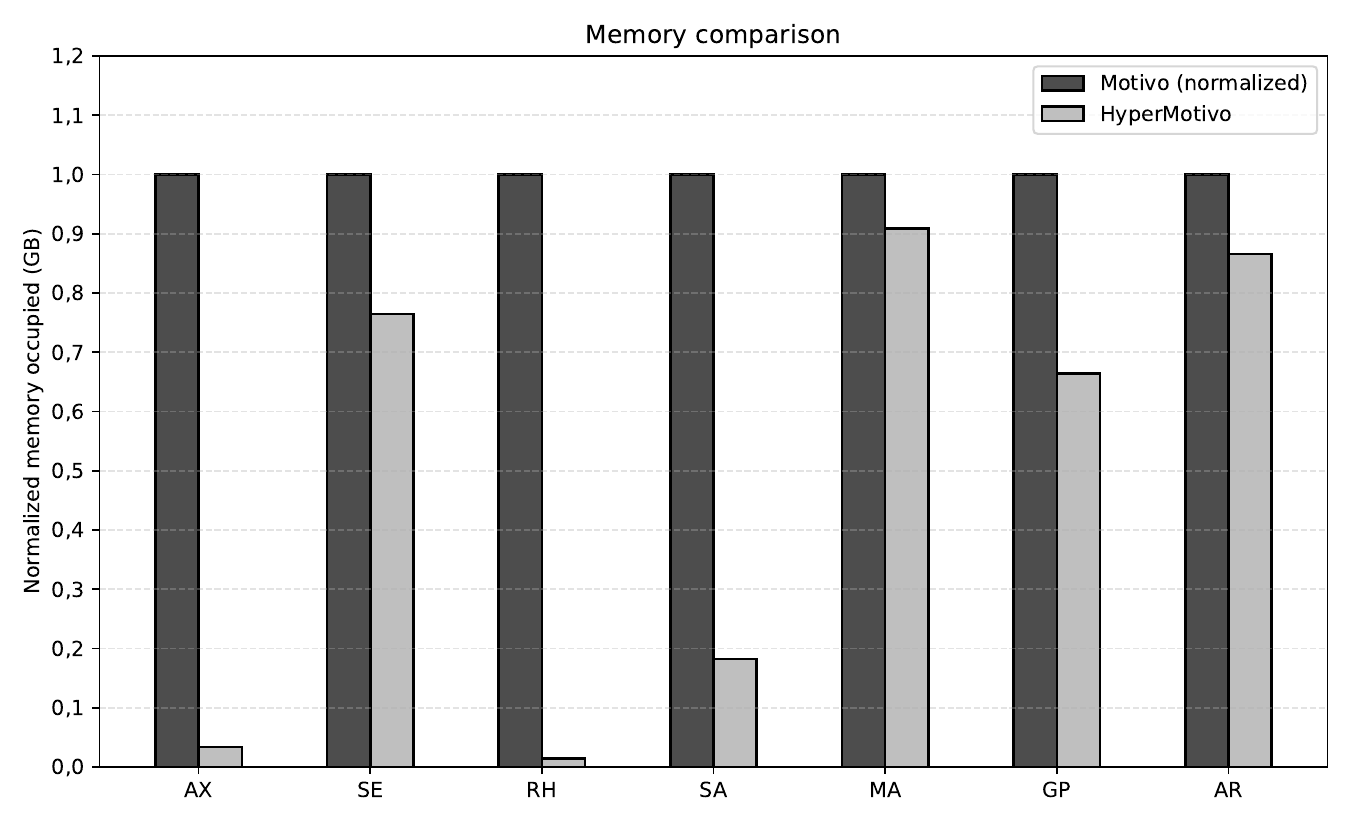}
  \caption{A comparison}
  \label{fig:memory}
\end{figure}
}

\subsubsection{Parallel scalability of the build-up phase}
\Cref{fig:parallel} shows the running time of \hypermotivo's build-up phase for \Giacomo{$k=5$} on the \texttt{MA} hypergraph, as a function of the number of CPU threads.
When using up to \Giacomo{$4$} threads, doubling the number of threads almost halves the running time.
Other datasets, and other values of $k$, exhibit a similar behaviour.
This shows that our techniques can be effectively employed in practice, too.

\begin{figure}
  \centering
  \includegraphics[width=0.8\linewidth,trim=0 20 0 0,clip]{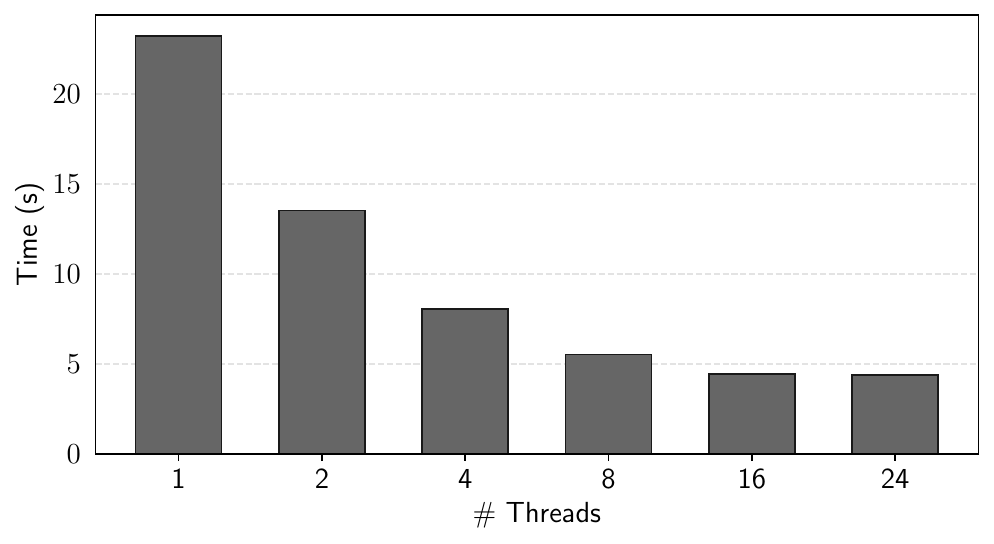}
  \caption{\Giacomo{Running time of \hypermotivo's build-up phase on the \texttt{MA} dataset, for $k=5$, as a function of the number of CPU threads used.
  Doubling the number of threads almost halves the running time, at least up to 4 threads}.
  }
  \label{fig:parallel}
\end{figure}

\subsubsection{Sampling speed}
\Giacomo{\Cref{fig:sampling_throughput_K5}} shows the number of samples per second achieved by \hypermotivo and by \Cref{algo:Naive}.
Note that \hypermotivo's sampling algorithm is more complex than \motivo's one (which is in fact used as a subroutine).
Nonetheless, \hypermotivo is always within a factor of 2 of \motivo.
\begin{figure}
    \centering
    \includegraphics[width=0.8\linewidth,trim=0 25 0 0,clip]{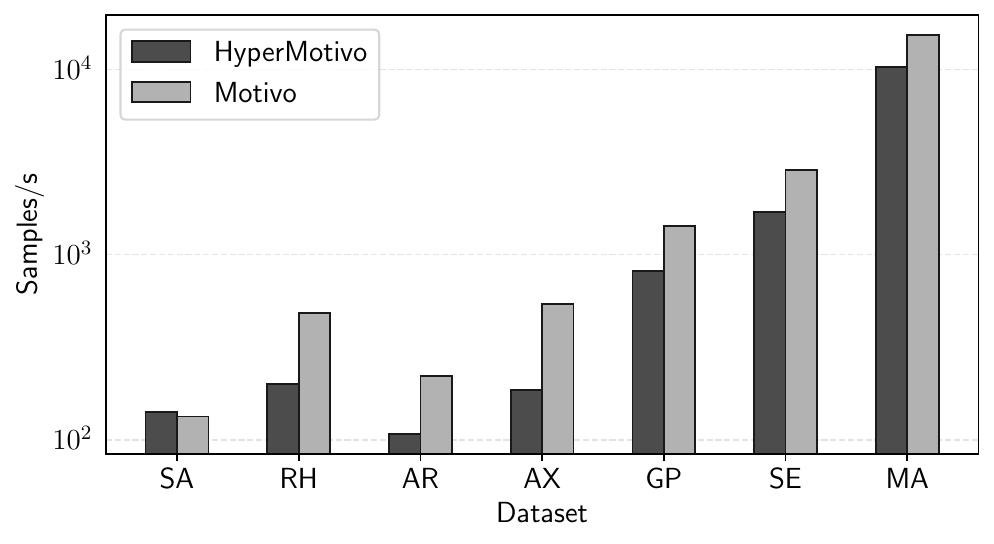}
    \caption{The sampling speed of \hypermotivo and \Cref{algo:Naive}, in samples per second.}
    \label{fig:sampling_throughput_K5}
\end{figure}

\subsection{Accuracy}\label{sec:accuracy}
\Giacomo{
We empirically evaluate the accuracy of \hypermotivo on small synthetic instances, for which exact counts can be computed.
To this end we fix $|V|=1000$ and $|E|=500$.
To generate an edge $e$, we first draw its size $s=|e|$ from a power-law distribution so that $\Pr[s]\propto s^{-3}$; then, we select $s$ vertices uniformly at random. 
We generate in this way four random hypergraphs, \texttt{RH1}, \ldots, \texttt{RH4}, so as to check that the results are consistent.
}
\Giacomo{For each hypergraphlet isomorphism type $H$, we measure the relative count error $\err_H = (\hat c_H - c_H)/c_H$, where $c_H$ denotes the exact count of $H$ in the input hypergraph and $\hat c_H$ is the estimate returned by \hypermotivo when using $10^5$ samples. Thus $\err_H=0$ corresponds to a perfect estimate, and $\err_H=-1$ means $H$ was never sampled.
\Cref{fig:accuracy_k5} reports the distribution of $\err_H$, clustered in bins in the form $[\err-0.25,\err+0.25)$. 
It can be seen that \hypermotivo yields accurate estimates for most of them, with the error distribution concentrated around $0$.
}

\Giacomo{
A note of warning is needed here.
The total number of (non-isomorphic) hypergraphlets on $k$ vertices is \emph{doubly exponential} in $k$; for $k=5$, it is already around $10$ million.\footnote{\url{https://oeis.org/A000612} shows that the number of hypergraphs on $5$ vertices is 18.666.624; at least half of them have the edge $e=V$ and are therefore connected.}
Thus, most of them necessarily have a very low frequency, and their count cannot be easily estimated.
For this reason, \Cref{fig:accuracy_k5} is constructed considering only the 50 hypergraphlets with largest absolute counts; across the four inputs, the ground truth exhibits roughly $100$ distinct hypergraphlet types in total; the remaining types are extremely rare and would therefore be missed by sampling with high probability (i.e., \(\err_H=-1\)).
}

\begin{figure}[!tbp]
    \centering
    \begin{minipage}{0.8\linewidth}
        \centering
        \includegraphics[width=\linewidth,trim=0 25 0 0,clip]{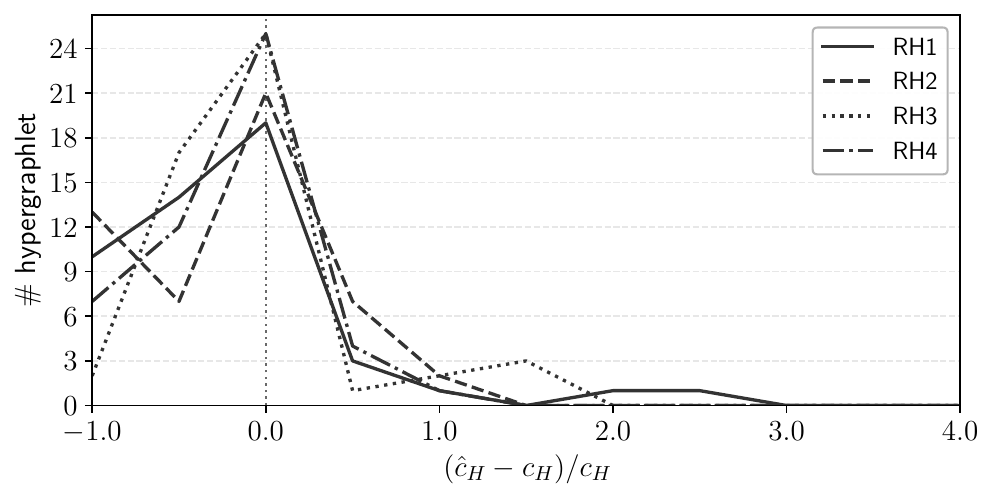}
        
        \vspace{2pt}
        {\small $\err_H$}
    \end{minipage}
    \caption{\Giacomo{Distribution of per-hypergraphlet count errors on four synthetic inputs, in bins in the form $[\err-0.25,\err+0.25)$. To prevent extremely rare hypergraphlets from dominating the tails, we considered only the 50 most frequent ones.}}
    \label{fig:accuracy_k5}
\end{figure}

\subsection{The $k$-hypergraphlet distribution}\label{counting_hypergraphlets}
\Giacomo{Finally, \Cref{fig:K5_frequencies} shows the frequency distribution of $5$-hypergraph\-lets, as estimated by \hypermotivo.
Note that the number of $k$-hypergraphlets explodes super-exponentially with $k$, and moreover for $k>3$ even \emph{depicting} $k$-hypergraphlets is challenging, save for the simplest ones. For this reason we show only the $5$-hypergraphlets with highest aggregate frequency across all datasets.}

\begin{figure}[t]
    \centering
    \includegraphics[width=0.8\linewidth]{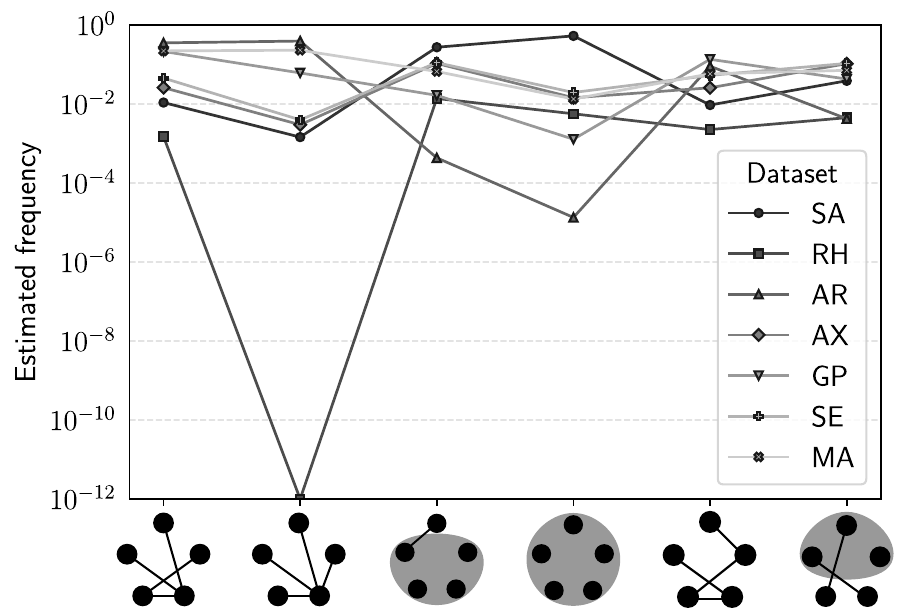}
    \caption{\Giacomo{Frequency distribution of $5$-hypergraphlets with highest aggregate frequency across all hypergraphs.}}
    \label{fig:K5_frequencies}
\end{figure}

%% file: 7_appendix.tex
\appendix

This section complements the main paper by providing additional material that was omitted from the main text due to space constraints.
It is organized as follows:
\begin{itemize}
    \item \Cref{appendix:proofs} contains the proof of \Cref{thm:KSH-Hard}
    \item \Cref{appendix:experiments} reports additional experimental results and the corresponding figures.
\end{itemize}

\ignore{\subsection*{More experiments}
This section of the appendix presents additional plots that complement those shown in \Cref{sec:experiments} but were omitted from the main paper for brevity.
\Cref{fig:time_vs_memory_k3-4-5} reports the ratio between the running time and the peak memory usage of the build-up phase of \Cref{algo:Naive} and \hypermotivo, as in \Cref{sec:experiments_results}. Results for $3$-hypergraphlets and $4$-hypergraphlets are included.
\Cref{fig:absolute_time_k3-4-5} shows the same underlying wall-clock times as \Cref{fig:time_vs_memory_k3-4-5}, but plotted as absolute times for each of the two algorithms, side by side, rather than as speedup ratio. \Cref{fig:K4_frequencies,fig:K5_frequencies} shows more in detail relative frequencies of most frequent $4$ and $5$-hypergraphlets, respectively. Finally, \Cref{fig:threads_math,fig:threads_datascience,fig:threads_arxiv} illustrate how \hypermotivo scales across different datasets under varying numbers of threads and values of $k$.}

\section{Proofs}\label{appendix:proofs}
\subsection*{Proof of \Cref{thm:KSH-Hard}}
This section of the appendix is dedicated to the proof of \Cref{thm:KSH-Hard}. In particular, we prove the following statement.

\begin{theorem}\label{thm:strong_reduction}
    There exists an FPT-reduction from \probKC\ to \probKSH\ that, for every instance $(G,k)$ of \probKC, yields an instance $(H,k')$ of \probKSH\ with the following properties:
    \begin{itemize}
        \item $k' = \binom{k}{2}(k+1) \le \frac{k^3}{2}$.
        \item $H$ is $(k^2,0)$-nice.
    \end{itemize}
    Moreover, the reduction runs in time $O\left(k^2 \cdot \left( |E(G)| + |V(G)| \right)\right)$.
\end{theorem}
This clearly implies that \probKSH is $W[1]$-hard.
\Cref{thm:KSH-Hard} follows from \Cref{thm:strong_reduction} combined with the well-known lower bound $n^{\Omega(k)}$ of \probKC under ETH~\cite{Chen&2006}. 
Note that, by the second item, \probKSH\ remains $W[1]$-hard and is subject to a $n^{\Omega(\sqrt[3]{\kappa})}$ ETH-conditional lower bound even if one takes into account the $(\alpha,\beta)$-niceness of the hypergraph by adopting as a parameter $\kappa=k+\alpha+\beta$.

\ignore{
As a simple corollary, we show that the statement of \Cref{thm:strong_reduction} remains true even under the $(\alpha, \beta)$ parameterization introduced in \Cref{sec4:abnice}. Formally we prove that
\begin{corollary}\label{corollary:strong_reduction_alfabeta}
    The reduction of \Cref{thm:strong_reduction} produces an instance $(H,k')$ of \probKSH such that $H$ is $(O(k^2),0)$-nice.
\end{corollary}
Thus, \probKSH\ remains $W[1]$-hard and inherits the same $n^{o(k')}$ lower bound even under the $(\alpha,\beta)$-parameterization.
}

The rest of this section gives the proof of \Cref{thm:strong_reduction}.

\paragraph{The reduction.}
Let $(G,k)$ be an instance of \probKC.
We construct the corresponding instance $(H,k')$ of \probKSH as follows.
For each vertex $v \in V(G)$, create a set
$Q_v = \{a_v^{1}, \dots, a_v^{\binom{k}{2}}\}$,
and for each edge $e \in E(G)$, create a singleton
$R_e = \{c_e\}$.
Define the hypergraph $H$ by
\begin{equation}
    V(H) = \bigcup_{v \in V(G)} Q_v \;\cup\; \bigcup_{e \in E(G)} R_e
\end{equation}
and
\begin{equation}
    E(H) = \{\, E_{uv} : e = \{u,v\} \in E(G), \; E_{uv} = Q_u \cup Q_v \cup R_e \,\}.
\end{equation}
Finally, we set $k' := (k+1)\binom{k}{2}$, and we take $(H,k')$ as the output instance of \probKSH.

See \Cref{fig:kc-ksh-instance} for a visual representation of the reduction. Intuitively, one can also think of this construction as encoding a weighted graph $G’$ on the same vertex and edge set as $G$, where each vertex has weight $\binom{k}{2}$ and each edge weight $1$. 

We now prove that the construction above is indeed an \ccFPT\ reduction that satisfies the claim of \Cref{thm:strong_reduction}.

\newcommand{\fillQ}[1]{%
  \node[smallvertex] at ($(#1.center)+(-0.18,0.10)$) {};
  \node[smallvertex] at ($(#1.center)+(-0.08,0.04)$) {};
  \node[smallvertex] at ($(#1.center)+(0.10,0.06)$) {};
  \node[smallvertex] at ($(#1.center)+(0.18,-0.02)$) {};
  \node[smallvertex] at ($(#1.center)+(-0.14,-0.06)$) {};
  \node[smallvertex] at ($(#1.center)+(0.00,-0.12)$) {};
}

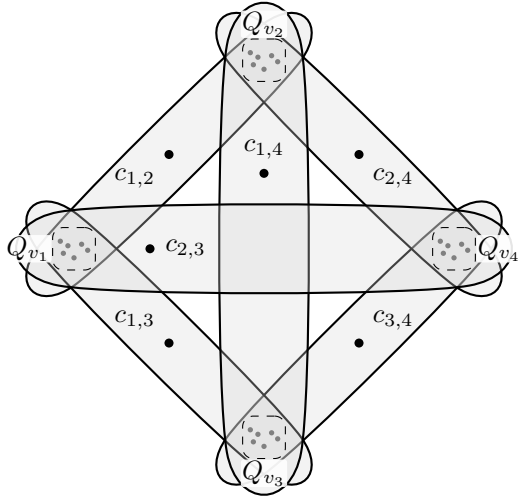
\begin{figure}[t]
  \centering

\ignore{
  \begin{subfigure}{0.45\textwidth}
    \centering
    \begin{tikzpicture}[
        scale=1.0,
        vertex/.style={circle,fill=black,inner sep=1.5pt}
      ]

      \node[vertex,label=above left:$v_1$] (v1) at (0,0) {};
      \node[vertex,label=above:$v_2$]      (v2) at (1.7,1.7) {};
      \node[vertex,label=below:$v_3$]      (v3) at (1.7,-1.7) {};
      \node[vertex,label=above right:$v_4$](v4) at (3.4,0) {};

      \draw (v1) -- (v2);
      \draw (v1) -- (v3);
      \draw (v1) -- (v4);
      \draw (v2) -- (v3);
      \draw (v2) -- (v4);
      \draw (v3) -- (v4);

    \end{tikzpicture}
    \caption{\Giacomo{descrizione o è sufficiente quella della figure?}}
    \label{fig:kc-instance}
  \end{subfigure}
  \hfill
}
    \centering
    \begin{tikzpicture}[
        vertex/.style={circle,fill=black,inner sep=1.2pt},
        smallvertex/.style={circle,fill=gray,inner sep=0.7pt},
        block/.style={draw,dashed,rounded corners,inner sep=8pt},
        cnode/.style={vertex},
        hyperedge/.style={draw,thick,fill=gray!30,fill opacity=0.3},
        every node/.style={font=\small},
        labelbox/.style={fill=white,fill opacity=0.9,text opacity=1,inner sep=0.5pt,rounded corners=1pt}
      ]

      \node[block,
            label={[labelbox]left:$Q_{v_1}$}] (Q1) at (0,0) {};
      \node[block,
            label={[labelbox]above:$Q_{v_2}$}] (Q2) at (2.5,2.5) {};
      \node[block,
            label={[labelbox]below:$Q_{v_3}$}] (Q3) at (2.5,-2.5) {};
      \node[block,
            label={[labelbox]right:$Q_{v_4}$}](Q4) at (5,0) {};

      \fillQ{Q1};
      \fillQ{Q2};
      \fillQ{Q3};
      \fillQ{Q4};

      \node[cnode,label=below left:$c_{1,2}$] (c12) at ($(Q1)!0.5!(Q2)$) {};
      \node[cnode,label=above left:$c_{1,3}$] (c13) at ($(Q1)!0.5!(Q3)$) {};
      \node[cnode,label=above:$c_{1,4}$] (c14) at (2.5, 1) {};
      \node[cnode,label=right:$c_{2,3}$] (c23) at (1, 0) {};
      \node[cnode,label=below right:$c_{2,4}$] (c24) at ($(Q2)!0.5!(Q4)$) {};
      \node[cnode,label=above right:$c_{3,4}$] (c34) at ($(Q3)!0.5!(Q4)$) {};

      \begin{scope}[on background layer]
        \draw[hyperedge]
          plot[smooth cycle, tension=0.5] coordinates {
            ($(Q1.south)+(0,-0.2)$)
            ($(Q1.west)+(-0.2,0)$)
            ($(Q2.north)+(0,0.2)$)
            ($(Q2.east)+(0.2,0)$)
          };
      \end{scope}
      \begin{scope}[on background layer]
        \draw[hyperedge]
          plot[smooth cycle, tension=0.5] coordinates {
            ($(Q3.south)+(0,-0.2)$)
            ($(Q3.west)+(-0.2,0)$)
            ($(Q4.north)+(0,0.2)$)
            ($(Q4.east)+(0.2,0)$)
          };
      \end{scope}

      \begin{scope}[on background layer]
        \draw[hyperedge]
          plot[smooth cycle, tension=0.5] coordinates {
            ($(Q1.north)+(0,0.2)$)
            ($(Q1.west)+(-0.2,0)$)
            ($(Q3.south)+(0,-0.2)$)
            ($(Q3.east)+(0.2,0)$)
          };
      \end{scope}
      \begin{scope}[on background layer]
        \draw[hyperedge]
          plot[smooth cycle, tension=0.5] coordinates {
            ($(Q2.north)+(0,0.2)$)
            ($(Q2.west)+(-0.2,0)$)
            ($(Q4.south)+(0,-0.2)$)
            ($(Q4.east)+(0.2,0)$)
          };
      \end{scope}
      \begin{scope}[on background layer]
        \draw[hyperedge]
          plot[smooth cycle, tension=0.5] coordinates {
            ($(Q2.east)+(0.2,0.2)$)
            ($(Q2.west)+(-0.2,0.2)$)
            ($(Q3.west)+(-0.2,-0.2)$)
            ($(Q3.east)+(0.2,-0.2)$)
          };
      \end{scope}
      \begin{scope}[on background layer]
        \draw[hyperedge]
          plot[smooth cycle, tension=0.5] coordinates {
            ($(Q1.north)+(-0.2,0.2)$)
            ($(Q1.south)+(-0.2,-0.2)$)
            ($(Q4.south)+(0.2,-0.2)$)
            ($(Q4.north)+(0.2,0.2)$)
          };
      \end{scope}

    \end{tikzpicture}

    \caption{The hypergraph $H$ obtained by the reduction of \Cref{thm:strong_reduction} on $G=K_4$.}
  \label{fig:kc-ksh-instance}
\end{figure}

\begin{claim}
The construction above runs in time $O\left(k^2 \cdot \left( |E(G)| + |V(G)| \right)\right)$.
\end{claim}
\begin{proof}
    Let $(G,k)$ be an instance of \probKC. For each vertex $v \in V(G)$ we create the set $Q_v$ of size $\binom{k}{2}$, which takes a total of $O(k^2 \cdot |V(G)|)$ time. Then, for each edge $e = \{u,v\} \in E(G)$ we create the vertex $c_{e}$ and the hyperedge $E_{uv} = Q_u \cup Q_v \cup \{c_{e}\}$, of size $2\binom{k}{2} + 1 = O(k^2)$, for a total time of $O(k^2 \cdot |E(G)|)$. Summing these two bounds gives the claimed running time $O(k^2 \cdot (|V(G)| + |E(G)|))$.
\end{proof}

The next observation is useful to prove the correctness of the reduction. From now on, $H[U]$ refers to the section hypergraph induced by $U \subseteq V(H)$, as formalized in \Cref{def:SectionHypergraph}.
\begin{oss}\label{oss:strong_reduction_obs}
    Let $(G,k)$ be an instance of \probKC and $(H,k')$ the instance of \probKSH obtained via the reduction of \Cref{thm:strong_reduction}. Let $U\subseteq V(H)$ be such that $H[U]$ is connected and $|U|\geq 2$. Then:
    \begin{itemize}
        \item For every $v \in V(G)$, either $Q_v \cap U =\emptyset$ or $Q_v \subseteq U$.
        \item For every $e=\{u,v\} \in E(G)$, if $c_e \in U$ then $Q_u \cup Q_v \subseteq U$.
    \end{itemize}
\end{oss}
\begin{proof}
For the first item, suppose for a contradiction that there exists $v \in V(G)$ such that $Q_v \cap U \neq \emptyset$ but $Q_v \nsubseteq U$. Take any $x \in Q_v \cap U$ and any $y \in Q_v \setminus U$. Every hyperedge of $H$ containing $x$ is of the form $E_{vw} = Q_v \cup Q_w \cup R_{e}$, where $e = \{v,w\}$ for some neighbor $w$ of $v$ in $G$. However, since $y \notin U$, no such hyperedge $E_{vw}$ is entirely contained in $U$, and therefore no hyperedge containing $x$ appears in $H[U]$. Hence $x$ is an isolated vertex in $H[U]$. As $|U| \ge 2$ and $H[U]$ is connected, this is impossible. Thus either $Q_v \subseteq U$ or $Q_v \cap U = \emptyset$.

For the second item, if $c_e \in U$ but some vertex of $Q_u$ or $Q_v$ is not in $U$, then the hyperedge $E_{u,v} = Q_u \cup Q_v \cup R_{e}$ is not entirely contained in $U$. Thus no hyperedge of $H[U]$ contains $c_e$, and $c_e$ would be isolated in $H[U]$. Again this contradicts the connectivity of $H[U]$ (for $|U| \ge 2$), so we must have $Q_u \cup Q_v \subseteq U$.
\end{proof}

Intuitively, the observation above states that $U$ cannot "partially pick" a $Q_v$: as soon as $U$ contains one vertex of $Q_v$, then it must contain the whole block. Similarly, $U$ cannot contain a vertex $c_e$ without also containing the whole blocks $Q_u$ and $Q_v$.

\begin{claim}
The construction above is an \ccFPT reduction from \probKC to \probKSH.
\end{claim}
\begin{proof}
We prove that $(G,k)$ is a YES-instance of \probKC if and only if $(H,k')$ is a YES-instance of \probKSH. Recall that a YES-instance of \probKSH means that there exists $U \subseteq V(H)$ such that $|U| = k'$ and the section hypergraph $H[U]$ is connected. 
    
    \paragraph*{($\Rightarrow$) From \probKC to \probKSH.} Assume that $G$ contains a $k$-clique $K \subseteq V(G)$; define $U$ as 
    \begin{equation}
        U = \bigcup_{e=\{u,v\} \in E(G[K])}E_{uv} = \bigcup_{e=\{u,v\} \in E(G[K])}\left(Q_u \cup Q_v \cup R_e \right),
    \end{equation}
    where $G[K]$ is the subgraph induced by $K$. 
    By construction, each vertex $v \in K$ contributes $\binom{k}{2}$ vertices through $Q_v$ and each edge $e=\{u,v\} \in E(G[K])$ contributes exactly one vertex through $R_{e}$. Clearly, since $K$ has $\binom{k}{2}$ edges, $|U| = k\binom{k}{2} + \binom{k}{2} = (k+1)\binom{k}{2} = k'$. 

    We now show that $H[U]$ is connected. For every edge $e=\{u,v\} \in E(G[K])$ we have $E_{uv} \subseteq U$, hence $E_{uv}$ is an edge of $H[U]$. In particular, for every pair of distinct vertices $u,v \in K$ there is a hyperedge $E_{uv}$ contained in $H[U]$ that includes all vertices of $Q_u$, all vertices of $Q_v$, and the vertex $c_e$ corresponding to $e=\{u,v\}$; thus, $H[U]$ is connected. 

    \paragraph*{($\Leftarrow$) From \probKSH to \probKC.} 
    Assume that $(H,k')$ is a YES-instance of \probKSH. Thus, there exists $U \subseteq V(H)$ such that $|U| = k'$ and $H[U]$ is connected. Intuitively, in what follows we decompose $U$ into whole vertex blocks $Q_v$ and edge vertices $c_e$, and then use counting and connectivity arguments to show that $U$ must correspond to exactly $k$ such blocks that are pairwise adjacent in $G$, that is, to a $k$-clique.
    
    First, define $T$ and $S$ as follows
    \begin{equation}
        T = \{v \in V(G) : Q_v\subseteq U\} \ \ \text{and} \ \ S =\{e \in E(G): c_e \in U\}.
    \end{equation}
    Each vertex $v \in T$ contributes all $\binom{k}{2}$ vertices of $Q_v$ to $U$, and each edge $e \in S$ contributes the vertex $c_e$.
    By construction of $H$ we have
    \begin{equation}
        V(H) = \bigcup_{v \in V(G)} Q_v \;\cup\; \bigcup_{e \in E(G)} R_e,
    \end{equation}
    so every vertex of $U$ lies either in some $Q_v$ or in some $R_e = \{c_e\}$.
    Moreover, by Observation~1, if $v \notin T$ then $Q_v \cap U = \emptyset$, and if $c_e \in U$ then $e \in S$.
    Hence
    \begin{equation}
        U = \bigcup_{v \in T} Q_v \;\cup\; \{c_e : e \in S\},
    \end{equation}
    and we can express the size of $U$ as
    \begin{equation}
        |U| = |T|\binom{k}{2} + |S|.
    \end{equation}
    Call $t = |T|$ and $s = |S|$. Using $|U|=k'=(k+1)\binom{k}{2}$, we obtain
    \begin{equation}\label{eq:s-k-t}
      (k+1)\binom{k}{2} = t\binom{k}{2} + s \qquad\Longrightarrow\qquad s = (k+1 - t)\binom{k}{2}.
    \end{equation}
    
    Next, note that Observation~\ref{oss:strong_reduction_obs} also implies that every edge $e = \{u,v\} \in S$ has both endpoints in $T$: indeed, since $c_e \in U$, the second item of Observation~\ref{oss:strong_reduction_obs} gives $Q_u \cup Q_v \subseteq U$, so $u,v \in T$ by definition of $T$. Hence $S \subseteq \binom{T}{2}$. Let $G[T]$ be the graph with vertex set $T$ and edge set $S$.
    By construction of $H$, the connected components of $H[U]$ correspond exactly to the connected components of $G[T]$: contracting each block $Q_v$ to a single
    vertex $v$ and deleting the vertices $\{c_e\}$ yields $G[T]$, and conversely every edge of $G[T]$ corresponds to a hyperedge $E_{uv}$ in $H[U]$. Since $H[U]$ is connected, the graph $G[T]$ is connected.
    Therefore,
    \begin{equation}\label{eq:s-lower}
      s \ge t-1,
    \end{equation}
    because any connected graph on $t$ vertices has at least $t-1$ edges. On the other hand, we clearly have
    \begin{equation}\label{eq:s-upper}
      s \le \binom{t}{2}.
    \end{equation}
    We now derive constraints on $t$ using \cref{eq:s-k-t,eq:s-lower,eq:s-upper}. In particular, we show that $t = k$, meaning that $s = \binom{k}{2}$ and thus $G$ contains a $k$ clique. We may assume $k \ge 3$, since for $k \le 2$ both \probKC and \probKSH are trivial.

    \emph{Case 1: $t \ge k+1$.}
    From \cref{eq:s-k-t} we obtain $s = (k+1-t)\binom{k}{2} \le 0$.
    However, $H[U]$ is connected and $|U| = k' > 1$, so $H[U]$ must contain at least one hyperedge, which implies $s \ge 1$. This is a contradiction. Thus $t \le k$.

    \emph{Case 2: $t \leq k-1$.} From \cref{eq:s-k-t} we have $k+1-t \ge 2$, and hence
    \begin{equation}
        s = (k+1-t)\binom{k}{2} \;\ge\; 2\binom{k}{2} \;=\; k(k-1).
    \end{equation}
    On the other hand, since $s \le \binom{t}{2}$ and $t \le k-1$, we have
    \begin{equation}
      k(k-1) \leq s  \le \binom{k-1}{2} \;=\; \frac{(k-1)(k-2)}{2}.
    \end{equation}
    But clearly, for $k \ge 3$ we obtain
    \begin{equation}
      k(k-1) > \frac{(k-1)(k-2)}{2},
    \end{equation}
    which leads to a contradiction. Therefore $t \not\le k-1$. The only remaining possibility is $t = k$. Finally, substituting into~\eqref{eq:s-k-t} we get
    \begin{equation}
        s = (k+1-k)\binom{k}{2} = \binom{k}{2}.
    \end{equation}
    Since $s = |S|$ and $S \subseteq \binom{T}{2}$ with $|T| = k$, this means that $S$ contains all possible $\binom{k}{2}$ edges on the vertex set $T$. Therefore $G[T]$ is a clique on $k$ vertices, and hence $G$ contains a $k$-clique. This is sufficient to prove the claim.
\end{proof}

\begin{claim}
The construction above ensures that $k' = (k+1)\binom{k}{2} \le \frac{k^3}{2}$ and that $H$ is $(k^2,0)$-nice.
\end{claim}    
\begin{proof}
First point follows simply by definition of the reduction; indeed we have $k' = (k+1)\binom{k}{2} \leq \frac{k^2 \cdot k}{2} = \frac{k^3}{2}$ for $k \ge 1$.

We now show that $H$ is $(k^2,0)$-nice. By construction, every hyperedge of $H$ is of the form
\begin{equation}
  E_{uv} = Q_u \cup Q_v \cup R_e,
\end{equation}
where $e = \{u,v\} \in E(G)$. Thus
\begin{equation}
  |E_{uv}| = |Q_u| + |Q_v| + |R_e|
           = \binom{k}{2} + \binom{k}{2} + 1
           = k(k-1) + 1,
\end{equation}
and hence
\begin{equation}
  k(k-1) + 1 = k^2 - k + 1 \le k^2
\end{equation}
for every $k \ge 1$. Therefore the rank of $H$ (the maximum size of a hyperedge) is at most $k^2$.

Consider the $\alpha$-split of $H$ with $\alpha := k^2$. Since no hyperedge has size greater than $\alpha$, the upper part $H_{>\alpha}$ contains no hyperedges, and in particular every vertex has degree $0$ in $H_{>\alpha}$. By the definition of $(\alpha,\beta)$-niceness, this means that $H$ is $(k^2,0)$-nice.
\end{proof}

\clearpage

\section{Additional Experimental Results}\label{appendix:experiments}
This section presents additional experimental results that complement the evaluation in Section 7
We extend the analysis along four main directions.

First, we provide further measurements of the build-up phase, including results for smaller hypergraphlet sizes ($k \in \{3,4\}$), absolute running times, and a detailed breakdown of the build-up cost across datasets, which complement the speedup analysis reported in the main paper.

Second, we report a more detailed characterization of hypergraphlet frequency distributions, including the most frequent $4$- and $5$-hypergraphlets across all datasets.

Third, we include an extended scalability study that evaluates the parallel performance of \hypermotivo{} across multiple datasets and values of $k$, covering a wider range of thread configurations and hypergraphlet sizes than those shown in the main paper.

Finally, we provide additional accuracy results, such as per-dataset error histograms and metrics for different sampling budgets and values of $k$, which refine and corroborate the empirical findings reported in Section 7.

Overall, these experiments offer a more complete picture of the performance, (parallel) scalability, and accuracy of hypergraphlet counting with \hypermotivo.

\clearpage

\subsection{Speedup and absolute build-up timings}

\Cref{fig:mem_time_and_abs_timings} shows, on the left, the speedup computed across all datasets. In addition to the case $k=5$, which is already reported in the paper (see Section 7.2.2), we also include results for $k \in \{3, 4\}$. As observed in the paper, the speedup is consistent across different values of k. Higher values of  k are computationally too demanding for the largest graphs. The right side of \Cref{fig:mem_time_and_abs_timings} instead reports the absolute running times. \Cref{fig:time_breakdown} reports absolute timing breakdown for each dataset and $k\in\{3,4,5\}$. For every dataset, the left group of bars reports the preprocessing and build-up times of \motivo, whereas the right group reports the preprocessing, build-up, and inclusion--exclusion times of \hypermotivo. For \motivo, preprocessing includes projecting the input hypergraph $H$ to its Gaifman graph $\Gaif(H)$. For \hypermotivo, preprocessing includes selecting $\alpha$, splitting $H$ into $H_{\le \alpha}$ and $H_{>\alpha}$, and projecting $\Gaif(H_{\le \alpha})$.

\vspace{-0.5em}
\begin{figure*}[!t]
  \centering
  \setlength{\abovecaptionskip}{4pt}
  \setlength{\belowcaptionskip}{0pt}

  \newcommand{\panelh}{0.17\textheight}
  \newcommand{\leftpanel}[1]{\includegraphics[height=\panelh,keepaspectratio]{#1}}
  \newcommand{\rightpanel}[1]{\includegraphics[height=\panelh,keepaspectratio]{#1}}

  \begin{tabular}{@{}c@{\hspace{0.03\textwidth}}c@{}}

    \leftpanel{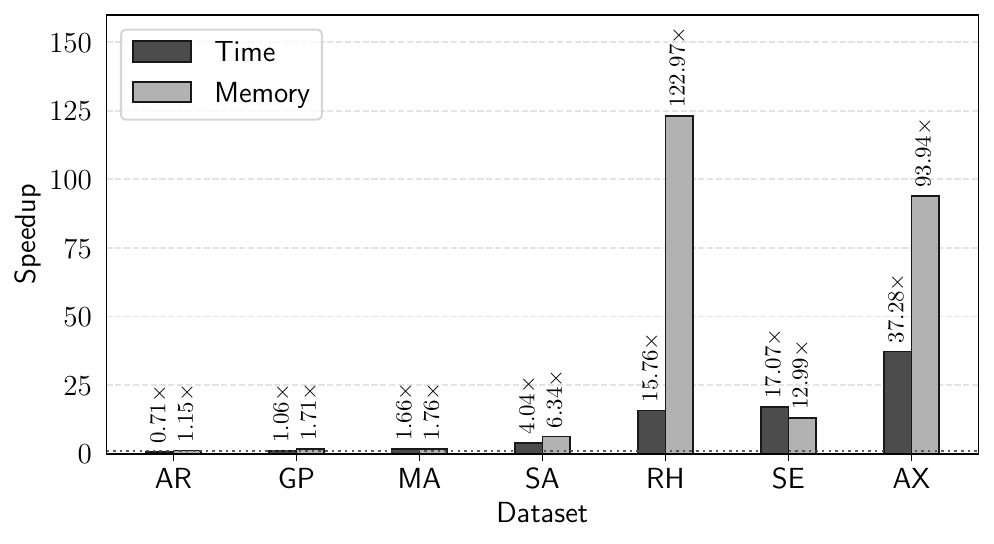} &
    \rightpanel{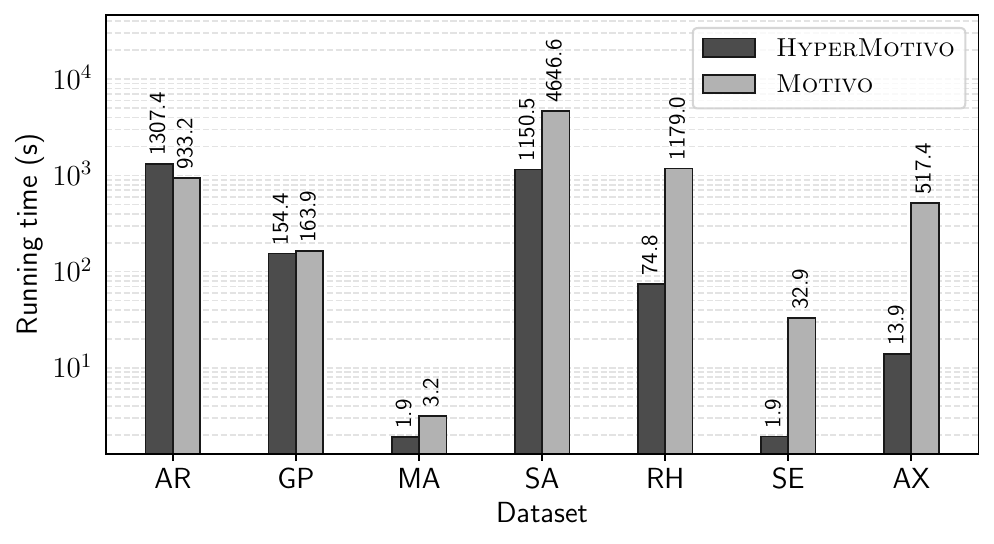} \\
    {\small Time vs.\ memory ($k=3$)} & {\small Absolute timings ($k=3$)} \\[0.55em]

    \leftpanel{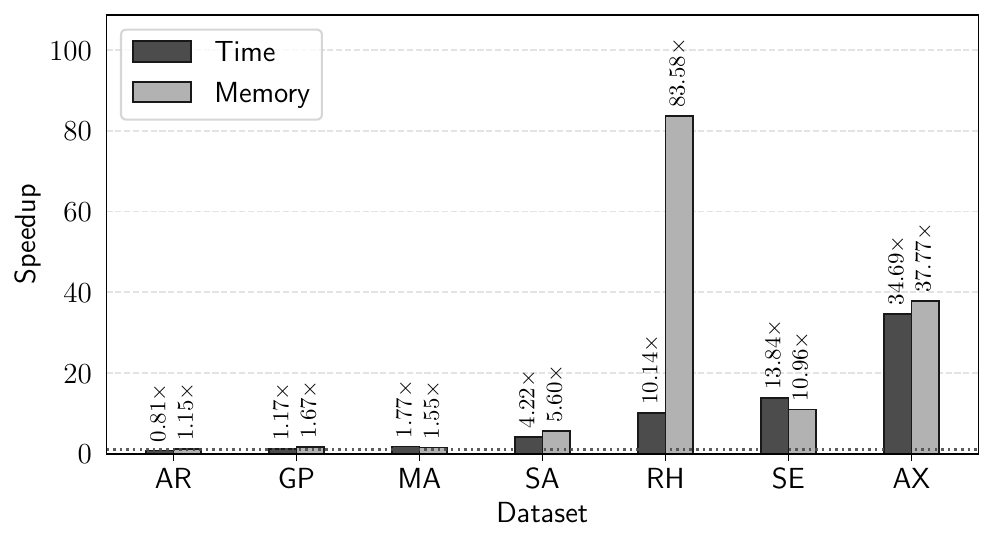} &
    \rightpanel{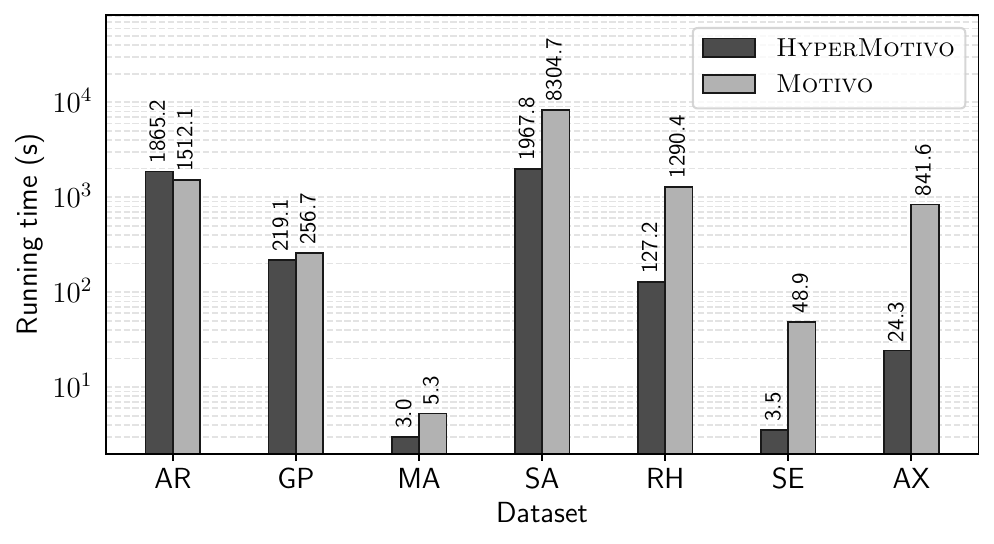} \\
    {\small Time vs.\ memory ($k=4$)} & {\small Absolute timings ($k=4$)} \\[0.55em]

    \leftpanel{figures/timings/K5_timings_mem.pdf} &
    \rightpanel{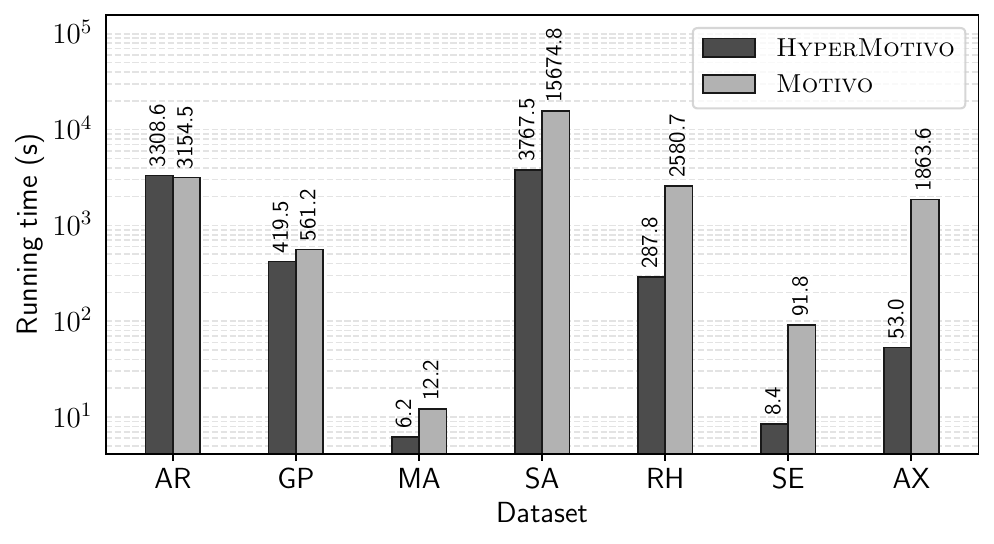} \\
    {\small Time vs.\ memory ($k=5$)} & {\small Absolute timings ($k=5$)}
  \end{tabular}

  \caption{Memory--time tradeoff and absolute build-up timings for $k \in \{3,4,5\}$}
  \label{fig:mem_time_and_abs_timings}
\end{figure*}

\begin{figure}
    \centering
     \begin{subfigure}[t]{\textwidth}
        \centering
        \includegraphics[width=0.7\linewidth]{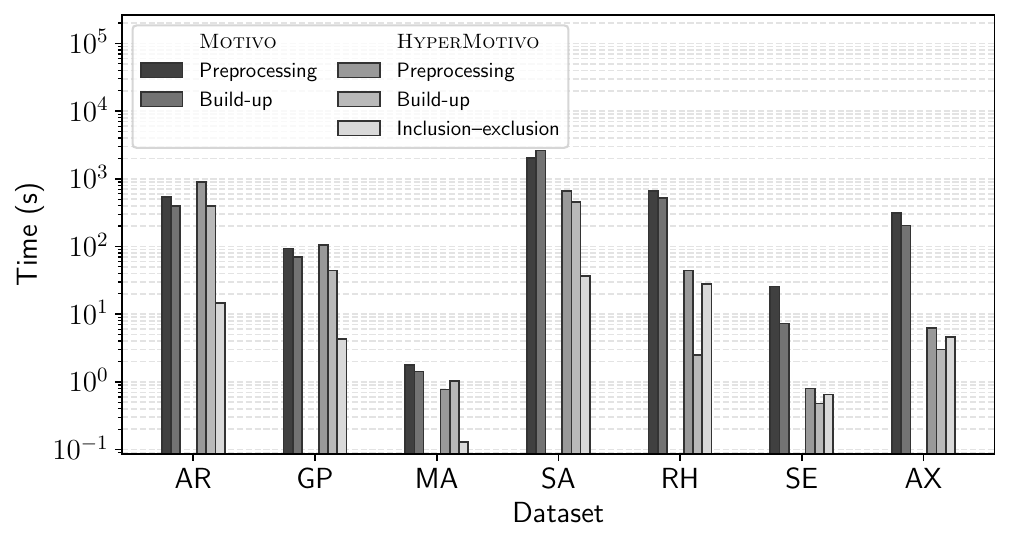}
        \label{fig:time_breakdown_k3}
      \end{subfigure}
      \begin{subfigure}[t]{\textwidth}
        \centering
        \includegraphics[width=0.7\linewidth]{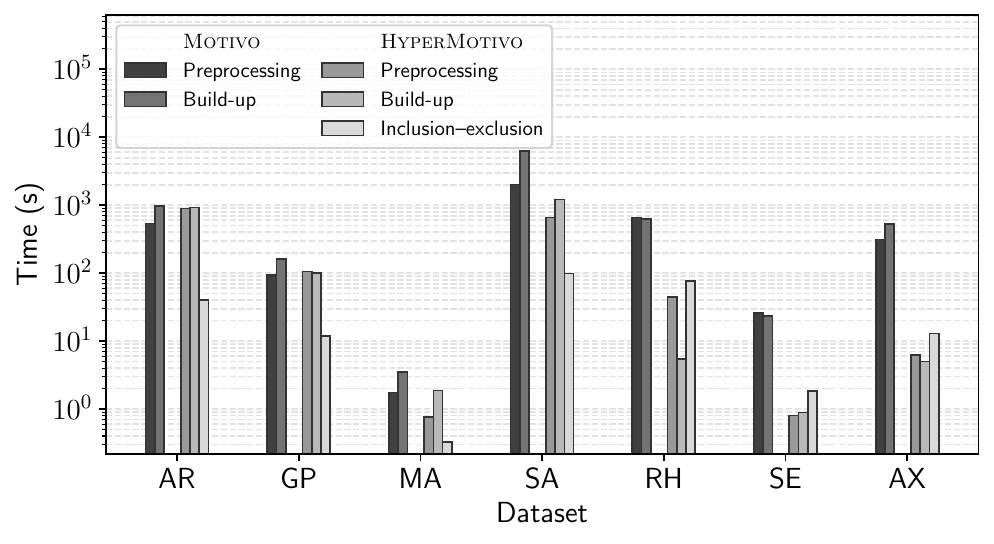}
        \label{fig:time_breakdown_k4}
      \end{subfigure}
      \begin{subfigure}[t]{\textwidth}
        \centering
        \includegraphics[width=0.7\linewidth]{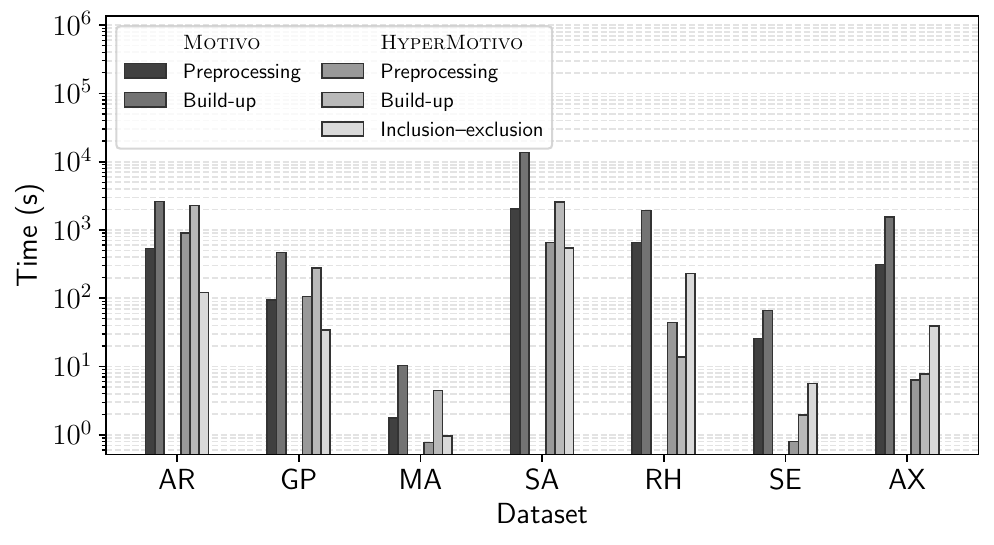}
        \label{fig:time_breakdown_k5}
      \end{subfigure}
    \caption{Timings breakdown for $k \in \{3,4,5\}.$}
    \label{fig:time_breakdown}
\end{figure}

\clearpage

\subsection{Hypergraphlet Frequencies}
\Cref{fig:hypergraphlet_frequencies} shows frequency distributions of $3,4$ and $5$-hypergraphlets with highest aggregate frequency across all hypergraphs. \Cref{fig:K4K5_frequencies} reports the frequencies of the 20 most frequent $k$-hypergraphlets for $k \in \{4,5\}$ across all datasets.

\begin{figure*}[t]
    \centering
     \begin{subfigure}[t]{0.49\textwidth}
        \centering
        \includegraphics[width=\linewidth]{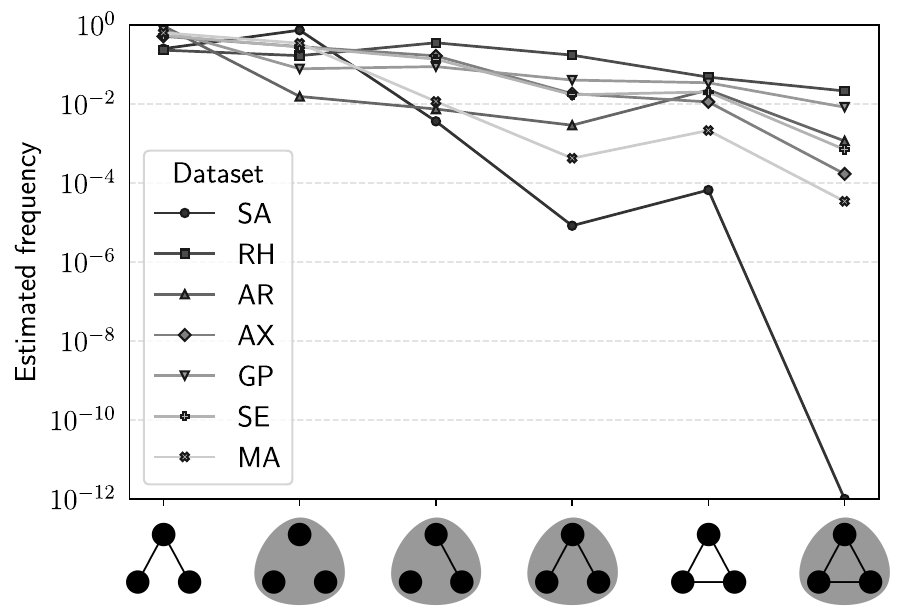}
        \caption{k = 3}
        \label{fig:hypergraphlet_frequencies_k3}
      \end{subfigure}
      \begin{subfigure}[t]{0.49\textwidth}
        \centering
        \includegraphics[width=\linewidth]{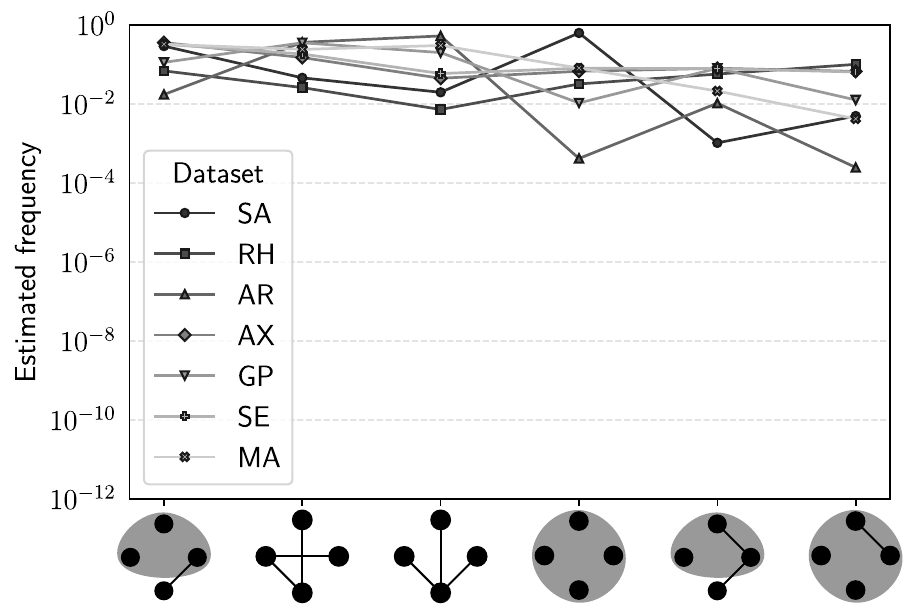}
        \caption{k = 4}
        \label{fig:hypergraphlet_frequencies_k4}
      \end{subfigure}
      \begin{subfigure}[t]{0.49\textwidth}
        \centering
        \includegraphics[width=\linewidth]{figures/frequencies/K5_hypergraphlet_frequencies_all_log.pdf}
        \caption{k = 5}
        \label{fig:hypergraphlet_frequencies_k5}
      \end{subfigure}
    \caption{Frequency distributions of $k$-hypergraphlets for $k \in \{3,4,5\}$.}
    \label{fig:hypergraphlet_frequencies}
\end{figure*}

\begin{figure*}[t]
  \centering

  \begin{minipage}[t]{0.49\textwidth}
    \centering
    \vspace{0pt}

    \includegraphics[width=\linewidth]{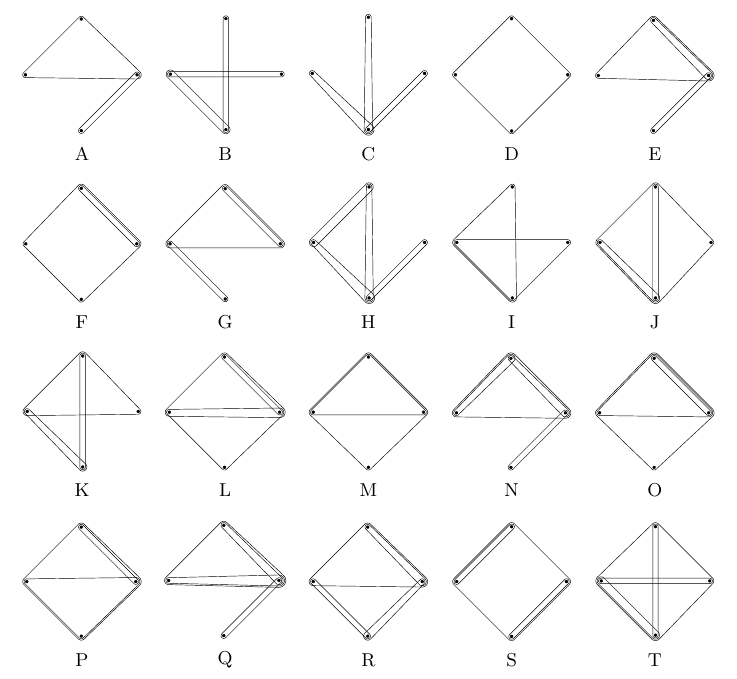}

    \vspace{0.6em}

    \resizebox{\linewidth}{!}{%
      \input{figures/frequencies/k4_table}
    }
  \end{minipage}\hfill
  \begin{minipage}[t]{0.49\textwidth}
    \centering
    \vspace{0pt}

    \includegraphics[width=\linewidth]{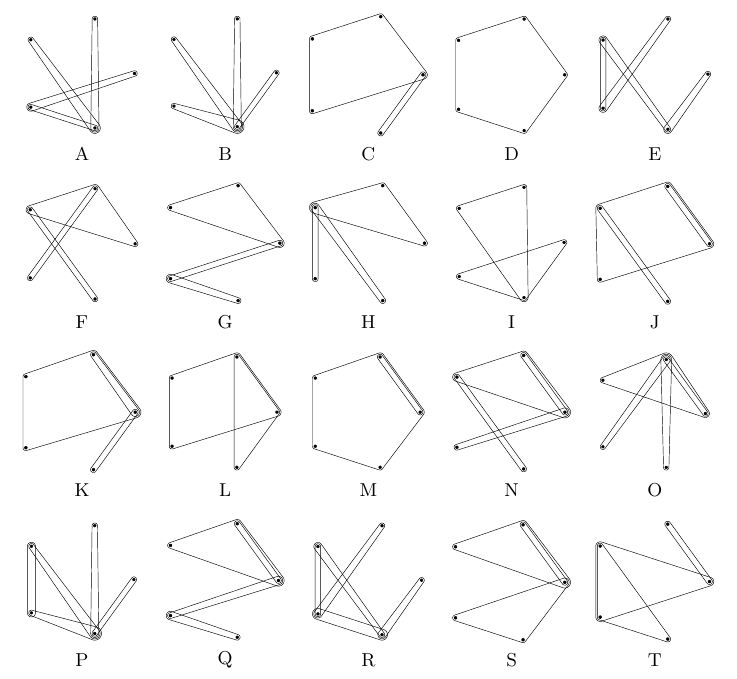}

    \vspace{0.6em}

    \resizebox{\linewidth}{!}{%
      \input{figures/frequencies/k5_table}
    }
  \end{minipage}

  \caption{First 20 $4$- and $5$-hypergraphlets with highest aggregate frequency across all datasets. Left: $K_4$. Right: $K_5$.}
  \label{fig:K4K5_frequencies}
\end{figure*}
\clearpage

\subsection{Parallel scalability}

Here we present parallel scalability experiments for the datasets
MA (\Cref{fig:threads_MA}), SE (\Cref{fig:threads_SE}), and AX (\Cref{fig:threads_AX}),
for $k \in \{3,\dots,8\}$.
We observe that scalability is consistent across all values of $k$ and across all datasets.
In some cases, when using 24 threads, a slight performance degradation is observed compared
to 16 threads, likely due to communication overhead.

\vspace{-0.5em}
\begin{figure*}[!t]
  \centering

  \setlength{\abovecaptionskip}{4pt}
  \setlength{\belowcaptionskip}{0pt}

  \newcommand{\scalplot}[1]{\includegraphics[height=0.175\textheight,keepaspectratio]{#1}}

  \begin{tabular}{@{}c@{\hspace{0.03\textwidth}}c@{}}

    \scalplot{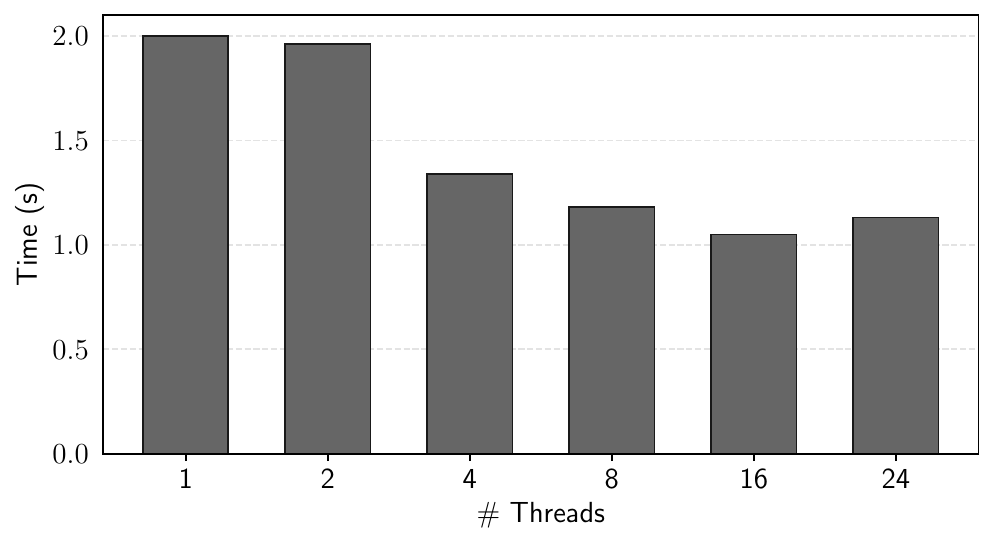} &
    \scalplot{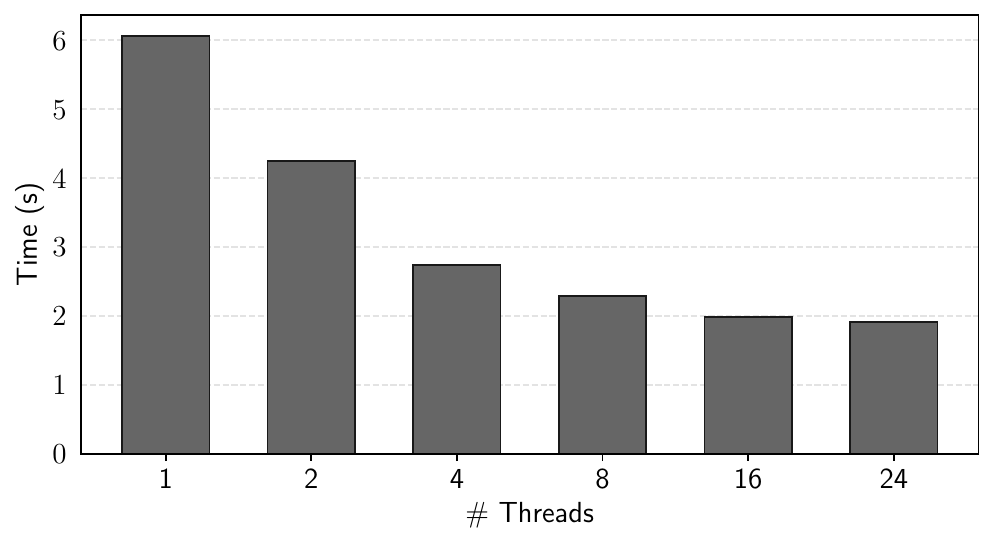} \\
    {\small $k=3$} & {\small $k=4$} \\[0.55em]

    \scalplot{figures/threads/math_algo_time_vs_threads_k5.pdf} &
    \scalplot{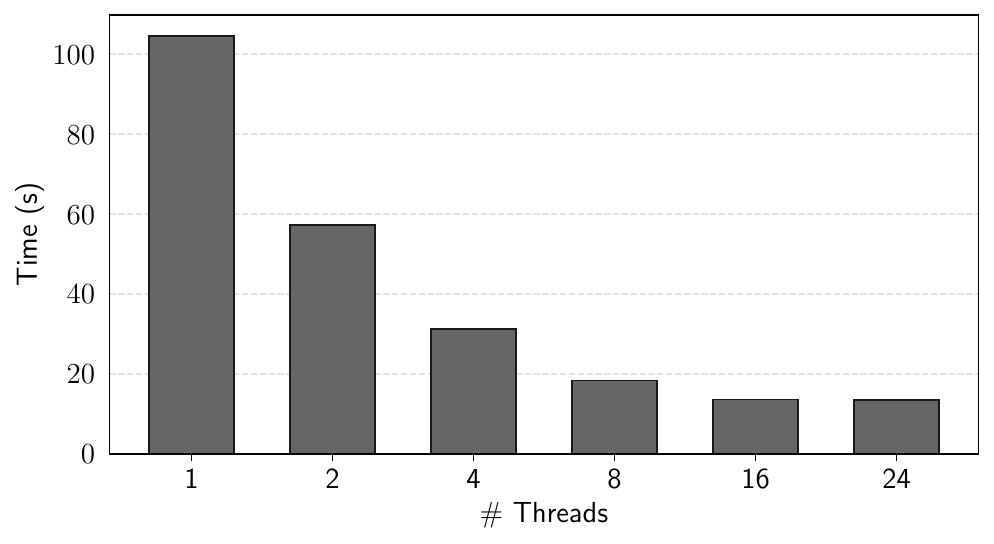} \\
    {\small $k=5$} & {\small $k=6$} \\[0.55em]

    \scalplot{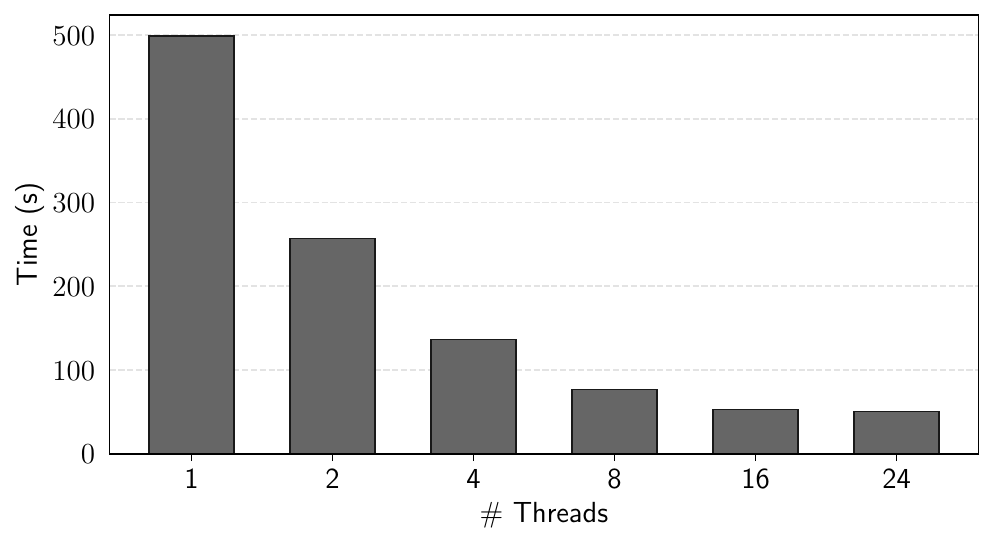} &
    \scalplot{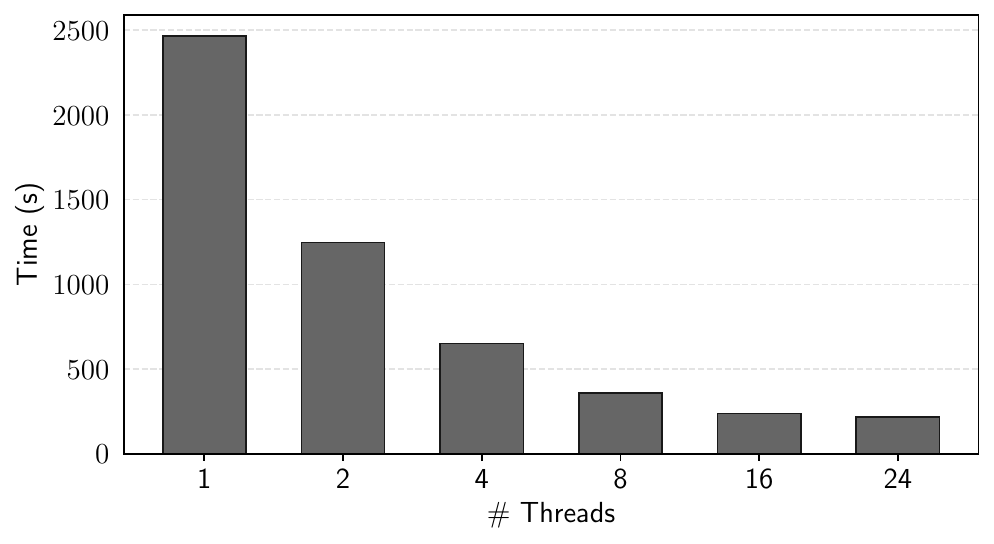} \\
    {\small $k=7$} & {\small $k=8$}

  \end{tabular}
  \caption{Running time of \hypermotivo's build-up phase on the MA dataset for $k \in \{3,\dots,8\}$.}
  \label{fig:threads_MA}
\end{figure*}
\vspace{-0.8em}

\begin{figure*}[p]
  \centering
  \includegraphics[width=0.48\textwidth]{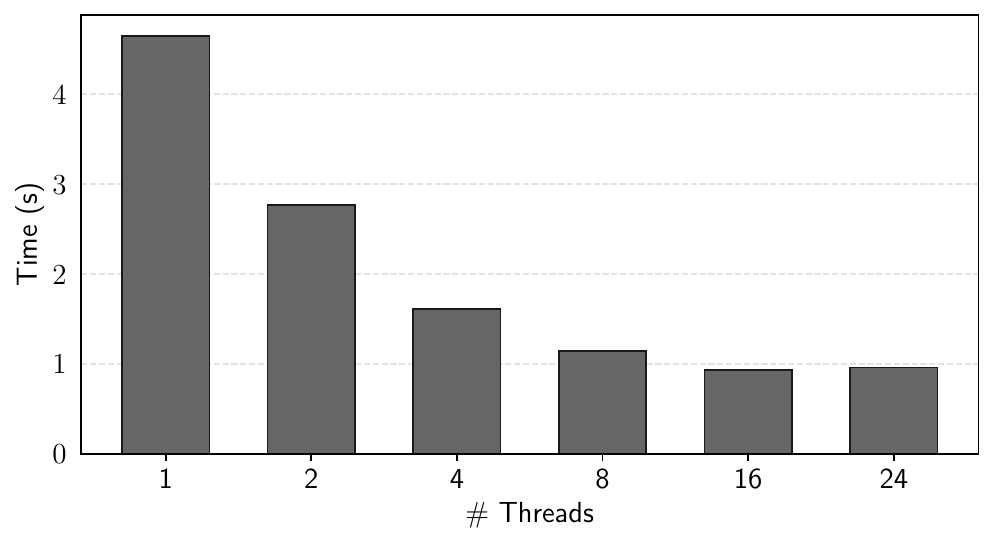}\hfill
  \includegraphics[width=0.48\textwidth]{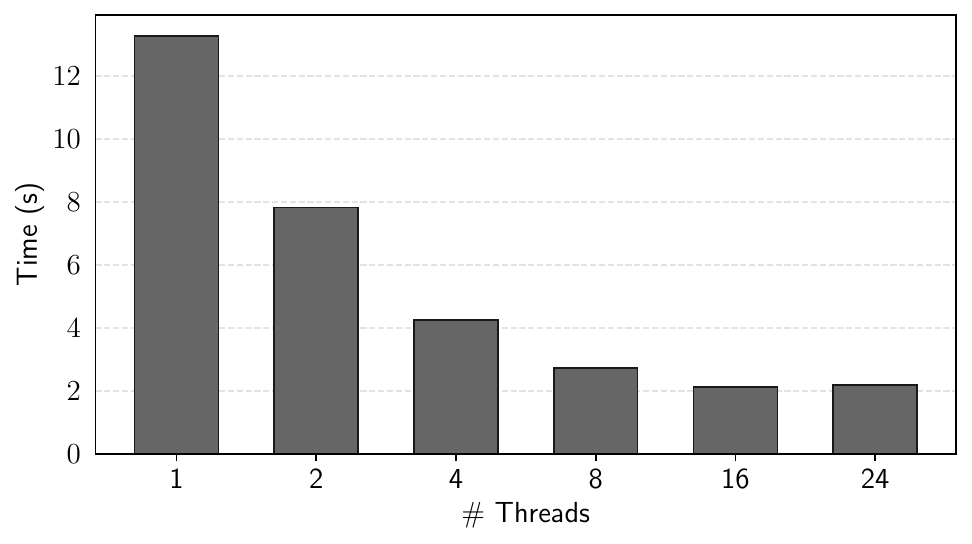}

  \par\vspace{-0.3em}
  \makebox[0.48\textwidth][c]{\small $k=3$}\hfill
  \makebox[0.48\textwidth][c]{\small $k=4$}

  \vspace{1.0em}

  \includegraphics[width=0.48\textwidth]{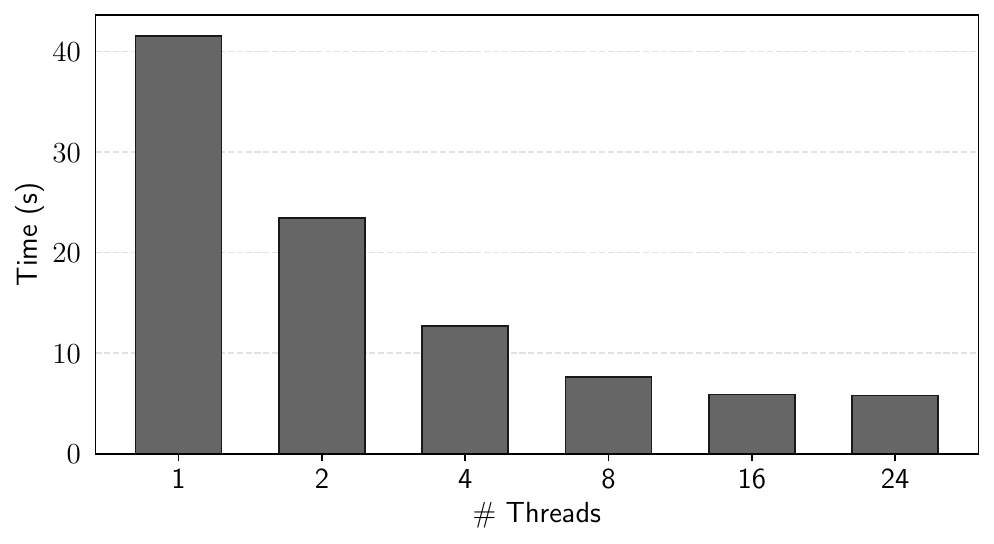}\hfill
  \includegraphics[width=0.48\textwidth]{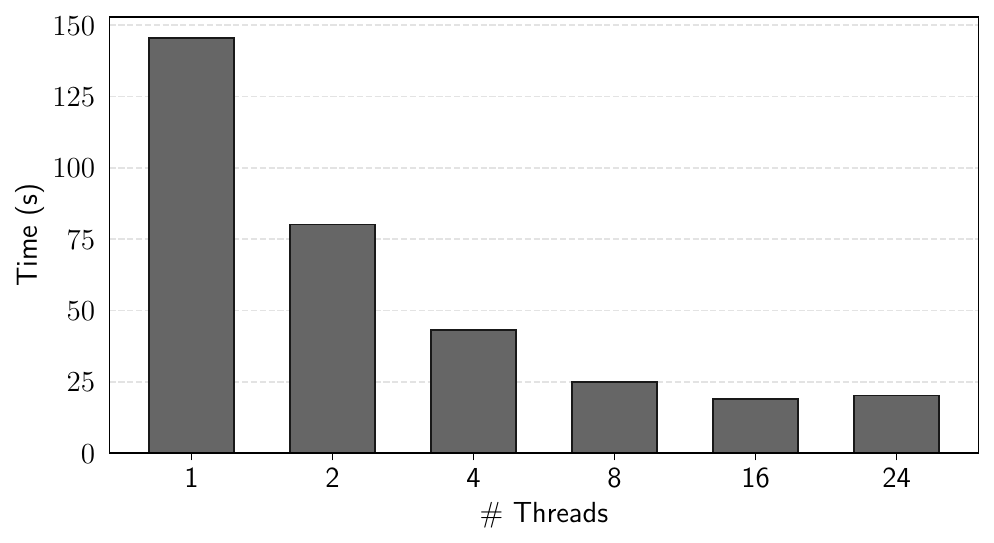}

  \par\vspace{-0.3em}
  \makebox[0.48\textwidth][c]{\small $k=5$}\hfill
  \makebox[0.48\textwidth][c]{\small $k=6$}

  \vspace{1.0em}

  \includegraphics[width=0.48\textwidth]{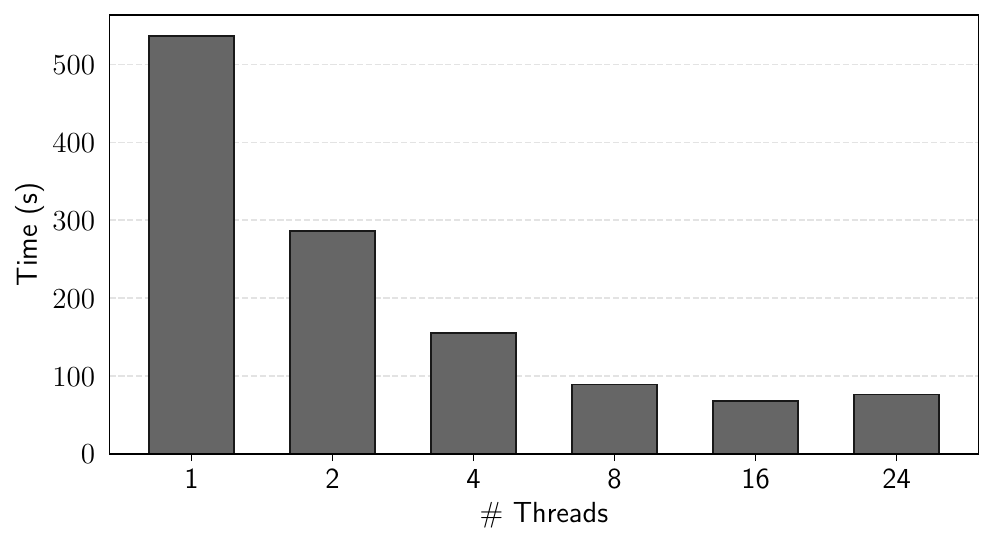}\hfill
  \includegraphics[width=0.48\textwidth]{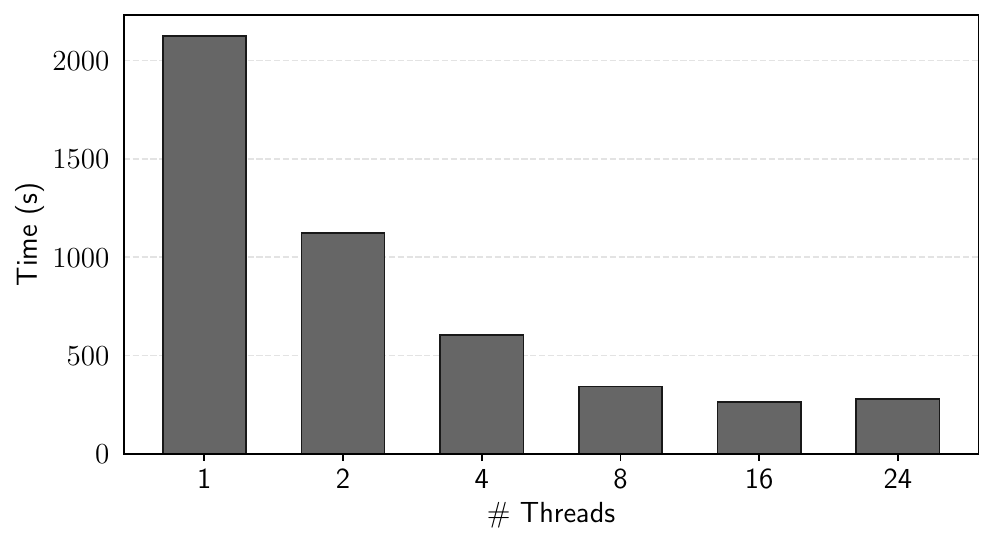}

  \par\vspace{-0.3em}
  \makebox[0.48\textwidth][c]{\small $k=7$}\hfill
  \makebox[0.48\textwidth][c]{\small $k=8$}

  \caption{Running time of \hypermotivo's build-up phase on the SE dataset for $k \in \{3,\dots,8\}$.}
  \label{fig:threads_SE}
\end{figure*}

\clearpage
\begin{figure*}[p]
  \centering

  \includegraphics[width=0.48\textwidth]{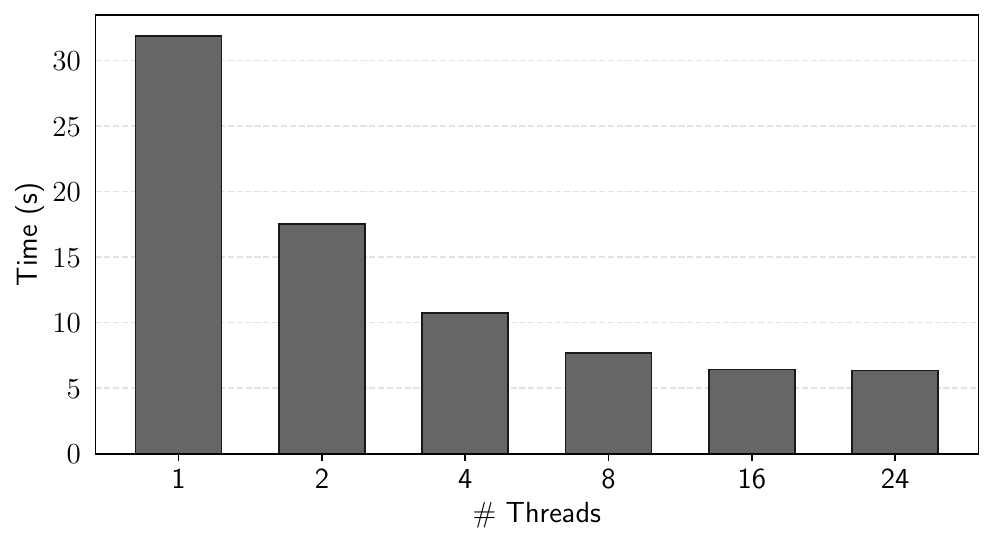}\hfill
  \includegraphics[width=0.48\textwidth]{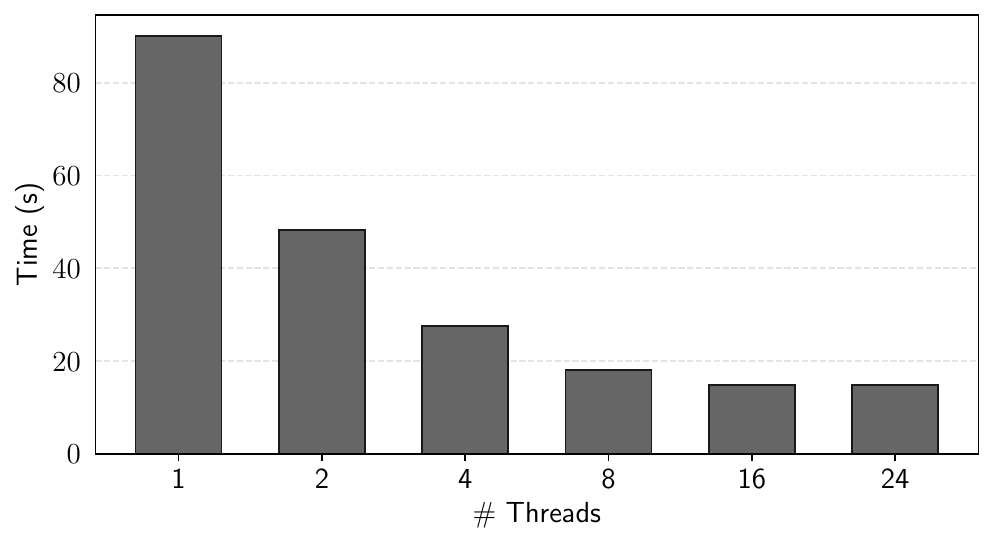}

  \par\vspace{-0.3em}
  \makebox[0.48\textwidth][c]{\small $k=3$}\hfill
  \makebox[0.48\textwidth][c]{\small $k=4$}

  \vspace{1.0em}

  \includegraphics[width=0.48\textwidth]{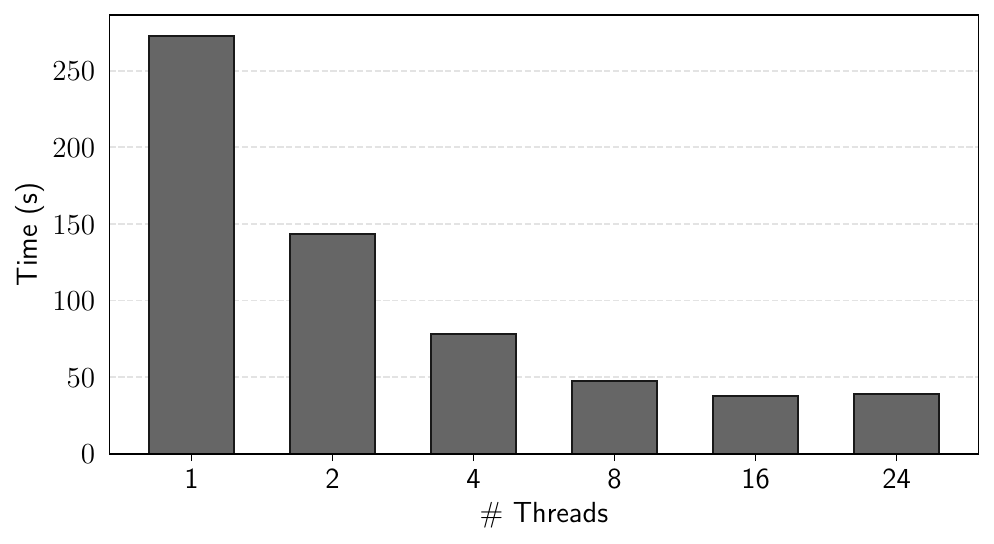}\hfill
  \includegraphics[width=0.48\textwidth]{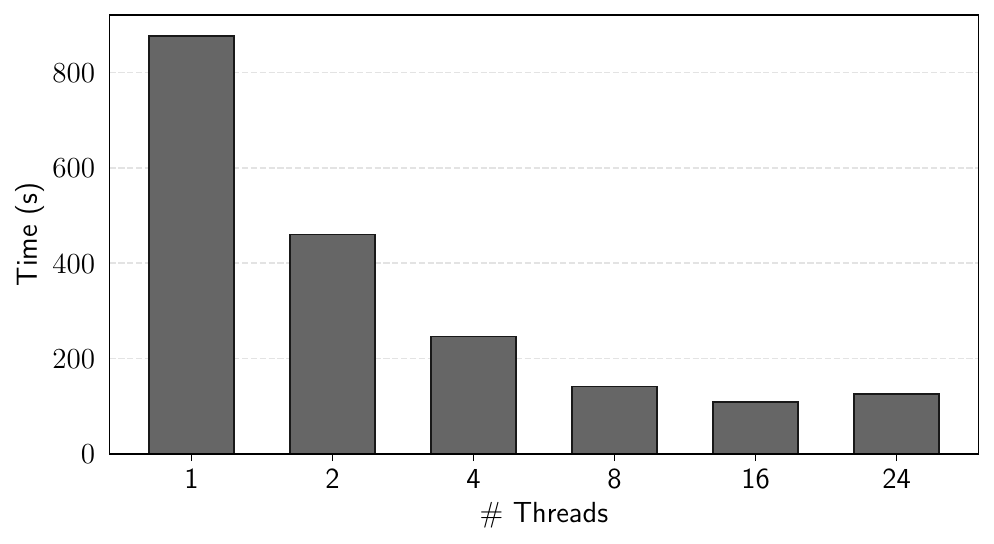}

  \par\vspace{-0.3em}
  \makebox[0.48\textwidth][c]{\small $k=5$}\hfill
  \makebox[0.48\textwidth][c]{\small $k=6$}

  \vspace{1.0em}

  \includegraphics[width=0.48\textwidth]{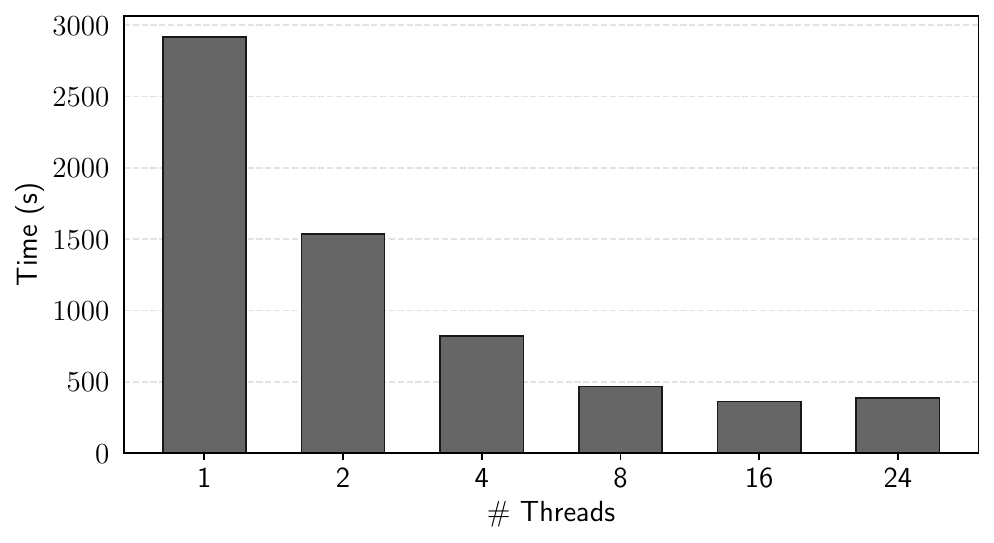}\hfill
  \includegraphics[width=0.48\textwidth]{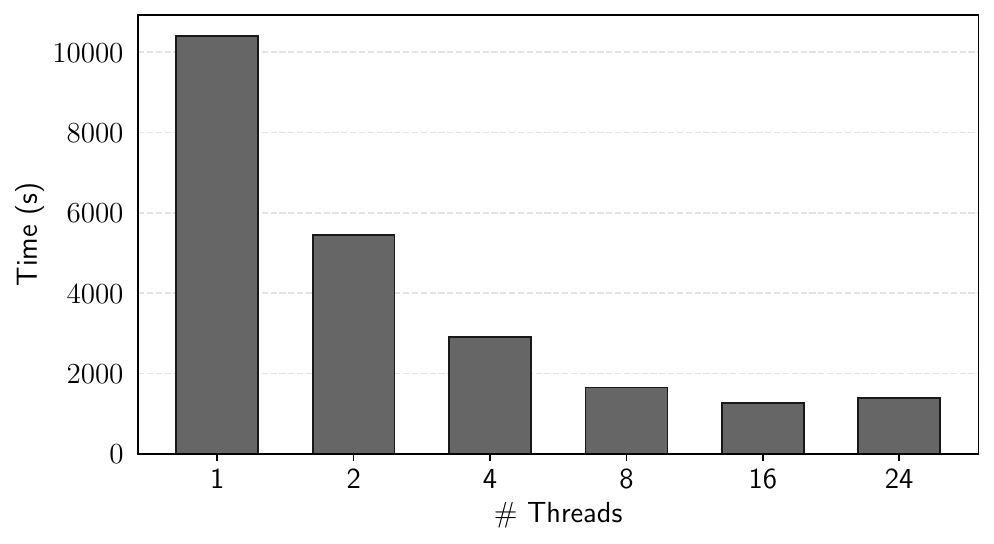}

  \par\vspace{-0.3em}
  \makebox[0.48\textwidth][c]{\small $k=7$}\hfill
  \makebox[0.48\textwidth][c]{\small $k=8$}

  \caption{Running time of \hypermotivo's build-up phase on the AX dataset for $k \in \{3,\dots,8\}$.}
  \label{fig:threads_AX}
\end{figure*}
\clearpage

\subsection{Accuracy}

We empirically assess the accuracy of \hypermotivo{} on small synthetic hypergraphs, following the same protocol as in the main paper. For each hypergraphlet isomorphism type $H$, we report the relative count error $\err_H = (\hat c_H - c_H)/c_H$, where $c_H$ is the exact induced count of $H$ in the input hypergraph and $\hat c_H$ is the estimate produced by \hypermotivo{} using $10^5$ samples. \Cref{fig:error_count} summarizes the resulting error distributions for $k\in\{3,4,5\}$ via histograms with bins of the form $[\err-0.25,\err+0.25)$.

For $k=5$, the plot is restricted to the 50 most frequent hypergraphlet types by exact frequency, to avoid that extremely rare types, whose estimates are inherently high-variance under sampling, dominate the tails and obscure the typical behavior. For $k=3$ and $k=4$, instead, we include all observed types, at most 6 and 20, respectively. Despite the smaller number of types, their frequency distribution is still highly skewed: a few motifs account for the vast majority of occurrences, while the remaining ones are very rare. Consequently, even for $k=3$, the overall error distribution is affected by these low-frequency types, which are more likely to be missed or to exhibit larger relative errors.

\begin{figure}
    \centering
     \begin{subfigure}[t]{0.49\textwidth}
        \centering
        \includegraphics[width=\linewidth]{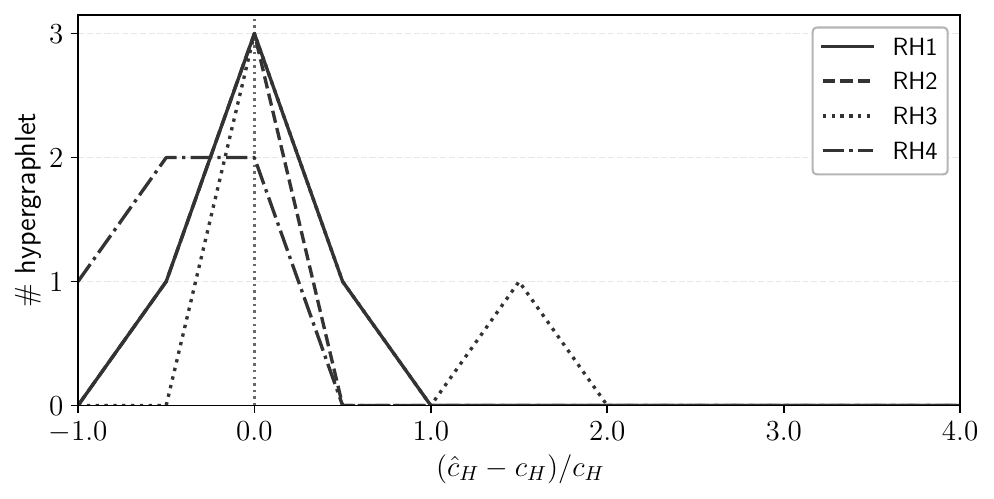}
        \caption{k=3}
        \label{fig:error_count_k3}
      \end{subfigure}
      \begin{subfigure}[t]{0.49\textwidth}
        \centering
        \includegraphics[width=\linewidth]{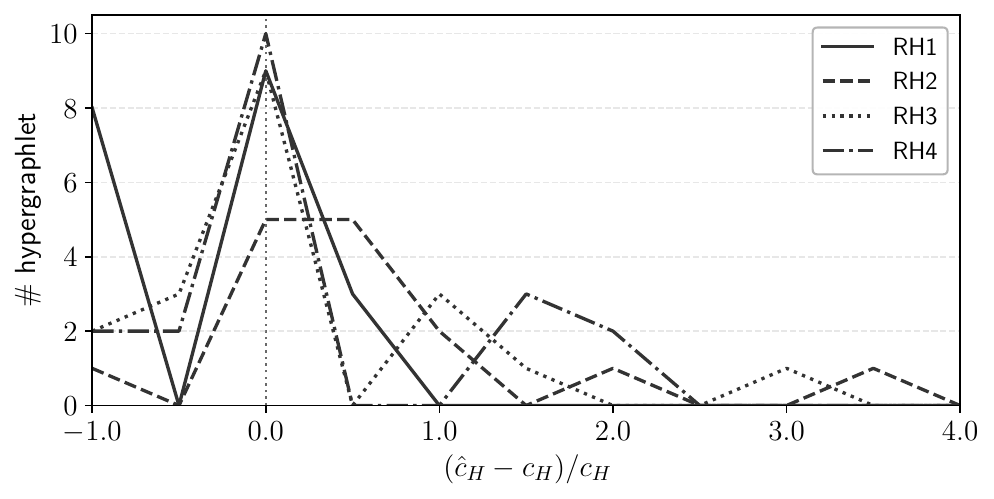}
        \caption{k=4}
        \label{fig:error_count_k3}
      \end{subfigure}
      \begin{subfigure}[t]{0.49\textwidth}
        \centering
        \includegraphics[width=\linewidth]{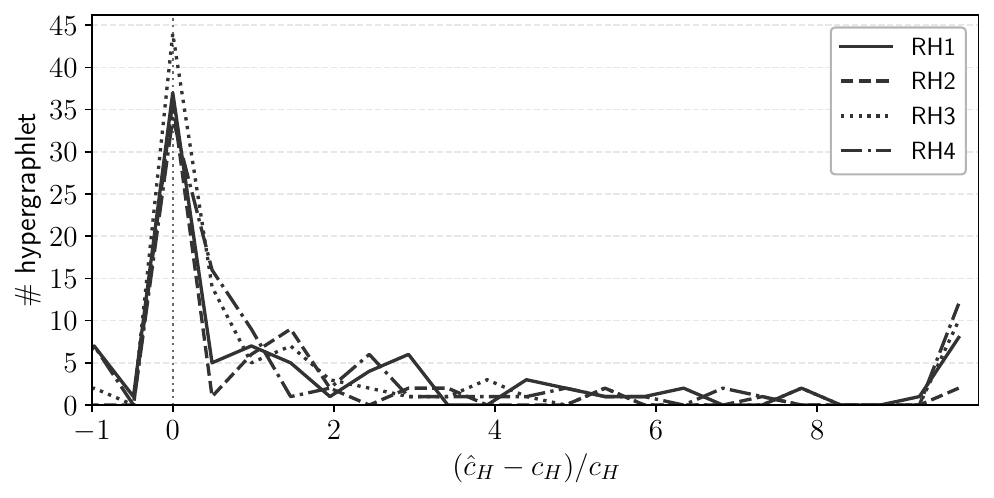}
        \caption{k=5}
        \label{fig:error_count_k3}
      \end{subfigure}
    \caption{$err_H$ distributions for $k \in \{3,4,5\}$ and $10^5$ samples on four synthetic datasets.}
    \label{fig:error_count}
\end{figure}

%% file: figures/frequencies/k4_table.tex
\begin{tabular}{lccccccc}
\hline
& AR & AX & GP & SE & MA & RH & SA \\
\hline
A & 0.0175 & 0.3605 & 0.1136 & 0.3281 & 0.3258 & 0.0692 & 0.2936 \\
B & 0.3649 & 0.1506 & 0.3564 & 0.1851 & 0.2424 & 0.0261 & 0.0461 \\
C & 0.5316 & 0.0450 & 0.2009 & 0.0592 & 0.3043 & 0.0073 & 0.0199 \\
D & 0.0004 & 0.0678 & 0.0107 & 0.0808 & 0.0814 & 0.0323 & 0.6324 \\
E & 0.0105 & 0.0806 & 0.0804 & 0.0790 & 0.0214 & 0.0575 & 0.0010 \\
F & 0.0002 & 0.0672 & 0.0127 & 0.0649 & 0.0042 & 0.1013 & 0.0050 \\
G & 0.0025 & 0.0515 & 0.0360 & 0.0461 & 0.0037 & 0.0287 & 0.0008 \\
H & 0.0543 & 0.0054 & 0.0487 & 0.0148 & 0.0082 & 0.0100 & 0.0001 \\
I & 0.0001 & 0.0515 & 0.0036 & 0.0280 & 0.0023 & 0.0300 & 0.0009 \\
J & 0.0001 & 0.0127 & 0.0060 & 0.0131 & 0.0002 & 0.0817 & 0.0000 \\
K & 0.0008 & 0.0230 & 0.0094 & 0.0288 & 0.0023 & 0.0298 & 0.0002 \\
L & 0.0001 & 0.0223 & 0.0052 & 0.0131 & 0.0002 & 0.0417 & 0.0000 \\
M & 0.0000 & 0.0173 & 0.0011 & 0.0061 & 0.0000 & 0.0442 & 0.0000 \\
N & 0.0029 & 0.0053 & 0.0297 & 0.0086 & 0.0006 & 0.0192 & 0.0000 \\
O & 0.0000 & 0.0049 & 0.0018 & 0.0029 & 0.0000 & 0.0426 & 0.0000 \\
P & 0.0000 & 0.0045 & 0.0010 & 0.0019 & 0.0000 & 0.0428 & 0.0000 \\
Q & 0.0029 & 0.0027 & 0.0204 & 0.0031 & 0.0010 & 0.0087 & 0.0000 \\
R & 0.0004 & 0.0025 & 0.0061 & 0.0064 & 0.0002 & 0.0202 & 0.0000 \\
S & 0.0000 & 0.0051 & 0.0015 & 0.0041 & 0.0000 & 0.0211 & 0.0000 \\
T & 0.0000 & 0.0008 & 0.0016 & 0.0012 & 0.0000 & 0.0266 & 0.0000 \\
\hline
\end{tabular}

%% file: figures/frequencies/k5_table.tex
\begin{tabular}{lccccccc}
\hline
& AR & AX & GP & SE & MA & RH & SA \\
\hline
A & 0.3549 & 0.0261 & 0.2148 & 0.0455 & 0.2222 & 0.0015 & 0.0109 \\
B & 0.3941 & 0.0030 & 0.0615 & 0.0040 & 0.2310 & 0.0000 & 0.0015 \\
C & 0.0004 & 0.1070 & 0.0166 & 0.1118 & 0.0667 & 0.0139 & 0.2747 \\
D & 0.0000 & 0.0144 & 0.0013 & 0.0197 & 0.0134 & 0.0057 & 0.5330 \\
E & 0.0908 & 0.0259 & 0.1360 & 0.0533 & 0.0584 & 0.0023 & 0.0095 \\
F & 0.0042 & 0.1055 & 0.0430 & 0.1071 & 0.0715 & 0.0045 & 0.0389 \\
G & 0.0052 & 0.0884 & 0.0425 & 0.0888 & 0.0872 & 0.0048 & 0.0430 \\
H & 0.0080 & 0.0447 & 0.0337 & 0.0465 & 0.1218 & 0.0025 & 0.0206 \\
I & 0.0001 & 0.0640 & 0.0029 & 0.0366 & 0.0303 & 0.0028 & 0.0541 \\
J & 0.0001 & 0.0418 & 0.0083 & 0.0402 & 0.0016 & 0.0159 & 0.0010 \\
K & 0.0003 & 0.0315 & 0.0093 & 0.0328 & 0.0043 & 0.0158 & 0.0016 \\
L & 0.0000 & 0.0419 & 0.0010 & 0.0210 & 0.0013 & 0.0131 & 0.0021 \\
M & 0.0000 & 0.0207 & 0.0019 & 0.0241 & 0.0010 & 0.0244 & 0.0069 \\
N & 0.0022 & 0.0185 & 0.0274 & 0.0205 & 0.0042 & 0.0033 & 0.0000 \\
O & 0.0074 & 0.0060 & 0.0276 & 0.0078 & 0.0239 & 0.0013 & 0.0002 \\
P & 0.0452 & 0.0004 & 0.0170 & 0.0012 & 0.0097 & 0.0000 & 0.0000 \\
Q & 0.0026 & 0.0110 & 0.0274 & 0.0198 & 0.0055 & 0.0045 & 0.0002 \\
R & 0.0305 & 0.0004 & 0.0181 & 0.0028 & 0.0062 & 0.0000 & 0.0000 \\
S & 0.0001 & 0.0266 & 0.0041 & 0.0172 & 0.0028 & 0.0032 & 0.0004 \\
T & 0.0000 & 0.0294 & 0.0026 & 0.0145 & 0.0011 & 0.0041 & 0.0004 \\
\hline
\end{tabular}